\documentclass[11pt]{article}
\usepackage[titletoc,title]{appendix}
\usepackage{amsmath,amsfonts,amssymb,bm,amsthm}
\usepackage[colorlinks,citecolor=blue,linkcolor=blue,pagebackref]{hyperref}
\usepackage{tikz}
\usetikzlibrary{shapes,arrows}
\usetikzlibrary{intersections,shapes.arrows}
\usetikzlibrary{decorations.pathreplacing}
\usetikzlibrary{arrows.meta}
\usepackage{xcolor}
\definecolor{bostonred}{rgb}{0.8, 0.0, 0.0}
\usepackage{hyperref}
\hypersetup{
    colorlinks=true,
    linkcolor=bostonred,
    urlcolor=bostonred,
    citecolor=bostonred
}
\numberwithin{equation}{section}
\newtheorem{theorem}{Theorem}[section]
\newtheorem{lemma}[theorem]{Lemma}
\newtheorem{proposition}[theorem]{Proposition}

\theoremstyle{definition}
\newtheorem{example}[theorem]{Example}
\newtheorem{definition}[theorem]{Definition}

\theoremstyle{remark}
\newtheorem{remark}[theorem]{Remark}
\newtheorem{notation}[theorem]{Notation}

\DeclareFontFamily{U}{mathx}{}
\DeclareFontShape{U}{mathx}{m}{n}{<-> mathx10}{}
\DeclareSymbolFont{mathx}{U}{mathx}{m}{n}
\DeclareMathAccent{\widecheck}{0}{mathx}{"71}

\usepackage{multicol}
\usepackage{longtable}
\usepackage{slashed}
\usepackage{accents}
\usepackage[T1]{fontenc}
\usepackage{pifont}
\usepackage{tikz}
\usepackage{simpler-wick}
\usetikzlibrary{shapes,arrows}
\usetikzlibrary{intersections,shapes.arrows}
\usetikzlibrary{decorations.pathreplacing}
\usetikzlibrary{arrows.meta}
\usepackage{xcolor}

\makeatletter
\newcommand{\rc}{%
  \begingroup
  \tikzset{every path/.append style={draw=red}}%
  \aftergroup\endgroup
  \c
}
\makeatother

\usepackage[top=1in, bottom=1in, left=1in, right=1in]{geometry}
\usepackage{fancyhdr}
\definecolor{amber}{rgb}{1.0, 0.49, 0.0}

\def\mcM{{\mathcal{M}}}

\def\bPsi{{\bar{\Psi}}}
\usepackage{enumerate}
\usepackage{xcolor,cancel}
\usepackage{relsize}
\usepackage{slashed}

\renewcommand{\tilde}{\widetilde}

\usepackage{orcidlink}

\allowdisplaybreaks

\usepackage{pifont}
\renewcommand*{\backrefalt}[4]{%
\ifcase #1 %
No citations%
\or
\ding{43}~p.~#2%
\else
\ding{43}~pp.~#2%
\fi}
\usepackage{soul}

\begin{document}

\title{A perturbative approach to the Wetterich equation for Bosonic and Fermionic interacting fields
}
\author{
Beatrice Costeri\,\orcidlink{0009-0004-6594-5926}\thanks{
Dipartimento di Fisica ``Alessandro Volta'',
Universit\`a degli Studi di Pavia \& INFN, Sezione di Pavia \& INdAM, Sezione di Pavia, 
Via Bassi 6, 
I-27100 Pavia, 
Italia;
beatrice.costeri01@universitadipavia.it.
}
}


\date{}

\maketitle

\vspace{-.6cm}

\begin{abstract}
We study the Lorentzian Wetterich Renormalization Group (RG) flow equation for interacting quantum fields on curved backgrounds within the framework of perturbative Algebraic Quantum Field Theory (pAQFT). Specifically, we consider two classes of models: \textbf{two mutually interacting scalar fields} on globally hyperbolic spacetimes without boundary and, under the further assumption that the underlying background is spin, \textbf{self-interacting Dirac fields}. In both cases, we derive the corresponding RG flow equations within a Local Potential Approximation and compute the beta functions for the relevant couplings. For the scalar model, we also discuss an asymmetric interaction potential which is formally reminiscent of the Martin-Siggia-Rose description of a stochastic dynamics, thereby indicating a possible connection between Lorentzian algebraic RG methods \cite{DAngeloDragoPinamontiRejzner2024} and stochastic field-theoretic models, \cite{Duch_2025}. In addition, we address the local well-posedness of the resulting flow equations. Adapting the strategy detailed in \cite{Dangelo_23} and based on the the Nash-Moser theorem, we prove local existence and uniqueness of solutions for both the scalar and the Dirac models.

\


\

{\bf 2020 MSC classes:} 81T05 

\end{abstract}

\tableofcontents

\allowdisplaybreaks

\section{Introduction}\label{Sec: Introduction} 

In this paper, we discuss the local well-posedness of Lorentzian Wetterich-type flow equations for interacting quantum fields on  $d-$dimensional globally hyperbolic spacetimes with empty boundary. More precisely, we consider two classes of models: two mutually interacting scalar fields and self-interacting Dirac fields.

The natural framework for addressing these issues is \textbf{perturbative Algebraic Quantum Field Theory (pAQFT)}, which provides a local and covariant construction of quantum field theories on Lorentzian backgrounds, both in the free and in the interacting case, \textit{cf.} \cite{Brunetti_2009, Rejzner}. In this setting, one starts from a suitable algebra of classical observables, modelled as functionals on the configuration space of the theory, and then deforms its pointwise product into a non-commutative quantum product. For Bosonic fields, the deformation is governed by the causal propagator, namely, the retarded-minus-advanced fundamental solution of the operator ruling the dynamics, and it encodes the canonical commutation relations [CCR]. For Fermionic fields, the construction is performed in the corresponding graded setting:  the deformation is induced by the causal fundamental solution of the Dirac operator, together with the spinorial pairing, and it encodes the canonical anticommutation relations [CAR]. States are then introduced as positive, normalized linear functionals on the quantum algebra. In applications, one usually restricts attention to quasi-free Hadamard states, whose two-point functions have a prescribed microlocal singular structure. This condition is essential in order to define physically relevant observables, such as the stress-energy tensor \cite{Moretti1997, Moretti2003, CosteriDappiaggiGoi2026}, both in the free theory and in interacting models. Interactions are incorporated perturbatively through the Bogoliubov map, which associates to a classical observable its interacting quantum counterpart as a formal power series in the coupling constants. This perturbative framework is particularly natural from the viewpoint of physical applications on Lorentzian spacetimes: it preserves locality and covariance, and, at the same time, it allows for an algorithmic control of the renormalization freedoms of the underlying theory. These features make pAQFT especially well suited for deriving Lorentzian Renormalization Group flow equations, where the r\^ole of the inverse propagator entering the Wetterich equation must be replaced by state-dependent, microlocally well-behaved propagators.

\textbf{Renormalization Group (RG) flow equations} capture in a systematic way how the effective description of a quantum field theory changes with respect to a scale parameter. In Wilson's approach to the Renormalization Group \cite{WilsonKogut1974, Freidel:2025uws}, this flow encodes the progressive transition from microscopic degrees of freedom to macroscopic observables, yielding a coarse-grained picture of the system at low energies. In the functional formulation, the RG flow is implemented at the level of a scale-dependent effective average action $\Gamma_k$, \textit{cf}., \cite{Dupuis_2021, BergesTetradisWetterich2002}. The latter interpolates between the microscopic action at the ultraviolet scale and the full quantum effective action in the infrared regime, and its evolution is governed, in the Euclidean setting, by the \textbf{Wetterich equation} \cite{Wetterich1991, Wetterich1993}. Although the Wetterich equation has the algebraic form of a one-loop trace, the propagator entering its right-hand side is the full scale-dependent propagator, namely the inverse of $\Gamma_k^{(2)}+Q_k$, where $Q_k$ is an infrared regulator, rather than a fixed free propagator. The equation is therefore non-linear in the average effective action $\Gamma_k$ and encodes non-perturbative information. This feature makes the Functional Renormalization Group particularly suited to approximation schemes such as the derivative expansion and the \textbf{Local Potential Approximation} \cite{Dangelo_23, DAngeloDragoPinamontiRejzner2024}, while at the same time allowing for the dynamical generation of interaction terms along the flow \cite{Morris1994,BergesTetradisWetterich2002}. In Lorentzian signature, however, the formulation of RG flow equations requires additional care, since the hyperbolic operator admits canonical retarded and advanced Green's operators, but no distinguished Feynman inverse or vacuum two-point function. Furthermore, the infrared regulator $Q_k$ has to be compatible with covariance, causality and the microlocal structure of quantum states. 

In this paper, following the algebraic approach to quantum field theory \cite{Brunetti_2009} and the Lorentzian formulation of Wetterich-type equations developed in \cite{Dangelo_23, DAngeloDragoPinamontiRejzner2024}, we derive RG flow equations for two classes of interacting theories: two mutually interacting scalar fields on globally hyperbolic, Lorentzian spacetimes $(\mcM, g_{\mcM})$ without boundary and -- under additional requirements on the underlying spin structure of $\mcM$ -- self-interacting Dirac fields. In both scenarios, the resulting flow governs the scale dependence of the corresponding effective interaction, and, within the chosen approximation scheme, yields the beta functions for the relevant couplings. The beta functions obtained in this way are consistent with recent results in the physics literature \cite{JuarezWysozka_2021, BorgesShapiro_2025}.
Beyond the derivation of the beta functions, we also address the local well-posedness of the resulting RG flow equations: adapting the strategy of \cite{Dangelo_23}, we prove the existence and uniqueness of local solutions in both cases by means of the Nash-Moser theorem, \textit{cf.}, \cite{Moser1966a, Moser1966b, Nash1956, Hamilton_1982}.

As an application of the results derived in this work, we succinctly discuss a model of two mutually interacting scalar fields via a potential asymmetric in the fields. At the formal level, this model closely resembles the field-theoretic description of stochastic dynamics obtained through the Martin-Siggia-Rose formalism \cite{Bonicelli_2025}, although it is here considered in a Lorentzian setting. This passage is delicate. Indeed, stochastic equations driven by additive Gaussian white noise are naturally formulated in Euclidean or parabolic frameworks, where the relevant elliptic or parabolic operator selects, under suitable assumptions, a distinguished Green kernel. By contrast, in Lorentzian signature the underlying wave operator admits no unique inverse: retarded, advanced and Feynman-type inverses encode different causal or state-dependent prescriptions. This lack of a preferred inverse is precisely one of the structural reasons why the Lorentzian formulation of RG flow equations requires the tools of perturbative algebraic quantum field theory, where the choice of a Hadamard state and the associated microlocal control replace the Euclidean momentum-space construction. Conversely, even in the Euclidean algebraic approach, the construction of Wick powers, Euclidean products and interacting observables is not automatic, but relies on suitable parametrices and on the control of the associated renormalization ambiguities, as discussed in \cite{Dappiaggi_2020}. 

Singular Euclidean stochastic scalar field theories have recently been analysed via complementary methods, notably paracontrolled calculus and related flow-equation techniques \cite{GubinelliImkellerPerkowski2015,CatellierChouk2018,GubinelliKochOh2024,Duch_2025,DuchGubinelliRinaldi2023}. In particular, in the flow-equation approach one introduces scale-dependent effective force coefficients, whose relevant components play a r\^ole analogous to running couplings in the functional RG picture. However, the precise relation between this stochastic, Polchinski-type description and the Wetterich-type equations obtained within the Lorentzian pAQFT framework remains unclear. In this sense, the possible correspondence between perturbative, algebraic and non-perturbative RG descriptions of stochastic field theories is still largely unexplored.

A further natural direction, not pursued here, is the extension of the present analysis to spacetimes with timelike boundaries. In this setting, boundary conditions affect both the causal propagators and the microlocal structure of admissible states. The analysis of the Lorentzian Wetterich equation in the presence of Dirichlet and Neumann boundary conditions in \cite{DappiaggiNavaSinibaldi2025}, together with the construction of Hadamard states for Robin boundary conditions in \cite{CosteriDappiaggiJuarezAubrySingh2026}, provides a natural starting point for understanding the interplay between boundary conditions and RG flow. This would clarify whether the corresponding Wetterich-type equation admits a local and covariant interpretation compatible with the renormalization freedoms already present in the bulk theory.

\paragraph*{Synopsis --} 
The outline of the paper is the following.

In Section~\ref{Sec: pAQFT for multiple interacting fields}, we review key notions of perturbative algebraic quantum field theory (pAQFT), first in the general setting of arbitrary vector bundles and then specialized to the cases under consideration. Section~\ref{Sec: two scalar fields} is devoted to the construction of both the free and the covariant quantum algebra of local observables for two mutually interacting scalar fields on globally hyperbolic spacetimes without boundary. In addition, in Section~\ref{Sec: Wetterich for two scalar fields}, we derive the Renormalization Group flow equations within the Local Potential Approximation framework from the Wetterich equation. As a relevant application of the formalism developed in this section, we consider an asymmetric interaction potential reminiscent of stochastic fields in the Martin-Siggia-Rose formalism, \textit{cf}., Section \ref{Sec: MSR fields}. 
Section~\ref{Sec: Thirring model} is devoted to the analysis of the algebraic framework for self-interacting Dirac fields, where anticommutativity is encoded in the properties of the deformation map. In this setting, we also derive the Wetterich equation and determine the scaling behaviour of the relevant couplings in three dimensions, \textit{cf.}, Section \ref{Sec: Wetterich equation for the Thirring model}.

In Section~\ref{Sec: Existence of local solutions to the Renormalization Group flow equations}, we establish the local existence of solutions to the Wetterich equation via the Nash-Moser Theorem~\ref{Thm: Nash-Moser}, both for two mutually interacting scalar fields, \textit{cf.}, Section \ref{Sec: (A) Two mutually interacting scalar fields} and for Dirac fields, \textit{cf.}, Section \ref{Sec: (B) Self-interacting Dirac fields}.

Appendix~\ref{Sec: Appendix A} collects the necessary background material for the application of the Nash-Moser theorem to the Renormalization Group flow equations, \textit{e.g.}, the notion of tame maps between graded Fréchet spaces, see Definition \ref{Def: tame map}.

\section{A pAQFT approach to flow equations for interacting fields}
\label{Sec: pAQFT for multiple interacting fields}
The aim of this section is to derive the Lorentzian Wetterich-type RG equation for two quantum field theoretic models: two mutually interacting scalar fields and self-interacting Dirac fields, discussed respectively in Sections~\ref{Sec: two scalar fields} and~\ref{Sec: Thirring model}. To this avail, we first fix the geometric setting and recall the general structures underlying perturbative Algebraic Quantum Field Theory (pAQFT), \cite{Brunetti_2009, brunetti2, Rejzner}.

Let $(\mcM, g_\mcM)$ be a $d-$dimensional globally hyperbolic spacetime, $d \ge 2$, with $g$ a Lorentzian metric of signature $(1, d-1)$, \textit{cf}., \cite{HawkingEllis1973}. Consider thereon some finite-rank vector bundle $E \equiv E[\mcM, \pi, V]$, where the fibre $V$ is a finite-dimensional vector space, \textit{i.e.}, $\text{dim} (V) = n < \infty$. We denote by $\mathcal{E}(E) := \Gamma^{\infty}(E)$ the \textbf{configuration space} of the field theory under scrutiny, corresponding to the space $C^\infty(E \leftarrow \mcM)$ of smooth sections of the bundle $E$. 

\begin{example} We now present some illustrative examples of configuration spaces.
    \begin{itemize}
        \item[(i)] Consider $n \in \mathbb{N}$ real or complex scalar fields on $(\mcM, g_\mcM)$. In this case, $E \equiv E[\mcM, \pi, \mathbb{K}^n]$ is the trivial $\mathbb{K}$-vector bundle, where $\mathbb{K}$ denotes either $\mathbb{R}$ or $\mathbb{C}$, over $\mcM$, \textit{i.e.}, $E$ is globally diffeomorphic to $\mcM \times \mathbb{K}^n$. Smooth sections of this bundle $\mathcal{E}(E) := \Gamma^{\infty}(E) \simeq C^{\infty}(\mcM) \otimes \mathbb{K}^n$ correspond to \emph{field multiplets}, where $\otimes$ denotes the algebraic tensor product. In particular, for a real scalar field, $E \simeq \mcM \times \mathbb{R}$. In this case, the kinematical field configurations are identified with smooth, real-valued functions on $\mcM$, that is, $\mathcal{E}(\mcM) := \Gamma^{\infty}(\mcM \times \mathbb{R}) \simeq C^{\infty}(\mcM, \mathbb{R})$. The case of two mutually interacting scalar fields is discussed in Section~\ref{Sec: two scalar fields}. 
        \item[(ii)] Under the additional assumptions that $\mcM$ is spin with a trivial spin bundle $S\mcM$, the r\^ole of $E$ is played the Dirac bundle $D\mcM = S\mcM \times_{\sigma_d} \Sigma_d$, where $\Sigma_d \simeq \mathbb{C}^{N_d}$ is the \emph{spinor space}, whilst $\sigma_d: \text{Spin}(d) \rightarrow GL(\Sigma_d)$ is the \emph{spinor representation} -- see . Spinor fields are smooth sections of the Dirac bundle, \textit{i.e.}, $\Gamma^{\infty}(D\mcM) \simeq C^{\infty}(\mcM, \Sigma_d)$ -- see Section \ref{Sec: Thirring model} for additional details. 
        \item[(iii)] Let $\mathrm{G}$ be some compact Lie group and denote by $(\mathfrak{g}, [\cdot, \cdot]_\mathfrak{g})$ the corresponding Lie algebra. For Yang-Mills theories with a trivial principal bundle $\mathrm{P} \simeq \mathcal{M} \times \mathrm{G}$, the adjoint bundle $ \text{Ad}(\mathrm{P}) \simeq \mcM \times \mathfrak{g}$ plays the r\^ole of the vector bundle $E$. In this setting, the space of field configurations coincides with the space of $\mathfrak{g}-$valued $1-$forms on $\mcM$, \textit{i.e.}, $\mathcal{E}(\text{Ad}(\mathrm{P})) := \Omega^1 (\mcM, \text{Ad}(\mathrm{P})) \simeq \Omega^1 (\mcM, \mathfrak{g})$  -- see \cite[Ch. 7]{Rejzner}. 
    \end{itemize}
\end{example}

The configuration space $\mathcal{E}(E)$ can be equipped with a Fréchet topology induced by a suitable family of seminorms. Since $\mathcal{M}$ is a smooth, second-countable Hausdorff manifold, it is paracompact and, hence, it admits a Riemannian metric $g_R$. On the vector bundle $E$, we fix a fibre metric $h$ and a compatible connection $\nabla^E$. Let $\{K_m\}_{m \in \mathbb{N}}$ be an exhaustion of $\mathcal{M}$ by compact sets. For $j,m \in \mathbb{N}$, define
\begin{equation}
\label{Eq: seminorms on the configuration space}
    p_{j,m}(s) := \sup_{x \in K_m} \|\nabla^j s(x)\|_{g_R,h}, \quad s \in \Gamma^\infty(E),
\end{equation}
where the norm is induced by $g_R$ and $h$. The seminorms $\{p_{j,m}\}_{j,m}$ generate a locally convex topology $\tau$ on $\mathcal{E}(E)$, with respect to which $\Gamma^\infty(E)$ is a nuclear Fréchet space, see, \textit{e.g.}, \cite[Ch.~10]{Treves_1967}, in particular Example I for the analysis of the trivial bundle case.

In this context, physically relevant observables can be modelled -- both at a classical and quantum level -- as smooth functionals from the configuration space $\mathcal{E}(E)$ of the underlying theory to the complex field $\mathbb{C}$. The reason for considering complex-valued objects, as opposed to real-valued ones, lies in the fact that the quantization procedure, even for real field theories, inherently involves the introduction of a complex structure. 
Henceforth, we shall confine our attention to those complex-valued functionals $F: \mathcal{E}(E) \rightarrow \mathbb{C}$, which are smooth in the Bastiani sense, \textit{cf.}, \cite{Bastiani_1964}, \cite{Hamilton_1982}, \cite{Milnor_1984}, \cite{Nee_2006}. In particular, to define the functional derivatives of such $F$, we shall adopt the Michal-Bastiani notion of differentiability, see \cite{Michal_1938} and \cite{Michor_1984}. Following slavishly \cite[Ch. 3]{Rejzner}, we denote by \begin{equation}
    \label{Eq: space of functionals on E}
    \mathcal{F}(E) := \{F: \mathcal{E}(E) \to \mathbb{C} \, | \, F \, \text{is Bastiani smooth in the sense of \cite[Def. 3.4]{Rejzner}} \},
\end{equation}
the \textbf{space of functionals} of the underlying field theory.

First and foremost, we need to equip the space $\mathcal{F}(E)$ with a suitable notion of \emph{localisation} of functionals on the spacetime manifold $\mathcal{M}$. The underlying motivation is that physical observables are intrinsically local. Therefore, in the framework of algebraic quantum field theory, one aims to associate algebraic structures to bounded regions of spacetime in a way compatible with causality, \textit{cf.}, \cite{brunetti2}, \cite{Haag_1996}, \cite{Araki_1999}. In this regard, it is essential to distinguish functionals according to the region of spacetime on which they effectively depend. This allows one to determine whether a given observable is localized in a prescribed region and, consequently, to organize $\mathcal{F}(E)$ into a net of local algebras $\mathcal{O} \mapsto \mathcal{A}(\mathcal{O})$ indexed by suitable spacetime domains $\mathcal{O} \subset \mathcal{M}$. From a functional analytic viewpoint, the localisation properties of functionals also play a crucial r\^ole in ensuring that algebraic operations remain compatible with the microlocal regularity conditions required for their well-definiteness, see \cite{Hormander_1990},\cite{Brunetti_1996}, \cite{Rejzner}. To this end, we introduce the notion of \emph{spacetime support} of functionals. 

\begin{definition}
\label{Def: spacetime support}
    Let $F \in \mathcal{F}(E)$ be as per Equation \eqref{Eq: space of functionals on E}. We call \emph{spacetime support} of $F$ the set 
    \begin{equation}
        \label{Eq: spacetime support}
        \text{supp}(F) := \{ x \in \mcM \, | \, \forall \mathcal{O}_x, \exists \eta, \theta \in \mathcal{E}(E), \text{supp}(\theta) \subset \mathcal{O}_x \, \text{such that} \, F(\eta + \theta) \ne F(\eta) \}, 
    \end{equation}
where $\mathcal{O}_x \subset \mcM$ is a neighbourhood of $x$ in $\mcM$. 
\end{definition}

\begin{remark} 
Notice that the above definition admits a natural generalisation to Fréchet derivatives of functionals, which can be regarded as elements of $ \mathcal{D}'(\mcM^k, (E^*)^{\boxtimes k})$ for any $k \in \mathbb{N}$, where $E^*$ denotes the dual bundle to $E$ -- see also Definition \ref{Def: microcausal functionals}. 
\end{remark}

To avoid unnecessary complications, from now on we abandon the general bundle setting and focus instead on the concrete models considered in this work: two mutually interacting scalar fields, see Section~\ref{Sec: two scalar fields}, and self-interacting Dirac fields, \textit{cf.}, Section~\ref{Sec: Thirring model}. This restriction will be in place for both the construction of the free and interacting algebras of observables. A more thorough treatment of the general framework can be found in \cite{Rejzner}. The same distinction will be retained in the derivation of the corresponding Lorentzian Wetterich-type equations in Sections \ref{Sec: Wetterich for two scalar fields} and \ref{Sec: Wetterich equation for the Thirring model}, in the spirit of the construction developed in \cite{Dangelo_23}. 

\subsection{(A) Two mutually interacting scalar fields}
\label{Sec: two scalar fields}
On top of $(\mcM, g_\mcM)$, consider two real, mutually interacting, scalar fields $\varphi_1, \varphi_2 \in C^{\infty}(\mcM; \mathbb{R})$. The full action of the theory is thus specified by
\begin{equation}
\label{Eq: action 2 scalar fields}
    I[\varphi_1, \varphi_2] := - \int_{\mcM} d\mu_x \, \left(\frac{1}{2}\left[\sum_{\ell = 1,2} g^{\mu \nu}(x) \partial_{\mu} \varphi_{\ell} (x) \partial_{\nu} \varphi_{\ell}(x) + (m_\ell^2 + \xi_{\ell}R) \varphi_{\ell}^2(x) \right] + V(\varphi_1, \varphi_2)(x) \right) \, f(x), 
\end{equation}
where $d\mu_x$ denotes the metric-induced volume measure on $\mcM$ and $f \in C^{\infty}_0(\mcM)$. In addition, we assume that the two fields have different masses, namely $m_1 \ne m_2$. In Equation \eqref{Eq: action 2 scalar fields}, the functional $V(\varphi_1, \varphi_2)$ is assumed to have a polynomial dependence in the fields. In particular, we shall consider a symmetric interacting potential of the form 
\begin{equation}
    \label{Eq: potential 2 scalar field}
    V_h(\varphi_1, \varphi_2) := \int_{\mcM} d\mu_x \, \left[ \frac{\lambda_1}{4!} \varphi^4_1(x) + \frac{\lambda_2}{4!} \varphi^4_2(x) + \frac{\lambda_3}{4} \varphi_1^2(x) \varphi_2^2(x)\right] h(x), \quad \lambda_1, \lambda_2 > 0, \, \lambda_1 \lambda_2 > 9 \lambda_3^2. 
\end{equation}
where $h \in C^{\infty}_0(\mcM)$ is an adiabatic cut-off needed to avoid infrared divergences. The assumptions $\lambda_1,\lambda_2>0$ and 
$\lambda_1\lambda_2>9\lambda_3^2$ ensure that the quartic part of the 
potential in Equation~\eqref{Eq: potential 2 scalar field} is bounded from 
below. Indeed, setting $\rho_1 := \frac{1}{2} \varphi_1^2$ and 
$\rho_2 := \frac{1}{2} \varphi_2^2$, and disregarding the smooth cut-off, the integrand 
can be written as
\begin{equation*}
    V(\rho_1,\rho_2)(x)
    =
    \frac{1}{3!}
    \left(
        \lambda_1 \rho_1^2(x)
        +6\lambda_3 \rho_1(x)\rho_2(x)
        +\lambda_2 \rho_2^2(x)
    \right).
\end{equation*}
The quadratic form in the variables $(\rho_1,\rho_2)$ has an associated matrix 
\begin{equation*}
     \begin{pmatrix}
        \lambda_1 & 3\lambda_3\\
        3\lambda_3 & \lambda_2
    \end{pmatrix}.
\end{equation*}
The stated assumptions are precisely the conditions ensuring that this matrix 
is positive-definite.

\begin{remark}
We stress that the splitting of the total action as the sum of free and interacting contributions in Equation \eqref{Eq: action 2 scalar fields} carries an intrinsic degree of arbitrariness. Indeed, another admissible choice would be to include into the interacting potential also the mass terms, which are quadratic in the fields. Yet, on account of the principle of perturbative agreement -- see \cite{DragoHackPinamonti2017}, such a choice does not affect the final result.   
\end{remark}

\begin{remark}
\label{Rmk: field doublet}
Note that the two scalar fields $\varphi_1, \varphi_2 \in C^{\infty}(\mcM)$ can equivalently be regarded as a single smooth section $\Phi \in \Gamma^{\infty}(\mcM \times \mathbb{R}^2) \simeq C^{\infty}(\mcM; \mathbb{R}^2)$ of the trivial vector bundle $\mcM \times \mathbb{R}^2$. 
\end{remark}

\begin{remark}
\label{Rmk: coupling}
    In Equation \eqref{Eq: action 2 scalar fields} we have implicitly assumed that the two fields are coupled only by the interaction term. Working in full generality, we could also have incorporated into the theory a mass term of the form $m_3^2 \varphi_1 \varphi_2$. Indeed, being such a contribution of zeroth order, it does not affect the Green hyperbolicity of the Klein-Gordon operator. Yet, dealing with this term at the level of the free dynamics yields several technical complications, as it prevents the decoupling of the free equations of motion -- see Equation \eqref{Eq: free dynamics 2 scalar fields}. To bypass this hurdle, we can invoke the \emph{principle of perturbative agreement} \cite{DragoHackPinamonti2017} to include it in the interacting potential. Nonetheless, as will be shown in detail in Section \ref{Sec: Wetterich for two scalar fields}, the presence of such a mixed term precludes the choice of a diagonal infrared regulator and, hence, it further complicates the derivation of the Wetterich equation. 
\end{remark}

As a preliminary step towards the construction of a suitable algebra of observables for the theory described by Equation \eqref{Eq: action 2 scalar fields}, we begin by fixing notations and conventions once and for all. Let $\mathcal{E}(\mcM):= C^{\infty}(\mcM; \mathbb{R})$. We denote by
\begin{equation*}
   \mathcal{E}_2 (\mathcal{M}) := \mathcal{E}(\mathcal{M}) \oplus \mathcal{E}(\mathcal{M}), 
\end{equation*}
the \emph{configuration space} of the theory, which is a nuclear, Fréchet space equipped with the product topology, \textit{i.e.}, the Fréchet topology generated by the family of seminorms:
$$p_{n, K}^{(2)} (\varphi,\varphi') := \max \{p_{n, K} (\varphi) + p_{n, K} (\varphi')\}, \quad \text{for any compact set} \, K \subset \mcM \,  \text{and for} \, \varphi,\varphi' \in \mathcal{E}(\mcM),$$ where $$p_{n, K}(\varphi) := \sup_{x \in K} \max_{m \le n} |\varphi^{(n)}(x)|.$$ 

\noindent We can now introduce the key objects of our analysis. 

\begin{definition}[Space of functionals]
\label{Def: space of functionals 2 scalar fields}
In the aforementioned setting, we call \emph{space of functionals} of the quantum field theory described by Equation \eqref{Eq: action 2 scalar fields}, the space
\begin{equation}
    \label{Eq: MSR space of functionals}
    \mathcal{F}(\mathcal{M}) := \{ F: \mathcal{E}_2 (\mathcal{M}) \rightarrow \mathbb{C} \, | \, F \, \text{is smooth in the Bastiani sense} \, \},
\end{equation}
see \cite{Rejzner}. 
\end{definition}

\begin{remark}
An element $u \in \mathcal{F}(\mcM)$ should actually be regarded a functional 
$$F: \mathcal{D}(\mcM) \times \mathcal{E}_2(\mcM) \rightarrow \mathbb{C}, \quad (f; \varphi_1, \varphi_2) \mapsto F(f; \varphi_1, \varphi_2),$$ which is linear in the first entry and smooth with respect to the locally convex topology of $\mathcal{E}_2(\mcM)$. Henceforth, with a slight abuse of notation, the dependence on the smooth, compactly supported test function will be left implicit and we shall still denote by $F$ the functional $F_f := F(f; \cdot, \cdot)$. 
\end{remark}

The space of functionals $\mathcal{F}(\mcM)$ can be also endowed with a notion of Fréchet derivative as clarified by the following definition. 

\begin{definition}[Fréchet derivatives]
\label{Def: Frechet derivatives}
Let $F \in \mathcal{F}(\mcM)$, where $\mathcal{F}(\mcM)$ is the space of functionals as per Definition \ref{Def: space of functionals 2 scalar fields}. We call \emph{Fréchet derivative} of $F$ of order $(k_1, k_2) \in \mathbb{N}_0 \times \mathbb{N}_0$, the functional-valued distribution $F^{(k_1, k_2)} \in \mathcal{E}'(\mcM^{k_1+k_2}; \mathcal{F}(\mcM))$ defined by
\begin{flalign}
    \label{Eq: Frechet derivatives}
    \notag & F^{(k_1, k_2)}: \underbrace{\mathcal{E}(\mcM) \otimes \ldots \otimes \mathcal{E}(\mcM)}_{k_1 - \text{times}} \otimes \underbrace{\mathcal{E}(\mcM) \otimes \ldots \otimes \mathcal{E}(\mcM)}_{k_2 - \text{times}}  \times \mathcal{E}_2(\mcM) \longrightarrow \mathbb{C}, \\ \notag
    F(\eta_1 \otimes & \ldots \otimes \eta_{k_1} \otimes \eta'_{1} \otimes \ldots \otimes \eta'_{k_2}; \varphi_1, \varphi_2) = \\ &\frac{\partial^{k_1+k_2}}{\partial s_1 \ldots \partial s_{k_1} \partial s'_{1} \ldots \partial s'_{k_2}} F(\varphi_1 + s_1 \eta_1 + \ldots + s_{k_1} \eta_{k_1}; \varphi_2 + s'_1 \eta'_1 + \ldots + s'_{k_2} \eta'_{k_2}) \big \vert_{\substack{s_1 = \ldots = s_{k_1} = 0 \\ s'_1 = \ldots = s'_{k_2} = 0}}. 
\end{flalign}
\end{definition}

\begin{remark}[Polynomial functionals]
\label{Rmk: polynomial functionals}
If a functional $F \in \mathcal{F}(\mcM)$ has only a finite number of non-vanishing Fréchet derivatives, $F$ is said to be of \emph{polynomial} type. The space of polynomial functionals over $\mcM$ will be denoted by $\text{Pol}(\mcM).$
\end{remark}

\begin{example}
\label{Ex: fun 2 scalar fields}
Notable examples of polynomial functionals which will play a crucial r\^ole in the following are: 
\begin{flalign}
    \boldsymbol{1} [f] (\varphi_1, \varphi_2) &:= \int_{\mcM} d\mu_x \, f(x), \\
    \label{Eq: Phi1, Phi2}
    \Phi_1[f] (\varphi_1, \varphi_2) &:= \int_{\mcM} d\mu_x \, \varphi_1(x) f(x), \quad \Phi_2[f] (\varphi_1, \varphi_2) := \int_{\mcM} d\mu_x \, \varphi_2(x) f(x),
\end{flalign}
where $(\varphi_1, \varphi_2) \in \mathcal{E}_2(\mathcal{M})$ and $f \in \mathcal{D}(\mathcal{M})$. A direct computation entails that $$(\Phi_1)^{(1, 0)}[f] (\eta_1; \varphi_1, \varphi_2) = \int_{\mcM} d\mu_x \, \eta_1(x) f(x), \quad (\Phi_2)^{(0, 1)}[f] (\eta_2; \varphi_1, \varphi_2) = \int_{\mcM} d\mu_x \, \eta_2(x) f(x),$$ all higher functional derivatives being identically vanishing. 
Particularly relevant in view of the ensuing discussion are also the functionals quadratic in the fields, \textit{i.e.}, 
\begin{flalign}
\label{Eq: Phi1^2}
(\Phi_1^2) [f] (\varphi_1, \varphi_2) &:= \int_{\mcM} d\mu_x \, \varphi_1^2(x) f(x), \\ \label{Eq: Phi2^2}
(\Phi_2^2) [f] (\varphi_1, \varphi_2) &:= \int_{\mcM} d\mu_x \, \varphi_2^2(x) f(x),\\ 
\label{Eq: Phi1Phi2}
      (\Phi_1 \Phi_2) [f] (\varphi_1, \varphi_2) \equiv (\Phi_2 \Phi_1)[f] (\varphi_1, \varphi_2) &:= \int_{\mcM} d\mu_x \, \varphi_1(x) \varphi_2(x) f(x). 
\end{flalign}
\end{example}

\noindent Observe that we can endow $\mathcal{F}(\mcM)$ with an algebra structure by defining thereon a notion of \emph{pointwise product}: for any $F, G \in \mathcal{F}(\mcM)$ and for any pair of field configurations $(\varphi_1, \varphi_2) \in \mathcal{E}(\mcM)$,
\begin{equation}
    \label{Eq: pointwise prod}
 (F \cdot G) (\varphi_1, \varphi_2) := \iota^*[F(\varphi_1, \varphi_2) \otimes G(\varphi_1, \varphi_2)], 
\end{equation}
where $\iota: \mathcal{E}_2(\mcM) \rightarrow \mathcal{E}_2(\mcM) \times \mathcal{E}_2(\mcM), (\varphi_1, \varphi_2) \mapsto ((\varphi_1, \varphi_2),(\varphi_1, \varphi_2))$ is the diagonal map on the configuration space. 

At the present stage, it would be tempting to consider as the classical $*$-algebra of observables the algebra of functionals on $\mcM$, equipped with the pointwise product as per Equation \eqref{Eq: pointwise prod}. At this level, there is no actual obstruction in making this choice. Yet, recall that in the algebraic formalism, the quantum information carried by the statistical correlations functions is encoded in the classical algebraic framework via a deformation of the algebra product. More precisely, this leads to considering pointwise products among derivatives of functionals lying in $\mathcal{F}(\mcM)$ and suitable distributional kernels with a non-trivial singular structure. On account of \cite[Thm. 8.2.10]{Hormander_1990}, such products are \emph{a priori} not necessarily well-defined. This issue can be avoided by focusing on a notable subclass of functionals possessing the sought-after microlocal behaviour. 

\begin{definition}[Microcausal functionals]
\label{Def: microcausal functionals}
Let $k := k_1 + k_2 \in \mathbb{N}_0$. We call \emph{microcausal functionals} the elements of
\begin{equation}
    \label{Eq: microcausal functionals}
    \mathcal{F}_{\mu c}(\mcM) := \{ F \in \mathcal{F}(\mcM) \, \vert \, \text{WF}(F^{(k_1, k_2)}(\varphi_1, \varphi_2)) \subseteq G_{k}(\mcM), \forall k_1, k_2 \in \mathbb{N}_0\}, 
\end{equation}
where 
\begin{equation*}
    G_{k_1 + k_2} (\mcM) := T^*(\underbrace{\mcM \times \ldots \times \mcM}_{k-\text{times}} \setminus \left( \bigcup_{x \in \mcM} (V^+_x)^k \cup \bigcup_{x \in \mcM} (V^-_x)^k \right)). 
\end{equation*}
In particular, we shall denote the space of polynomial microcausal functionals by $\text{Pol}_{\mu c}(\mcM):=\mathcal{F}_{\mu c} (\mcM) \cap \text{Pol}(\mcM)$. 
\end{definition}

Furthermore, we can identify two notable subclasses of microcausal functionals: 
\begin{itemize}
    \item[\ding{103}] \textit{Regular functionals}: 
    \begin{equation*}
        \label{Eq: regular fun 2 scalar fields}
        \mathcal{F}_{reg}(\mcM) := \{ F \in \mathcal{F}_{\mu c}(\mcM) \, \vert \, F^{(k_1, k_2)}(\varphi_1, \varphi_2) \in \mathcal{D}(\mcM^{k_1 + k_2}), \forall k_1, k_2 \in \mathbb{N}_0 \}.
    \end{equation*}
   In particular, we denote by $\text{Pol}_{reg}(\mcM) := \text{Pol}(\mcM) \cap \mathcal{F}_{reg}(\mcM)$.
    \item[\ding{103}] \textit{Local functionals}: 
       \begin{equation*}
        \label{Eq: local fun 2 scalar fields}
        \mathcal{F}_{loc}(\mcM) := \{ F \in \mathcal{F}_{\mu c}(\mcM) \, \vert \, \text{supp}(F^{(k_1, k_2)}(\varphi_1, \varphi_2)) \in \text{Diag}_{k_1 + k_2}(\mcM), \forall k_1, k_2 \in \mathbb{N}_0 \},
    \end{equation*}
    where $\text{supp}(F) := \{x \in \mcM \, \vert \, \forall U \in \mathcal{N}_x, \exists \eta, \theta \in \mathcal{E}_2(\mcM): \text{supp}(\eta) \subset U, F(\eta + \theta) \ne F(\theta)\}$, $\mathcal{N}_x$ being a collection of open neighbourhoods of $x$, is the spacetime support of $F$, while $\text{Diag}_k(\mcM) :=\{ (x, \ldots, x) \in \mcM^k \, \vert \, x \in \mcM\}$ is the total diagonal of $\mcM^k$. The space of local functionals which are in addition polynomials in the fields is denoted by $\text{Pol}_{reg}(\mcM) := \text{Pol}(\mcM) \cap \mathcal{F}_{loc}(\mcM).$
\end{itemize}

\begin{example}
Notable examples of polynomial local functionals are $\Phi_{\ell}^2[f]$ for $\ell = 1,2$ and the pointwise product $(\Phi_1 \Phi_2)[f]$, as one can check by direct inspection -- see Example \ref{Ex: fun 2 scalar fields}. In general, all physically relevant observables -- such as the stress-energy tensor of the underlying theory -- belong to this class. 
\end{example}

We now possess all the ingredients to characterise the classical algebra of observables. Heuristically speaking, since the interaction term is per assumption polynomial in the fields, such an algebra should contain any linear combinations of functionals with the desired microlocal behaviour, which are additionally of polynomial type as per Remark \ref{Rmk: polynomial functionals}.

\begin{definition}[Classical microcausal algebra]
Let $(\mcM, g_\mcM)$ be a $d-$dimensional globally hyperbolic, Lorentzian spacetime with $d \ge 2$. We call \emph{classical microcausal algebra} over $\mcM$, and we denote it by $\mathcal{A}_{cl}$, the commutative, unital $*-$algebra of microcausal functionals $\mathcal{F}_{\mu c}(\mcM)$, endowed with the pointwise product in Equation \eqref{Eq: pointwise prod} and with an involution given by the complex conjugation.    
\end{definition}

To encompass into the classical framework the information of the correlation functions, we shall introduce a suitable deformation of the algebra structure. This procedure is known as \emph{deformation quantization} and it consists in deforming the algebra product by means of a suitable bi-distribution. To identify the bi-distribution which will play this r\^ole, it is convenient to consider the dynamics of the free fields: 
\begin{equation}
    \label{Eq: free dynamics 2 scalar fields}
    \begin{cases}
           P_{0, 1} \varphi_1 := (-\Box_g + m_1^2 + \xi_1 R) \varphi_1 = 0 \\
            P_{0, 2} \varphi_2 := (-\Box_g + m_2^2 + \xi_2 R) \varphi_2 = 0. 
    \end{cases}    
\end{equation}
Observe that, in view of Remark \ref{Rmk: coupling}, the free dynamics is described by a system of uncoupled Green hyperbolic equations. Hence, there exists unique advanced $(A)$ and retarded $(R)$ fundamental solutions for the Klein-Gordon operators $P_{0, \ell}$, $\ell = 1,2$, which we shall denote by $\Delta_{A/R, \ell}$, $\ell = 1,2$, respectively. Moreover, for each $\ell =1,2$, let $\Delta_{+, \ell} \in \mathcal{D}'(\mcM \times \mcM)$ be the associated (global) Hadamard two-point functions. As a by-product of the factorization of the free field dynamics, no mixing between the two fields can occur. More precisely, in view of Remark \ref{Rmk: field doublet}, we can write the correlations of the theory as a diagonal matrix-valued bi-distribution
\begin{equation}
\label{Eq: Delta +}
    \mathbf{\Delta}_+ := \begin{pmatrix}
        \Delta_{+, 1} & 0 \\
        0 & \Delta_{+, 2}
    \end{pmatrix} \in \mathcal{D}'(\mcM \times \mcM; \text{Mat}(2; \mathbb{C})).
\end{equation} 
Hence, in considering a suitable deformation of the pointwise product, a natural choice for the deforming bi-distributions is represented by the Hadamard two-point functions $\Delta_{\ell, +}$, $\ell = 1,2$. For any $F, G \in \mathcal{F}_{\mu c}(\mcM)$ and for each $\ell = 1,2$, we define the \emph{deformation maps} as 
\begin{equation*}
    D_{\hbar \Delta_{+, \ell}} := \langle \hbar \Delta_{+, \ell}, \frac{\delta}{\delta \varphi_{\ell}} \otimes \frac{\delta}{\delta \varphi_{\ell}} \rangle = \int_{\mcM} d\mu_x d\mu_y \, \hbar \Delta_{+, \ell}(x,y)  \frac{\delta}{\delta \varphi_{\ell}(x)} \otimes \frac{\delta}{\delta \varphi_{\ell}(y)},
\end{equation*}
and we introduce the \emph{deformed product} acting as follows: 
\begin{equation}
    \label{Eq: deformed product 2 scalar fields}
    (F \star G)(\varphi_1, \varphi_2) := \mathsf{M} \circ e^{\sum_{\ell = 1,2} D_{\hbar \Delta_{+, \ell}}} \left[F(\varphi_1, \varphi_2) \otimes G(\varphi_1, \varphi_2)\right],
\end{equation}
for any pair of field configurations $(\varphi_1, \varphi_2) \in \mathcal{E}_2(\mcM)$, where $\mathsf{M}$ is the pull back on the space $\mathcal{F}_{\mu c}(\mcM) \otimes \mathcal{F}_{\mu c}(\mcM)$ via the diagonal map $\iota: \mathcal{E}_2 (\mcM) \rightarrow \mathcal{E}_2 (\mcM) \times \mathcal{E}_2 (\mcM)$, $\iota(\Phi) = (\Phi, \Phi)$, $\Phi$ denoting the field doublet as per Remark \ref{Rmk: field doublet}. 

\begin{remark}
Notice that, in Equation \eqref{Eq: deformed product 2 scalar fields}, we could have equivalently written the deformation map as
\begin{equation}
    \sum_{\ell = 1,2} D_{\hbar \Delta_{+, \ell}} := D_{\hbar \, \text{tr}_{\mathbb{C}^2} \mathbf{\Delta}_+}, 
\end{equation}
where $\text{tr}_{\mathbb{C}^2}(\mathbf{\Delta}_{+})$ denotes the trace taken with respect to the matrix indices of $\mathbf{\Delta}_{+}$ -- see Equation \eqref{Eq: Delta +}.
\end{remark}

\noindent Equation \eqref{Eq: deformed product 2 scalar fields} can also be recast in the form of a formal power series in $\hbar$: 
\begin{equation*}
    F \star G:= F \cdot G + \sum_{k := k_1 + k_2 \ge 1} \frac{\hbar^k}{k_1! k_2!} \langle F^{(k_1, k_2)}, \Delta_{+, 1}^{\otimes k_1} \Delta_{+,2}^{\otimes k_2} G^{(k_1, k_2)} \rangle, 
\end{equation*}
where we have left the dependence on the test function and on the field configurations implicit for notational ease. 

\begin{remark}
The choice of a deformation map constructed from the matrix trace of $\boldsymbol{\Delta}_+$ is motivated by the desire to contract only fields of the same type. This reflects the physical requirement that correlations arise between fields of the same species, while at the free level they remain independent and are coupled only through the interaction.
\end{remark}


We are now in the position to define the quantum counterpart of the algebra $\mathcal{A}_{cl}$. 


\begin{definition}[Quantum algebra of microcausal functionals]
Let $\mathbf{\Delta}_+ \in \mathcal{D}'(\mcM \times \mcM; \text{Mat}(2; \mathbb{C}))$ be as per Equation \eqref{Eq: Delta +}, where $\Delta_{+, \ell}, \ell = 1,2$ are the Hadamard two-point functions of the underlying scalar theory. We call \emph{quantum algebra of microcausal (polynomial) functionals} the unital, commutative $*-$algebra $(\text{Pol}_{\mathbf{\Delta}_+} (\mcM), \star, \ast)$. 
\end{definition}

\begin{remark}
We recall that in this setting the expectation value of any observable $A \in \text{Pol}_{\mathbf{\Delta}_+}$ is codified by the evaluation at vanishing fields, that is, $\langle A \rangle = A \vert_{\Phi = 0}$. 
\end{remark}

\begin{definition}[Local and covariant algebra]
\label{Def: local and covariant}
We call \emph{local and covariant algebra of microcausal polynomial functionals} the algebraic closure of $\text{Pol}_{loc}(\mcM)$ under the product $\star_{H}$, where the deformation is implemented by $\text{tr}_{\mathbb{C}^2} \boldsymbol{H}$ with $\boldsymbol{H} \in \mathcal{D}'(\mcM \times \mcM; \text{Mat}(2; \mathbb{C}))$ such that $\boldsymbol{\Delta}_+ - \boldsymbol{H} \in C^{\infty}(\mcM \times \mcM; \text{Mat}(2; \mathbb{C}))$.     
\end{definition}

\noindent Additionally, we define the abstract Wick ordering of an observable $A \in \text{Pol}(\mcM)$ as 
\begin{equation}
\label{Eq: alpha H}
    : A:_{\boldsymbol{H}} := e^{-\frac{1}{2} \text{tr}_{\mathbb{C}^2} \langle \boldsymbol{H}, \frac{\delta^2}{\delta \Phi(x) \delta\Phi(y)} \rangle} A =: \alpha_{- \boldsymbol{H}} (A), 
\end{equation}
where the pairing $\langle \cdot \rangle$ is here defined on $L^2(\mcM) \otimes \mathbb{C}^2$ and the trace is taken with respect to the $\mathbb{C}^2-$indices. The following definition clarifies how to relate this construction to the concrete Wick ordering commonly employed in quantum field theory. 

\begin{definition}\label{Def: Concrete Wick ordering}
     Let $\mathrm{Pol}(\mathcal{M})$ be as per Definition \ref{Def: local and covariant}. Given a (global) Hadamard state $\mathbf{\Delta}_+\in\mathcal{D}^\prime(\mathcal{M}\times\mathcal{M})$ as in Equation \eqref{Eq: Delta +}, we call operator ordering map of $\mathrm{Pol}(\mathcal{M})$ in $\mathrm{Pol}_{\mathbf{\Delta}_+}(\mathcal{M})$ the application 
    $$\alpha_{\mathbf{\Delta}_+}:\mathrm{Pol}(\mathcal{M})\to\mathrm{Pol}_{\mathbf{\Delta}_+}(\mathcal{M})$$
    where $\alpha_{\mathbf{\Delta}_+}$ is as per Equation \eqref{Eq: alpha H} with $\boldsymbol{H}$ replaced by $\boldsymbol{\Delta}_+$. This is a $*$-isomorphism, namely, for all $F,G\in\mathrm{Pol}_{\mathbf{\Delta}_+}(\mathcal{M})$,
    \begin{equation}\label{Eq: Wick products}
    F\star G=\alpha_{\mathbf{\Delta}_+}(\alpha^{-1}_{\mathbf{\Delta}_+}(F) \, \alpha^{-1}_{\mathbf{\Delta}_+}(G)).
    \end{equation}
\end{definition}

\begin{example}
In order to better clarify the above discussion, it is convenient to discuss an explicit computation. More precisely, we shall consider the deformed product of the quadratic functionals discussed in Example \ref{Ex: fun 2 scalar fields}. For any $f, f^{\prime} \in \mathcal{D}(\mcM)$ and for all pair of field configurations $(\varphi_1, \varphi_2) \in \mathcal{E}_2(\mcM)$, we have
\begin{flalign*}
    (\Phi_{\ell}^2[f] \star  \Phi_{\ell}^2[f']) (\varphi_1, \varphi_2) &= \int_{\mcM^2} d\mu_x d\mu_y \, \varphi_{\ell}^2(x) \varphi_{\ell}^2(y) f(x) f'(y)\\ &+ 4\hbar \int_{\mcM^2} d\mu_x d\mu_y \, \Delta_{+, \ell}(x,y) \varphi_\ell(x) \varphi_\ell(y) f(x) f'(y) \\ &+ 2\hbar^2 \int_{\mcM^2} d\mu_x d\mu_y \, \Delta^2_{+, \ell}(x,y) f(x) f'(y), \quad \ell = 1,2,
\end{flalign*}
or, equivalently, 
\begin{equation*}
(\Phi_{\ell}^2[f] \star  \Phi_{\ell}^2[f']) (\varphi_1, \varphi_2) = (\Phi_{\ell}^2[f] \Phi_{\ell}^2[f']) (\varphi_1, \varphi_2) + 4 \hbar\Delta_{+, \ell}(f,f^\prime)(\Phi_{\ell}[f] \Phi_{\ell}[f']) (\varphi_1, \varphi_2) +2\hbar^2\Delta_{+, \ell}^2(f,f^\prime). 
\end{equation*}
Similarly, for the deformed product of the mixed quadratic functionals, we obtain
\begin{flalign*}
    ((\Phi_1\Phi_2)[f] \star  (\Phi_1\Phi_2)[f']) (\varphi_1, \varphi_2) &= \int_{\mcM^2} d\mu_x d\mu_y \, \varphi_1(x) \varphi_2(x) \varphi_1(y) \varphi_2(y) f(x) f'(y) \\ &+ \hbar \sum_{\substack{\ell, j = 1, 2 \\ \ell \ne j}} \int_{\mcM^2} d\mu_x d\mu_y \, \Delta_{+, \ell}(x,y) \varphi_j(x) \varphi_j(y) f(x) f'(y) \\ &+ \hbar^2 \int_{\mcM^2} d\mu_x d\mu_y \, \Delta_{+, 1}(x,y) \Delta_{+,2}(x,y) f(x) f'(y),
\end{flalign*}
and, in a more compact notation, 
\begin{flalign*}
((\Phi_1\Phi_2)[f] \star  (\Phi_1\Phi_2) [f']) (\varphi_1, \varphi_2)& = ((\Phi_1\Phi_2)[f] (\Phi_1\Phi_2) [f']) (\varphi_1, \varphi_2) \\ &+ \hbar \sum_{\ell \ne j} \Delta_{+, \ell}(f,f^\prime)(\Phi_{j}[f] \Phi_{j}[f']) (\varphi_1, \varphi_2) +\hbar^2 \Delta_{+, 1}(f,f^\prime) \Delta_{+,2}(f,f^\prime). 
\end{flalign*}
Note that, upon evaluating the above expressions at vanishing field configurations, they yield the correlations of the underlying theory. Furthermore, to make contact with the standard notation adopted in quantum field theory, a pictorial representation can be employed: 
\begin{flalign*}
\Phi_{1}^2[f] \star \Phi_{1}^2 [f'] &= \Phi_1^2 [f] \Phi_1^2 [f'] + 4 \hbar \wick{
        \c1 \Phi^2_1 [f]
        \c1 \Phi^2_1 [f'] } +
 + 2 \hbar^2 \wick{
        \c1 \Phi_1 [f] \c2 \Phi_1 [f]
        \c1 \Phi_1 [f'] \c2 \Phi_1 [f']}, \\ 
\Phi_{2}^2[f] \star \Phi_{2}^2 [f'] &= \Phi_2^2 [f] \Phi_2^2 [f'] + 4 \hbar \wick{ \c1 {\textcolor{red}{\Phi_2[f]}}
         \c1 {\textcolor{red}{\Phi_2[f']}} } +
 + 2 \hbar^2 \wick{
        \c1 {\textcolor{red}{\Phi_2[f]}} \c2 {\textcolor{red}{\Phi_2[f]}}
        \c1 {\textcolor{red}{\Phi_2[f']}} \c2 {\textcolor{red}{\Phi_2[f']}}}, \\
(\Phi_{1} \Phi_2)[f] \star (\Phi_{1} \Phi_2)[f'] &= (\Phi_1 \Phi_2) [f] (\Phi_1 \Phi_2) [f'] \\ & + \hbar \wick{ \c1 \Phi_1 [f]
         \c1 \Phi_1 [f'] }  \Phi_2[f] \Phi_2[f'] + \hbar  \Phi_1[f] \Phi_1[f'] \wick{ \c1 {\textcolor{red}{\Phi_2[f]}}
         \c1 {\textcolor{red}{\Phi_2[f']}}}  \\ &
 + \hbar^2 \wick{
        \c1 \Phi_1 [f]  \c2 {\textcolor{red}{\Phi_2[f]}}
        \c1 \Phi_1 [f']  \c2 {\textcolor{red}{\Phi_2[f']}}},
\end{flalign*}
where each contraction highlighted in black is performed by means of $\Delta_{+, 1}$, whilst red contractions are implemented by $\Delta_{+,2}$. 
\end{example}

\noindent \textbf{Interacting theory}
We now possess all the ingredients to tackle the interacting scenario. Let us consider the potential in Equation \eqref{Eq: potential 2 scalar field}, which depends only on polynomials in the fields. This requirement will be crucial in Section \ref{Sec: Existence of local solutions to the Renormalization Group flow equations}, as it ensures that one can choose a local potential independent of the field derivatives, thereby preserving the Green hyperbolicity of the operator under scrutiny.
As for the case of a single, self-interacting scalar field, we can define the time ordered product via the Feynman propagators of the free theory. More precisely, for each $\ell = 1,2$, we denote by $\Delta_{F, \ell} := \Delta_{+, \ell} + i \Delta_{A, \ell}$ the Feynman propagators relative to the fields $\varphi_{\ell}$. Also in this case, we can adopt a convenient matrix notation: let $\boldsymbol{\Delta}_F \in \mathcal{D}'(\mcM \times \mcM; \text{Mat}(2; \mathbb{C})$ be defined by 
\begin{equation*}
    \boldsymbol{\Delta}_F(x,y) := \begin{pmatrix}
        \Delta_{F, 1}(x,y) & 0 
        \\
        0 & \Delta_{F, 2}(x,y)
    \end{pmatrix},
\end{equation*}
at the level of integral kernels. The time ordered product can then be completely characterised by the {\it time ordering map} $\mathcal{T}^{\hbar \boldsymbol{H}_F}$ acting on tensor products of elements in $\mathrm{Pol}_{loc}(\mathcal{M})$ as per Definition \ref{Def: local and covariant}. Without delving into the axiomatic construction of such a product, we refer the reader to \cite{brunetti2} and \cite{Rejzner} for further details on the topic. Here we limit ourselves to concretely define its action on the algebra $\mathrm{Pol}_{\boldsymbol{\Delta}_+}(\mathcal{M})$. The time ordering map can be explicitly realised exploiting the identity: 
\begin{equation}\label{Eq: time-ordered product}
	 \mathcal{T}^{\hbar \boldsymbol{H}_F}  (F_1 \otimes \ldots \otimes F_n):=
	  F_1 \star_{F} \ldots \star_{F} F_n =
	\mathsf{M} \circ e^{\sum_{k < j}^n D^{k j}_{\hbar \text{tr}_{\mathbb{C}^2} \boldsymbol{H}_F}}  (F_1
	\otimes \ldots \otimes F_n), 
\end{equation}
where $\mathsf{M}$ denotes the pull back on $\textrm{Pol}_{\mu c}(\mathcal{M})\otimes
	\textrm{Pol}_{\mu c}(\mathcal{M})$ via the diagonal map $\iota$, while $F_i\in\mathrm{Pol}_{\mu c}(\mathcal{M})$, for all $i=1,\dots,n$. Here $\boldsymbol{H}_F\in\mathcal{D}'(\mathcal{M}\times\mathcal{M}; \text{Mat}(2; \mathbb{C}))$ denotes the Feynman parametrix, whilst 
$$ D^{k j}_{\hbar \text{tr}_{\mathbb{C}^2} \boldsymbol{H}_F}:= \sum_{\ell = 1, 2} \langle \hbar H_{F, \ell}, \frac{\delta}{\delta \varphi_{\ell}}\otimes\frac{\delta}
{\delta \varphi_{\ell}}\rangle,$$ which is manifestly symmetric under exchange of $j$ and $k$. In the case $n =2$, Equation \eqref{Eq: time-ordered product} reduces to: 
\begin{equation*}
    F \cdot_T G:= F \cdot G + \sum_{k := k_1 + k_2 \ge 1} \frac{\hbar^k}{k_1! k_2!} \langle F^{(k_1, k_2)}, \Delta_{F, 1}^{\otimes k_1} \Delta_{F,2}^{\otimes k_2} G^{(k_1, k_2)} \rangle, \quad F, G \in \text{Pol}_{loc}(\mcM). 
\end{equation*}

The information of the non-linear, interacting potential in Equation \eqref{Eq: potential 2 scalar field} is encoded at an algebraic level by means of two key objects: the $S$-matrix and the Bogoliubov map. 

\begin{remark}
\label{Rmk: formal power}
Recalling Equation \eqref{Eq: potential 2 scalar field}, we will have to consider perturbative expansions with respect to three coupling parameters $\underline{\lambda} := (\lambda_1, \lambda_2, \lambda_3)$ with $\lambda_1, \lambda_2 > 0$ and $\lambda_1 \lambda_2 > 9 \lambda_3^2$. Unlike the single-scalar case, the interacting potential in Equation \eqref{Eq: potential 2 scalar field} cannot be treated as a power series in a single parameter. Instead, one must work with joint powers of the components of $\underline{\lambda}$. Consequently, factoring out the couplings from the potential is not convenient, and truncations at a given order are understood as truncations {\it jointly} in $\lambda_1, \lambda_2, \lambda_3$. 
\end{remark}

In this context, we call \emph{S-matrix}
    \begin{flalign}
        \mathcal{S}(V) &:= \sum_{n \ge 0} \frac{1}{n!} \left( \frac{i}{\hbar} \right)^n \underbrace{V \cdot_T \ldots \cdot_T V}_{n},\label{Eq: S-Matrix}
    \end{flalign}
where all the couplings are incorporated in the interacting potential $V$. In addition, $\mathrm{Pol}_{\mu c}[[\lambda]](\mathcal{M})$ is the space of functionals defined as a formal power series in the parameters $\underline{\lambda} := (\lambda_1, \lambda_2, \lambda_3)$, such that $F\in \mathrm{Pol}_{\mu c}[[\lambda]](\mathcal{M})$ if $F=\sum\limits_{n_1, n_2, n_3=0}^\infty\lambda_1^{n_1} \lambda_2^{n_2} \lambda_3^{n_3} F_{n_1, n_2, n_3}$ with $F_{n_1, n_2, n_3}\in\mathrm{Pol}_{\mu c}(\mathcal{M})$.

\noindent The S-matrix admits an inverse in the sense of formal power series -- see also Remark \ref{Rmk: formal power}. For future convenience, it is useful to express it explicitly by means of the anti-Feynman parametrix:  
\begin{equation}\label{Eq: Anti-Feynman Parametrix}
    \boldsymbol{H}_{AF} := \boldsymbol{H} - i \boldsymbol{\Delta}_R = \boldsymbol{H}_F^*\in\mathcal{D}^\prime(\mathcal{M}\times\mathcal{M}; \text{Mat}(2; \mathbb{C})), 
\end{equation}
where $\boldsymbol{H} \in \mathcal{D}'(\mathcal{M} \times \mathcal{M}; \text{Mat}(2; \mathbb{C}))$ is any bi-distribution which differs from the two-point correlation function of a Hadamard state by a smooth remainder, while $\ast$ denotes the matrix and $L^2$ formal adjoint. Thus, the inverse of $\mathcal{S}(\lambda V)$ can be written in the form: 
\begin{equation}
\label{Eq: inverse Smat}
    \mathcal{S}^{-1} (V) := \sum_{n \ge 0} \frac{1}{n!} \left( - \frac{i}{\hbar} \right)^n \underbrace{V \cdot_{\bar{T}} \ldots \cdot_{\bar{T}} V}_{n}, 
\end{equation}
where $\cdot_{\bar{T}}$ is the deformed product implemented by $\boldsymbol{H}_{AF}$. In addition the superscript $\star$ indicates that $\mathcal{S}^{\star -1}$ is the inverse of the S-matrix with respect to the product induced by $\boldsymbol{H}$, that is
$$\mathcal{S}(V)\star_{\boldsymbol{H}}\mathcal{S}^{ -1} ( V)=\mathcal{S}^{ -1} (V)\star_{\boldsymbol{H}} \mathcal{S}(\lambda V)=\textbf{1},$$ where $\textbf{1}$ is the identity functional as per Example \ref{Ex: fun 2 scalar fields}.   

The S-matrix and its inverse are essential tools in constructing the Bogoliubov map, which serves to represent interacting observables -- expressed as formal power series in the couplings -- in terms of their free counterparts. In the case of two, mutually interacting scalar fields as per Equation \eqref{Eq: action 2 scalar fields}, it can be defined in the following way. Given $F\in \mathrm{Pol}_{loc}(\mathcal{M})$, the Bogoliubov map is 
\begin{equation}
\label{Eq: Bogoliubov map}
    R_{V}(F):= \mathcal{S}^{-1}(V) \star_{\boldsymbol{H}} [\mathcal{S}(\lambda V) \cdot_T F], \, \, \forall F \in \mathrm{Pol}_{loc}(\mathcal{M}). 
\end{equation}
where $\cdot_T$ is the time ordered product as per Equation \eqref{Eq: time-ordered product}, whilst $\star_{\boldsymbol{H}}$ is the deformed product implemented by the bi-distribution $\boldsymbol{H}$ -- see also Definition \ref{Def: local and covariant}.

We conclude the section with an example showing how to perform explicit computations by means of the Bogoliubov map. 

\begin{example}
\label{Ex: interacting fields}
Consider the quadratic functionals as per Equations \eqref{Eq: Phi1^2}, \eqref{Eq: Phi2^2} and \eqref{Eq: Phi1, Phi2} and let us discuss how to compute the expectation value of their interacting counterparts up to $\mathcal{O}(\underline{\lambda}^3)$, \textit{i.e.}, up to second order jointly in the coupling parameters. Barring replacing $\boldsymbol{H}, \boldsymbol{H}_F$ and $\boldsymbol{H}_{AF}$ with the corresponding states, the expectation value of any observable corresponds to the evaluation at vanishing fields $\Phi := (\varphi_1, \varphi_2) = 0$. We begin by expanding the S-matrix and its inverse in Equations \eqref{Eq: S-Matrix} and \eqref{Eq: inverse Smat} as
\begin{flalign}
    \mathcal{S}(V) & = \mathbf{1} + \frac{i}{\hbar} V - \frac{1}{2\hbar^2} V \star_F V + \mathcal{O}(\underline{\lambda}^3), \\
     \mathcal{S}^{-1}(V) & = \mathbf{1} - \frac{i \lambda}{\hbar} V - \frac{1}{2} \frac{\lambda^2}{\hbar^2} V \star_{AF} V + \mathcal{O}(\underline{\lambda}^3).
\end{flalign}
Hence, suppressing the dependence on the field configurations for notational ease,
\begin{flalign}
    R_{V}(F[f]) &=_{\mathcal{O}(\underline{\lambda}^3)} F[f] + \frac{i}{\hbar} \left(V[h] \cdot_T F[f] - V[h] \star_H F[f] \right) \notag \\ & -\frac{1}{2 \hbar^2} \left[ (V[h] \cdot_{\bar{T}} V[h]) \star_H F[f] + (V[h] \cdot_T V[h]) \cdot_T F[f] - 2 V[h] \star_H (V[h] \cdot_T F[f])\right], \label{Eq: Second Order Expansion}
\end{flalign}
where the functional $F \in \text{Pol}_{loc}(\mcM)$ is taken to coincide either with $\Phi_{\ell}^2$, $\ell = 1,2$ or $(\Phi_1 \Phi_2)^2$. Notice that only maximally contracted terms contribute to the expectation value of $F$. To identify the relevant contributions modulus combinatorial factors, it is convenient to resort to a graphical approach. Henceforth, we shall denote in black the contractions between the fields $\varphi_1$, which are performed by means of $H_1, H_{F,1}$ and $H_{AF, 1}$. We will use red to mark those between the fields of type $2$ implemented by $H_2, H_{F,2}$ and $H_{AF, 2}$. For simplicity, we shall discuss in detail only the case $F[f] = (\Phi_1 \Phi_2)^2[f]$, the other two following suit.  
It is possible to infer that: 
\begin{itemize}
 \item[\ding{104}] none of the terms of the first order expansion survives. In fact, since the functional $F$ is quadratic in the fields, in the term $V[h] \star_{\boldsymbol{H}} F [f]$ as well as in $V[h] \cdot_T F[f]$, it is possible to perform at most two contractions. This leaves at least one power of $\Phi_1$ and $\Phi_2$ remaining,
\item[\ding{104}] at second order in perturbation theory, from $V[h] \star_{\boldsymbol{H}} (V[h] \cdot_T F[f])$, one obtains three separate contributions, corresponding to the three different terms in the expression of $V[h]$. The first contributions can be pictorially represented as follows 
     \begin{equation*}
        \wick{
        \c1 \Phi_1 [h] \c2 \Phi_1 [h] \c3 {\textcolor{red}{\Phi_2[h]}} \c4 {\textcolor{red}{\Phi_2[h]}} \, \, \;
        (\c1 \Phi_1[h] \c2 \Phi_1 [h] \c5 \Phi_1 [h] \c6 \Phi_1 [h] \, \;
        \c5 \Phi_1 [f] \c6 \Phi_1 [f] \c3 {\textcolor{red}{\Phi_2[h]}} \c4 {\textcolor{red}{\Phi_2[h]}})}, 
    \end{equation*}
and it yields, modulus constant factors,
\begin{equation*}
    \propto \lambda_1 \lambda_3 \int_{\mcM^3} d\mu_x d\mu_y d\mu_z \, H^2_{1}(x,y) H^2_{2}(x,z) H_{F,1}^2(y,z) h(x) h(y) f(z). 
\end{equation*}
The second non-vanishing term is of the form: 
\begin{equation*}
        \wick{
        \c1 \Phi_1 [h] \c2 \Phi_1 [h] \c3 {\textcolor{red}{\Phi_2[h]}} \c4 {\textcolor{red}{\Phi_2[h]}} \, \, \;
        (\c3 {\textcolor{red}{\Phi_2[h]}}  \c4 {\textcolor{red}{\Phi_2[h]}}  \c5 {\textcolor{red}{\Phi_2[h]}}  \c6 {\textcolor{red}{\Phi_2[h]}}  \, \;
        \c1 \Phi_1 [f] \c2 \Phi_1 [f] \c5 {\textcolor{red}{\Phi_2[h]}} \c6 {\textcolor{red}{\Phi_2[h]}})}, 
\end{equation*}
and, at the level of integral kernels, it reads
\begin{equation*}
    \propto \lambda_2 \lambda_3 \int_{\mcM^3} d\mu_x d\mu_y d\mu_z \, H^2_{2}(x,y) H^2_{1}(x,z) H_{F,2}^2(y,z) h(x) h(y) f(z). 
\end{equation*}
Finally, the third contributions corresponds to 
\begin{equation*}
        \wick{
        \c1 \Phi_1 [h] \c2 \Phi_1 [h] \c3 {\textcolor{red}{\Phi_2[h]}} \c4 {\textcolor{red}{\Phi_2[h]}} \, \, \;
        (\c2 \Phi_1[h]  \c5 \Phi_1[h]  \c4 {\textcolor{red}{\Phi_2[h]}}  \c6 {\textcolor{red}{\Phi_2[h]}}  \, \;
        \c5 \Phi_1 [f] \c1 \Phi_1 [f] \c3 {\textcolor{red}{\Phi_2[h]}} \c6 {\textcolor{red}{\Phi_2[h]}})}, 
\end{equation*}
\begin{equation*}
    \propto \lambda_3^2 \int_{\mcM^3} d\mu_x d\mu_y d\mu_z \, H_{1}(x,y) H_{2}(x,y) H_{1}(x,z) H_{2}(x,z) H_{F,1}(y,z) H_{F,2}(y,z) h(x) h(y) f(z). 
\end{equation*}
We stress once again that all the contractions between elements within the round brackets are implemented by $\boldsymbol{H}_F$, while all others by $\boldsymbol{H}$. 

\item[\ding{104}] Focusing instead on $(V[h] \cdot_T V[h]) \cdot_T F[f]$ at a graphical level the same contractions contribute, hence we refrain from repeating the above analysis. The only difference with the previous case is that the non-vanishing contributions are related to different kernels and, thus, they correspond to the analytic expressions:
\begin{flalign*}
         \propto \lambda_1 \lambda_3 \int_{\mcM^3} d\mu_x d\mu_y d\mu_z \, &[H_{F, 1}^2(x,y) H_{F,1}^2(x,z) H_{F,2}^2(y,z) \\ &+ H_{F,1}^2(x,y) H_{F,1}^2(y,z) H_{F,2}^2(x,z)] h(x) h(y) f(z),
\end{flalign*}
\begin{flalign*}
         \propto \lambda_2 \lambda_3 \int_{\mcM^3} d\mu_x d\mu_y d\mu_z \,& [H_{F,2}^2(x,y) H_{F,2}^2(x,z) H_{F,1}^2(y,z) \\ &+ H_{F,2}^2(x,y) H_{F,1}^2(x,z) H_{F,2}^2(y,z)] h(x) h(y) f(z).
\end{flalign*}
The term $(V[h] \cdot_{\bar{T}} V[h]) \star_{\boldsymbol{H}} F[f]$ can be dealt with in a similar way by plugging in the above formulae the right parametrices. 
\end{itemize}
\end{example}

\begin{remark}
\label{Rem: Renormalization ambiguities in the example}
    Observe that the expressions in Example \ref{Ex: interacting fields} is merely formal at the present stage, as they involve powers of the parametrices $\boldsymbol{H}_F$ and $\boldsymbol{H}_{AF}$, which are ill-defined on the diagonal on account of H\"ormander criterion \cite{Hormander_1990}. This requires a renormalisation procedure, which is inherently subject to ambiguities. Since they do not play a relevant r\^{o}le in the derivation of the results of this paper, we shall refrain from discussing them in detail.
\end{remark}

\paragraph{Wetterich Equation in the Local Potential Approximation}
\label{Sec: Wetterich for two scalar fields}

In order to regularize the infrared behaviour of the theory, we introduce a scale-dependent mass regulator that suppresses low-energy modes: 
\begin{equation}
    \label{Eq: infrared regulator}
    Q_k (\chi_1, \chi_2) = - \frac{1}{2} \int_{\mathcal{M}} d\mu_x \, \left[q_{k,1}(x) T\chi_1^2(x) + q_{k,2}(x) T\chi_2^2 (x)\right], 
\end{equation}
where, without loss of generality, we can fix $q_{k,1}(x) = q_{k,2}(x) = k^2 f(x)$ with $f(x)$ a smooth cut off function, compactly supported in large regions of spacetime. 

\begin{remark}
    The choice of a regulator of this type may seem arbitrary to some extent, since nothing prevents us from adding also contributions proportional to the mixed product $\varphi_1 \varphi_2$. Yet, the choice in Equation \eqref{Eq: infrared regulator} is dictated by the fact that the only coupling between the fields occurs at the level of the interaction potential via the term $\propto \lambda_3 \varphi_1^2 \varphi_2^2$. If a mass term $\propto \varphi_1 \varphi_2$ were also to be present, it would have been more convenient to retain in the expression of $Q_k$ the mixed terms as well. 
\end{remark}

\noindent Denoting by
\begin{equation*}
    \label{Eq: current}
    J(\chi_1, \chi_2) := \int_{\mathcal{M}} d\mu_x \left[ j_1(x) \chi_1(x) + j_2(x) \chi_2(x)\right], \, \, j_1, j_2 \in C^{\infty}_0(\mathcal{M}),
\end{equation*}
an external current, we can introduce the (regularised) generating functional of the theory
\begin{equation}
\label{Eq: Zk scalar}
    Z_k(j_1, j_2) = \omega(S(V)^{-1} \star S(V + J+ Q_k)),
\end{equation}
together with its connected counterpart
\begin{equation}
\label{Eq: Wk scalar}
    W_k(j_1, j_2) := -i \log Z_k(j_1, j_2). 
\end{equation}
Here, we denote by $\omega$ an arbitrary Hadamard state on the underlying quantum algebra, albeit not necessarily quasi-free for the free theory. Taking the first functional derivative of $W_k$, we obtain the \emph{classical fields}, namely,
\begin{equation}
        \frac{\delta W_k(j_1, j_2)}{\delta j_\ell} = \frac{1}{Z_k(j_{\varphi_1}, j_{\varphi_2})} \omega(S(V)^{-1} \star (S(V + J+ Q_k) \cdot_T \chi_\ell(x))) =: \varphi_\ell(x), \quad \ell = 1,2,
\end{equation}
where $j_{\varphi_\ell}, \ell = 1,2$ are the currents which solve the above equation as functions of $\varphi_\ell$. The \emph{effective action} of the theory is then obtained as the Legendre transform with respect to both currents, \textit{i.e.},
\begin{equation}
    \label{Eq: Gammak scalar}
    \tilde{\Gamma}_k (\varphi_1, \varphi_2) = W_k(j_{\varphi_1}, j_{\varphi_2}) - J_{\Phi}(\varphi_1, \varphi_2), 
\end{equation}
where $J_{\Phi}$ is the functional defined in Equation \eqref{Eq: current} with $j_\ell \equiv j_{\varphi_\ell}$, $\ell = 1,2$. By definition, the first functional derivative of $\tilde{\Gamma}_k$ yields the \textit{quantum equations of motion} which are given by: 
\begin{equation*}
    \frac{\delta \tilde{\Gamma}_k}{\delta \varphi_{\ell}} = -j_{\varphi_\ell}, \quad \ell =1,2.
\end{equation*}
In a similar fashion as in \cite{Dangelo_23}, one can also introduce the \emph{average effective action}
\begin{equation*}
    \Gamma_k(\varphi_1, \varphi_2) = \tilde{\Gamma}_k(\varphi_1, \varphi_2) - Q_k(\varphi_1, \varphi_2). 
\end{equation*} 
With all the necessary ingredients in place, we now turn to the derivation of the flow equations for the generating functionals, which take the form of differential equations in the scale parameter $k$. In particular, we derive the flow equation for the connected generating functional as well as that for the average effective action. The former is commonly known in the literature as \textbf{Polchinski equation} \cite{Polchinski_1984} and it reads
\begin{equation}
    \label{Eq: Polchinski eq 2 scalars}
    \partial_k W_k(j_1, j_2) = - \frac{1}{2} \int_{\mcM} d\mu_x \, \frac{1}{Z_k(j_1, j_2)} \sum_{\ell = 1,2} \partial_k q_{k, \ell}(x) \omega\left(S(V)^{-1} \star (S(V + J + Q_k) \cdot_{T} T\chi^2_{\ell}(x))\right),  
\end{equation}
whilst the latter goes under the name of \textbf{Wetterich equation} and it takes the form
\begin{flalign}
    \label{Eq: Wetterich eq 2 scalars}
    \notag
    \partial_k \Gamma_k(\varphi_1, \varphi_2) &= - \frac{1}{2} \int_{\mcM} d\mu_x \, \frac{1}{Z_k(j_1, j_2)} \sum_{\ell = 1,2} \partial_k q_{k, \ell}(x) \\ &\left[\omega\left(S(V)^{-1} \star (S(V + J + Q_k) \cdot_{T} T\chi^2_{\ell}(x)\right)\right] - \varphi_{\ell}^2\}.  
\end{flalign}
Observe that both the Polchinski and the Wetterich equations involve expectation values taken with respect to a chosen state $\omega$, which at this stage is not assumed to be quasi-free. However, since our ultimate interest lies in physically relevant states, we may, without loss of generality, restrict our attention to quasi-free ones. This assumption significantly simplifies the analysis, as it allows for the use of efficient approximation schemes to obtain approximate solutions to the flow equations. In particular, we adopt the following ansatz: we assume the existence of a local functional $U_k$, depending only on the fields and not on their derivatives, such that the average effective action can be written in the form
\begin{equation}
\label{Eq: LPA Scalar}
    \Gamma_k(\varphi_1, \varphi_2) = - \int_{\mcM} d\mu_x \, \left(\frac{1}{2} \nabla^{\mu} \varphi_1(x) \nabla_{\mu} \varphi_1(x) + \frac{1}{2} \nabla^{\mu} \varphi_2(x) \nabla_{\mu} \varphi_2(x) + U_k(\varphi_1, \varphi_2)\right). 
\end{equation}
This approximation scheme is commonly referred to in the literature as \emph{local potential approximation} (LPA), see \cite{DAngeloDragoPinamontiRejzner2024}. Within such scheme, the left hand-side of Equation \eqref{Eq: Wetterich eq 2 scalars} becomes
$$\partial_k \Gamma_k = - \int_{\mcM} d\mu_x \, \partial_k U_k.$$ For what concerns the right hand-side, the rationale is to truncate the effective action up to quadratic contributions in the fields. This can be achieved by considering an expansion of $\Gamma_k$ around a classical solution $\varphi_{\ell, cl}$ of the quantum equation of motion, \textit{i.e.}, $\frac{\delta \Gamma_k}{\delta \varphi_\ell} \vert_{\varphi_{\ell, cl}} = 0$, thus obtaining
\begin{equation*}
    \Gamma_k(\Phi) = \Gamma_k(\Phi_{cl}) + \frac{1}{2} \Phi^T \Gamma_k^{(2)} (\Phi_{cl}) \Phi + \mathcal{O} (|\Phi - \Phi_{cl}|^3) =: \Gamma_k^t(\Phi) + \mathcal{O} (|\Phi - \Phi_{cl}|^3),
\end{equation*}
where $\Phi$ is the field doublet as per Remark \ref{Rmk: field doublet}, whilst $\Gamma_k^{(2)}$ is the Hessian matrix of $\Gamma_k$. Hence, we can define the \textbf{truncated action} corresponding to the effective action $\Gamma_k^t$ as
\begin{equation*}
    I_0^t := - \int_{\mcM} d\mu_x \, \left( \frac{1}{2} (\nabla \Phi)^T (\nabla \Phi) + \frac{1}{2} \Phi^T U_{k}^{(2)} (\Phi_{cl}) \Phi \right). 
\end{equation*}
In the above expression, $U_k^{(2)}(\Phi_{cl})$ is the Hessian matrix of the truncated effective potential in Equation \eqref{Eq: LPA Scalar}, \textit{i.e.},
\begin{equation}
    \label{Eq: matrix Uk scalar}
    U_k^{(2)} (\Phi_{cl}) = \begin{pmatrix}
        \frac{\delta^2 U_{k}}{\delta \varphi_1^2} & 0\\
      0 &     \frac{\delta^2 U_{k}}{\delta \varphi_2^2}
    \end{pmatrix},
\end{equation}
where we have set the off-diagonal contributions to zero on account of the fact that no mixing occurs at the level of the masses. Thus, we can define the \textbf{effective mass of the theory} 
\begin{equation*}
    M := I_0^t - I_0 = - \frac{1}{2} \int_M d\mu_x \sum_{\ell=1,2} \left(\frac{\delta^2 U_{k} (\Phi_{cl})}{\delta \varphi_\ell^2} - m_{\ell}^2\right) \varphi_{\ell}^2(x) =: \sum_{\ell=1,2} M_{\ell}
\end{equation*}
where $I_0$ is the free action of the theory as per Equation \eqref{Eq: action 2 scalar fields}, while $M_{\ell} := \frac{\delta^2 U_{k}}{\delta \varphi_\ell^2} (\Phi_{cl}) - m_{\ell}^2$ for each $\ell =1,2$. This splitting of the effective mass suggests that at this level we can treat the two field theories as independent. As a consequence, applying the results in \cite{DAngeloDragoPinamontiRejzner2024} for a single scalar field to this setting, we can write the Wetterich equation in the LPA as
\begin{equation}
    \label{Eq: right handside}
    \partial_k U_k (\varphi_1, \varphi_2) = \lim_{y \rightarrow x} \frac{1}{2} \sum_{\ell =1,2} \partial_k q_{k, \ell} (x) (\Delta_{F, M_\ell, k}(y,x) - H_{F, M_{\ell}, k} (y,x)). 
\end{equation}
Here $\Delta_{F, M_{\ell}, k}$ is the Feynman propagator for the theory $I_0^t + Q_k$ with $P_{k, \ell} := \Box_g + \frac{\delta^2 U_{k}}{\delta \varphi_{\ell}^2}$, while $H_{F, M_\ell, k}$ is the associated Feynman parametrix. Following the rationale detailed in \cite{DAngeloDragoPinamontiRejzner2024}, observe that Equation \eqref{Eq: right handside} can be rewritten in terms of the difference between the symmetric part of the two-point function and the Hadamard parametrix of the theory described by the action $I_0 + M + Q_k$, \textit{i.e.}, 
\begin{equation}
\label{Eq: LPA symmetric part}
     \partial_k U_k (\varphi_1, \varphi_2) = \lim_{y \rightarrow x} \frac{1}{2} \sum_{\ell =1,2} \partial_k q_{k, \ell} (x) (\Delta_{S, M_\ell, k}(y,x) - H_{ M_{\ell}, k}(y,x)). 
\end{equation}
Note that, up to this point, we have not make any assumption on the expression of the effective potential $U_k$, but we have simply require it to be a local functional of the fields. Yet, to compute explicitly the beta functions for the relevant couplings in Equation \eqref{Eq: potential 2 scalar field}, such a choice needs to be specified. Heretofore, we will require that the potential $U_k$ is of the form: 
\begin{equation}
    \label{Eq: explicit choice of Uk}
    U_k(\varphi_1, \varphi_2) := U_{k, 0} (\varphi_1, \varphi_2) + \underbrace{\left( m_{1, k}^2 + \frac{\lambda_{3,k}}{2} \rho_2 + \frac{\lambda_{1,k}}{6} \rho_1 \right)}_{M^2_{1, k}} \rho_1 +  \underbrace{\left( m_{2, k}^2 + \frac{\lambda_{3,k}}{2} \rho_1 + \frac{\lambda_{2,k}}{6} \rho_2 \right)}_{M^2_{2, k}} \rho_2, 
\end{equation}
where $\rho_\ell := \frac{1}{2} \varphi_{\ell}^2$, for $\ell = 1,2$. 

\begin{remark}
The explicit expression of the effective potential $U_k$ is dictated by the form of the interacting potential in Equation \eqref{Eq: potential 2 scalar field}. Yet, notice that the splitting of the mutual interaction term $\frac{\lambda_{3, k}}{4} \varphi_1^2 \varphi_2^2$ is arbitrary, since nothing prevents us from incorporating it entirely in either one of the masses $M^2_{\ell, k}$. We stress that the choice made in Equation \eqref{Eq: explicit choice of Uk} reflects the intrinsic symmetry under exchange of the two fields, hence it is the most physically sensible.   
\end{remark}

In addition, the ensuing analysis will be tailored to the case of $4-$dimensional Minkowski spacetime so to have manageable, explicit expressions for the Hadamard parametrix of the underlying theory. If we choose as a reference state $\omega$ the vacuum state, the right hand-side of Equation \eqref{Eq: LPA symmetric part} takes the form, \textit{cf.}, \cite[Eq. (96)]{Brunetti_2009}: 
\begin{equation*}
    B := \frac{k}{8 \pi^2} (k^2 + M_{k, 1}^2) \log \left( \frac{k^2 + M_{ k, 1}^2}{\mu^2} \right) + \frac{k}{8 \pi^2} (k^2 + M_{k, 2}^2) \log \left( \frac{k^2 + M_{k, 2}^2}{\mu^2} \right), 
\end{equation*}
where $\mu^2$ is a reference energy scale, which encodes a residual renormalisation freedom. A direct computation yields the beta functions for the relevant couplings: 
\begin{flalign}
\label{Eq: beta function U0}
k \partial_k U_{k, 0} &= \frac{k^2}{8 \pi^2} \sum_{\ell =1,2} (m^2_{\ell, k} + k^2) \log \left( \frac{k^2 + m_{\ell, k}^2}{\mu^2}\right) \\
\label{Eq: beta function m}
k \partial_k m_{k, \ell}^2 &= \frac{k^2}{8 \pi^2} \left( \frac{\lambda_{\ell, k}}{6} \left[ \log\left(\frac{k^2 +m_{\ell, k}^2}{\mu^2} \right) + 1 \right] + \frac{\lambda_{k, 3}}{2} \left[ \log \left( \frac{k^2 + m_{j,k}^2}{\mu^2} \right) + 1\right] \right), \, \ell, j=1,2, j \ne \ell\\
\label{Eq: beta function lambda}
k \partial_k \lambda_{\ell, k} &= \frac{3 k^2}{8 \pi^2} \left[ \frac{\lambda_{\ell,k}^2}{36} \frac{1}{k^2 + m_{\ell, k}^2} + \frac{\lambda_{3,k}^2}{4} \frac{1}{k^2 + m_{j, k}^2} \right], \quad \ell, j=1,2, j \ne \ell\\
\label{Eq: beta function lambda3}
k \partial_k \lambda_{3,k} &= \frac{k^2}{8\pi^2} \left[ \frac{\lambda_{3,k} \lambda_{1,k}}{12(k^2+m_{1,k}^2)} + \frac{\lambda_{3,k} \lambda_{2,k}}{12(k^2 + m_{2,k}^2}
\right]. 
\end{flalign}
The system of flow equations defines a set of coupled first-order differential equations which, for $k > 0$ and away from the singular loci $k^2 + m^2_{\ell, k} = 0$, $\ell =1,2$ admits a unique local solution for given initial data on account of the Picard–Lindel\"of theorem, as the beta functions are smooth in the couplings and continuous with respect to the scale parameter.

\begin{remark}
The local form of the right hand-side in Equation \eqref{Eq: LPA symmetric part} and, in turn, the resulting beta functions in Equations \eqref{Eq: beta function U0}--\eqref{Eq: beta function lambda3} crucially depend on the dimension of the underlying spacetime. In generic dimension $d > 2$, provided one can identify a suitable equilibrium quasi-free state, the above computations can be reproduced \emph{mutatis mutandis} without any significant obstruction -- see \cite[Sec. 6]{Dangelo_23} for additional examples in the case of a single scalar field. 
\end{remark}

We conclude this section by observing that, after a suitable rescaling, one may derive the corresponding flow equations for dimensionless couplings, see, \textit{e.g.}, \cite[Sec.~6]{DAngeloDragoPinamontiRejzner2024}. In the same spirit as \cite[Sec.~6.2]{DAngeloDragoPinamontiRejzner2024}, one could also analyse the high-temperature limit $\beta \to 0$ on Minkowski spacetime, as well as Bunch-Davies states on de Sitter spacetime. We shall not pursue these directions here, as they lie beyond the scope of the present work.

\subsubsection{Stochastic fields via the Martin-Siggia-Rose formalism}
\label{Sec: pAQFT for Martin-Siggia-Rose fields}
As a direct application of the formalism developed for two mutually interacting scalar fields, we shall discuss the case of a stochastic quantum field theory in the Martin-Siggia-Rose (MSR) formalism, see, \textit{e.g.}, \cite{Bonicelli_2025}, \cite{Zanella_2002}, \cite{Carosso_2020}. This approach allows to deal with fields whose dynamics is driven by a random potential, here specified by an additive Gaussian white noise. The idea underpinning the MSR formalism is to introduce an auxiliary field, which encompasses the stochastic correlations of the noise. Although originally designed to deal with stochastic differential equations (SDEs) or parabolic stochastic partial differential equations, this framework can be applied \emph{mutatis mutandis} also to elliptic SPDEs. The key distinction lies in the interpretation of the auxiliary field: in the elliptic setting, it serves only as a Lagrange multiplier, rather than a genuine response field, as the absence of a temporal direction prevents the formulation of a causal structure.

\begin{remark}
In this work we do not address the question of the validity of the 
Martin-Siggia-Rose representation for stochastic partial differential 
equations. This issue involves analytic and probabilistic aspects which lie 
beyond the scope of the present analysis, see \cite{Bonicelli_2025}. Instead, we shall take the MSR representation as a standing assumption and use it as the starting point for 
the algebraic construction developed below.
\end{remark}

On $(\mathbb{R}^d, \delta)$, with $d \ge 1$, let $\xi$ denote a white noise, \textit{i.e.}, a Gaussian-centred random distribution taking values in $\mathcal{S}'(\mathbb{R}^d)$, whose correlations read
\begin{equation}
    \label{Eq: noise correlations}
    \mathbb{E}(\xi(f)) = 0,\, \, \, \mathbb{E}(\xi(f) \xi(g)) = 2 D \langle f, g \rangle_{L^2(\mathbb{R}^d)}, \quad \text{for all} \, f, g \in \mathcal{S}(\mathbb{R}^d). 
\end{equation}
Here $D \in \mathbb{R}$ denotes the \emph{diffusion coefficient} and the factor $2$ is purely conventional. 
With a slight abuse of notation, we can equivalently rewrite Equation \eqref{Eq: noise correlations} at the level of distributional kernels, that is, 
\begin{equation}
    \label{Eq: noise correlations components}
    \mathbb{E}(\xi(x)) = 0,\, \, \, \mathbb{E}(\xi(x) \xi(y)) = 2 D \delta(x-y). 
\end{equation}
In the aforementioned setting, consider the following non-linear elliptic stochastic PDE  
\begin{equation}
\label{Eq: MSR SPDE}
    (-\Delta + m^2) \varphi + \frac{\lambda}{2} \varphi^2 = \xi, \quad m^2 > 0, \lambda \in \mathbb{R}_+.
\end{equation}
Let $\varphi \in \mathcal{E}(\mathbb{R}^d) \equiv C^{\infty}(\mathbb{R}^d; \mathbb{R})$ be any smooth field configuration whose dynamics is governed by Equation \eqref{Eq: MSR SPDE}. Via the Martin-Siggia-Rose formalism applied to this case \cite{Martin_1973}, we can replace the noise in Equation \eqref{Eq: MSR SPDE} with an auxiliary field $\tilde{\varphi} \in \mathcal{E}(\mathbb{R}^d)$, which plays the r\^ole of a Lagrange multiplier enforcing Equation \eqref{Eq: MSR SPDE}. The full action of the theory thus takes the form 
\begin{equation}
\label{Eq: MSR action}
    I[\varphi, \tilde{\varphi}] = \int_{\mathbb{R}^d} dx \, \{ \tilde{\varphi}(x) \left[ - \Delta \varphi(x) + m^2 \varphi(x) + \frac{\lambda}{2} \varphi^2(x) \right] - D \tilde{\varphi}^2(x)\}. 
\end{equation}
On $(\mathbb{R}^d, \delta)$, thanks to translational invariance of the underlying background and to Fourier techniques, the Helmholtz operator $L:=-\Delta + m^2$ with $m^2 > 0$ admits an exact inverse $G \in \mathcal{S}'(\mathbb{R}^d \times \mathbb{R}^d)$, whose integral kernel reads
\begin{equation*}
    G(x-y) \sim |x-y|^{2-d} e^{-m |x-y|}, \quad \forall x,y \in \mathbb{R}^d, 
\end{equation*}
modulus constant factors. More in general, on any smooth, connected and oriented Riemannian manifold $(\mathcal{N}, h)$, the Helmholtz operator $L:=-\Delta + m^2$ with $m^2 > 0$ is elliptic. Therefore, locally in a neighbourhood $\mathcal{O} \subset \mathcal{N}$, there exists a bi-distribution $G \in \mathcal{D}'(\mathcal{O} \times \mathcal{O})$ such that $L_x G(x,y) = L_y G(x,y) = \delta(x,y)$, where $\delta$ is the Dirac delta distribution. Under the additional assumption that $(\mathcal{N}, g_{\mathcal{N}})$ is compact without boundary, the operator $L$ is self-adjoint, elliptic and strictly positive for $m^2 > 0$ on $L^2(\mathcal{N}, \mu_{g_{\mathcal{N}}})$ and, consequently, invertible. This implies that there exists a globally defined fundamental solution $G \in \mathcal{D}'(\mathcal{N} \times \mathcal{N})$, which is exact and not modulus smoothing. In these settings, it has been proven in \cite{Dappiaggi_2020} that the perturbative approach to both free and interacting quantum field theories extends to Euclidean scenario. In particular, it is possible to construct interacting observables via a counterpart of the Bogoliubov map, with the deformed product $\star$ replaced by the pointwise product and the time-order one $\cdot_T$ implemented in a similar way by the exact inverse $G$, \textit{cf.}, \cite[Sec. 7]{Dappiaggi_2020}. 
Yet, we refrain from pursuing this approach, as several technical and conceptual aspects of interacting Euclidean quantum field theories within the algebraic framework remain largely unexplored in the literature, such as the definition of time-ordered products in settings where only parametrices are available.
For these reasons, henceforth we shall focus on a counterpart of the action in Equation \eqref{Eq: MSR action} living on $d-$dimensional Minkowski spacetime $(\mathbb{R}^d, \eta)$.


\begin{remark}
\emph{A priori} there is no obstruction to constructing such a field theory on a generic Lorentzian manifold $(\mathcal{M},g)$. However, the advantages of working on $(\mathbb{R}^d, \eta)$ are two-fold. On the one hand, on Minkowski spacetime we have a natural notion of vacuum state and, hence, an explicit, global expression for the Hadamard parametrix. As we shall see in the derivation of the Wetterich equation, this entails relevant simplifications in the computations. On the other hand, working with the flat metric, allows us to make explicit contact with the results obtained in \cite{Duch_2025}, albeit via an alternative approach. 
\end{remark}

From this point onward, we shall denote by $P_0 := -\Box_{\eta} + m^2$ the Green hyperbolic operator ruling the free dynamics of the scalar field $\varphi$. The corresponding action is given by
\begin{equation}
\label{Eq: MSR Lorentzian action}
    I[\varphi, \tilde{\varphi}] = \int_{\mathbb{R}^d} dx \, \{ \tilde{\varphi}(x) \left[ - \Box_{\eta} \varphi(x) + m^2 \varphi(x) + \frac{\lambda}{2} \varphi^2(x) \right] - D \tilde{\varphi}^2(x)\}. 
\end{equation}
In this Lorentzian construction, one of the fields $\tilde{\varphi}$ retains the r\^ole of an auxiliary, non-dynamical variable, and it is removed at the end of the computation by setting it equal to a constant. 

\begin{remark}
In contrast with the case of the two, mutually interacting scalar fields investigated in Section \ref{Sec: two scalar fields}, the mass term in the action is of the form $m^2 \varphi(x) \tilde{\varphi}(x)$, hence it couples the dynamical and auxiliary fields. In spite of this fact, we can still introduce a field doublet of the form
$$\boldsymbol{\Phi} := \begin{pmatrix}
    \varphi \\
    \tilde{\varphi}
\end{pmatrix} \in C^{\infty}(\mathbb{R}^d) \otimes \mathbb{R}^2$$ and write the free action as a function of this doublet, \textit{i.e.}, 
\begin{equation*}
    I_0(\Phi) = \frac{1}{2} \int_{\mathbb{R}^d} dx \, \boldsymbol{\tilde{\Phi}}^T(x) \mathrm{P}_0 \boldsymbol{\Phi}(x),
\end{equation*}
where $\boldsymbol{\tilde{\Phi}} := \begin{pmatrix}
    0 & 1 \\
    1 & 0
\end{pmatrix} \boldsymbol{\Phi}$, while $\mathrm{P}_0 := P_0 \mathbb{I}_{\mathbb{R}^2}$.
\end{remark}

\noindent Notice that, upon taking the functional derivative of the full action with respect to the field $\tilde{\varphi}$, we obtain the equation of motion for $\varphi$, \textit{i.e.},
\begin{equation}
    \label{Eq: dynamics phi}
    P_0 \varphi + \frac{\lambda}{2} \varphi^2 - 2D \tilde{\varphi} = 0.
\end{equation}
The final term, which depends explicitly on the coefficient $D$, preserves a structure analogous to that induced by noise correlations. Nevertheless, in the Lorentzian formulation, it does not admit a genuine stochastic interpretation, as no underlying noise is present. 
Observe, in addition, that one can associate to the auxiliary field $\tilde{\varphi}$ a fictitious dynamics by considering the functional derivative of $I[\varphi, \tilde{\varphi}]$ with respect to the dynamical field $\varphi$, namely,  
\begin{equation}
\label{Eq: dynamics varphi}
    P_0 \tilde{\varphi} + \lambda \varphi \tilde{\varphi} = 0. 
\end{equation}
This shows that the fictitious free dynamics of $\tilde{\varphi}$ is still ruled by the Klein-Gordon operator $P_0$ with the \emph{same} mass term for both fields. 

Since several aspects of the ensuing discussion parallel those of Section \ref{Sec: two scalar fields}, we refrain from repeating the full analysis of the free and interacting theories, highlighting only the key differences. We denote by $\mathcal{E}_2(\mathbb{R}^d) \ni (\varphi, \tilde{\varphi})$ the configuration space and by $\mathcal{F}(\mathbb{R}^d)$ the space of functionals as per Definition \ref{Def: space of functionals 2 scalar fields}. This space can be equipped with a notion of Fréchet derivative as per Definition \ref{Def: Frechet derivatives}, where now the first index refers to functional derivatives taken with respect to $\varphi$ and the second one to those with respect to $\tilde{\varphi}$. Being the interaction term polynomial in the fields, it suffices to confine our attention to polynomial functionals, \textit{i.e.}, functionals with only a finite number of non-vanishing derivatives. In the same fashion as in Section \ref{Sec: two scalar fields}, we can also identify notable subspaces of $\mathcal{F}(\mathcal{M})$, which will play a crucial r\^ole in the following. Since the definitions of \emph{microcausal}, \emph{regular} and \emph{local functionals} remain exactly the same in this setting, we limit ourselves to give notable examples of elements lying in these classes.      



\begin{example}
\label{Ex: local functionals}
For future convenience, let us discuss some notable examples of polynomial functionals. For any pair of field configurations $(\varphi, \tilde{\varphi}) \in \mathcal{E}_2(\mathbb{R}^d)$ and for any test function $f \in \mathcal{D}(\mathbb{R}^d)$, we define
\begin{flalign}
    \boldsymbol{1} [f] (\varphi, \tilde{\varphi}) &:= \int_{\mathbb{R}^d} dx \, f(x), \\
    \label{Eq: Phi}
    \Phi [f] (\varphi, \tilde{\varphi}) &:= \int_{\mathbb{R}^d} dx \, \varphi(x) f(x), \\ 
    \label{Eq: tilde Phi}
     \tilde{\Phi}[f] (\varphi, \tilde{\varphi}) &:= \int_{\mathbb{R}^d} dx \, \tilde{\varphi}(x) f(x).
\end{flalign}
Observe that $\Phi^{(1,0)} [f] (\varphi_1; \varphi, \tilde{\varphi}) = \int_{\mathbb{R}^d} dx \, \varphi_1(x) f(x)$, all the higher functional derivatives being vanishing. Therefore, since $\text{supp}(\delta) \subset \text{Diag} (\mathbb{R}^d)$, we conclude that $\Phi_f \in \mathcal{F}_{loc} (\mathbb{R}^d)$. A similar argument entails that also $\tilde{\Phi}_f \in \mathcal{F}_{loc}(\mathbb{R}^d)$. Notice, in addition, that the functional defined by
\begin{equation}
   \label{Eq: Phi tilde Phi}
   (\Phi\tilde{\Phi})[f] (\varphi, \tilde{\varphi}) = \int_{\mathbb{R}^d} dx \, \varphi(x) \tilde{\varphi}(x) f(x). 
\end{equation}
is local as well.  
\end{example}


\noindent In a similar fashion as in Section \ref{Sec: two scalar fields}, we define the \emph{classical microcausal algebra} of functionals as $\mathcal{A}_{cl} (\mathbb{R}^d) := (\mathcal{F}_{\mu c} (\mathbb{R}^d), \cdot, \ast)$, where $\cdot$ denotes the pointwise product as per Equation \eqref{Eq: pointwise prod}. In order to encode into the algebraic structure the correlations of the underlying theory, we need to implement a deformation quantization procedure. More precisely, we want to allow only for contractions between the dynamical field configurations $\varphi$ and the auxiliary field $\tilde{\varphi}$, so to account for stochastic fluctuations. To this avail, we introduce a deformation map of the form
\begin{equation}
\label{Eq: deformation map MSR}
    D_{\hbar \Delta_{+}} := \langle \hbar \Delta_{+}, \frac{\delta}{\delta \varphi} \otimes \frac{\delta}{\delta \tilde{\varphi}} \rangle = \hbar \int_{\mathbb{R}^{2d}} dx dy \, \Delta_{+}(x,y)  \frac{\delta}{\delta \varphi(x)} \otimes \frac{\delta}{\delta \tilde{\varphi}(y)},
\end{equation}
where $\Delta_+ \in \mathcal{D}'(\mathbb{R}^d \times \mathbb{R}^d)$ is the Hadamard two-point function associated to the Klein-Gordon operator $P_0$. Then, for any pair of microcausal functionals $F,G \in \mathcal{F}_{\mu c} (\mathbb{R}^d)$, one can define their deformed product as in Equation \eqref{Eq: deformed product 2 scalar fields}, with the the deformation map chosen as per Equation \eqref{Eq: deformation map MSR}. For the sake of doing computations, it is convenient to expand the exponential in a formal power series and write: 
\begin{equation}
    \label{Eq: MSR deformed product}
    F \star G := F \cdot G + \sum_{k, \tilde{k} = 1}^{\infty} \frac{\hbar^{(k+\tilde{k})}}{k! \tilde{k}!} \langle F^{(k, \tilde{k})}, (\Delta_+)^{\otimes (k+\tilde{k})} G^{(\tilde{k}, k)} \rangle.
\end{equation}

\begin{remark}
\label{Rmk: commutativity}
    Note that the deformed product in Equation \eqref{Eq: MSR deformed product} is symmetric under exchange of $F$ and $G$, \textit{i.e.}, $F \star G = G \star F$ for any $F, G \in \mathcal{F}_{\mu c} (\mathbb{R}^d)$. This is a by-product of the symmetry of Equation \eqref{Eq: MSR deformed product} with respect to $k, \tilde{k}$ and it naturally codifies the commutativity of the fields $\varphi$ and $\tilde{\varphi}$. 
\end{remark}

\begin{example}
Let us consider the local functionals introduced in Example \ref{Ex: local functionals}. We compute the deformed product between $\Phi_f$ and $\tilde{\Phi}_f$, as it will play a fundamental r\^ole in the interacting scenario: for any pair of field configurations $(\varphi, \tilde{\varphi}) \in \mathcal{E}_2(\mathbb{R}^d)$ and for any test function $f \in \mathcal{D}(\mathbb{R}^d)$  
\begin{equation*}
    (\Phi [f] \star \tilde{\Phi} [f']) (\varphi, \tilde{\varphi}) = (\Phi [f] \cdot \tilde{\Phi}[f'])(\varphi, \tilde{\varphi}) + \hbar \int_{\mathbb{R}^{2d}} dx dy \, \Delta_+(x,y) f(x) f'(y). 
\end{equation*}
Since the deformation map is built in such a way to encode the commutativity of the underlying theory -- see Remark \ref{Rmk: commutativity}, it descends that 
\begin{equation*}
     (\tilde{\Phi}[f'] \star \Phi[f]) (\varphi, \tilde{\varphi}) = (\Phi[f] \star \tilde{\Phi}[f']) (\varphi, \tilde{\varphi}), 
\end{equation*}
as one can directly check by applying Equation \eqref{Eq: MSR deformed product}. 
\end{example}

As for the case of two scalar fields, one can consider a bi-distribution $H \in \mathcal{D}'(\mathbb{R}^d \times \mathbb{R}^d)$, which coincides with the Hadamard parametrix up to a smooth remainder, and define the local and covariant algebra of microcausal polynomial functionals with respect to the deformed product $\star_H$. The abstract Wick ordering of any local observable is then defined according to Equation \eqref{Eq: alpha H}, barring the replacement of $\text{tr}_{\mathbb{C}^2}(\boldsymbol{H})$ with the Hadamard parametrix $H$ of $P_0$.  
Consider now an interaction potential of the form 
\begin{equation}
    \label{Eq: V MSR}
    V[h] (\varphi, \tilde{\varphi}) := \int_{\mathbb{R}^d} dx \, \left(\frac{\lambda}{2} \varphi^2(x) \tilde{\varphi}(x) + - D\tilde{\varphi}^2(x) \right) h(x), \quad h \in \mathcal{D}(\mathbb{R}^d). 
\end{equation}
\begin{remark}
\label{Rmk: joint exp}
    Note that the potential in Equation \eqref{Eq: V MSR} depends on two coupling parameters: the interaction parameter $\lambda \in \mathbb{R}_+$ and the diffusion coefficient $D \in \mathbb{R}$. Accordingly, all perturbative expansions are to be understood as joint expansions in both parameters.
\end{remark}
For any pair of local, polynomial functionals $F, G \in \mathcal{F}_{loc}(\mathbb{R}^d)$ we define their time-order product via the formal power series 
\begin{equation}
\label{Eq: time-ordering MSR}
    F \cdot_T G := F \cdot G + \sum_{k, \tilde{k} = 1}^\infty \frac{\hbar^{k+\tilde{k}}}{k! \tilde{k}!}  \langle F^{(k, \tilde{k})}, \Delta_{F}^{\otimes (k+\tilde{k})} G^{(\tilde{k}, k)} \rangle, 
\end{equation}
where we omit the field configurations for the sake of notational ease. In Equation \eqref{Eq: time-ordering MSR}, $\Delta_F \in \mathcal{D}'(\mathbb{R}^d \times \mathbb{R}^d)$ is the Feynman propagator of the free theory. Aside from the different form of the interacting potential, the $S-$matrix and the Bogoliubov map take a form analogous to that given in Equations \eqref{Eq: S-Matrix} and \eqref{Eq: Bogoliubov map}. To further clarify  the action of the Bogoliubov map in this setting, we provide an illustrative example.

\begin{example}
 Let us compute the interacting counterpart of the functionals in Example \ref{Ex: local functionals} up to second order \emph{jointly} in the couplings $\underline{\lambda} := (\lambda, D)$ as clarified in Remark \ref{Rmk: joint exp}. Under the same {\it caveat} discussed in Example \ref{Ex: interacting fields}, to recover the expectation value of the interacting observables it suffices to evaluate them at vanishing field configurations. Also in this case, only terms quadratic in the couplings may survive upon contraction, as the field content of the zeroth and first order contributions does not allow for maximally contracted diagrams. A direct inspection entails that due to the explicit structure of the interaction, functionals containing an equal number of dynamical and auxiliary fields have vanishing expectation. Therefore, we can confine our attention to the functionals $\Phi[f]$ and $\tilde{\Phi}[f]$, which are linear in the fields. 
 \begin{flalign}
     R_V(\Phi[f]) \vert_{\substack{\varphi =0\\ \tilde{\varphi} = 0}} &=_{\mathcal{O}(\underline{\lambda}^3)} \frac{1}{2 \hbar^2} [ 2 V[h] \star_H (V[h] \cdot_T \Phi[f]) - (V[h] \cdot_T V[h]) \cdot_T \Phi[f] \\ &- (V[h] \cdot_{AT} \Phi[f]) - (V[h] \cdot_T V[h]) \star_H \Phi[f]],
 \end{flalign}
and, analogously, for $\tilde{\Phi}[f]$. Let us examine separately the three contributions in the previous expression: 
 \begin{itemize}
     \item[\ding{104}] The first term yields the following maximally contracted diagrams: 
   \begin{equation*}
        \wick{
        \c1 \Phi [h] \c2 \Phi [h] \c3 {\tilde{\Phi}} [h] \, \, \;
        (\c1 {\tilde{\Phi}} [h] \c2 {\tilde{\Phi}} [h] \, \; \c3 \Phi [f])} \propto \lambda D \int_{\mathbb{R}^{3d}} dx \, dy \, dz \, H^2(x,y) H(x,z) h(x) h(y) f(z), 
    \end{equation*}
       \begin{equation*}
        \wick{
        \c1 {\tilde{\Phi}} [h] \c2 {\tilde{\Phi}} [h] \, \, \;
        (\c1 \Phi [h] \c2 \Phi [h] \c3 {\tilde{\Phi}} [h] \, \; \c3 \Phi [f])} \propto \lambda D \int_{\mathbb{R}^{3d}} dx \, dy \, dz \, H^2(x,y) H_F(y,z) h(x) h(y) f(z).  
    \end{equation*}
\item[\ding{104}] The analysis of the second and third contribution in the expansion of the interacting fields up to $\mathcal{O}(\underline{\lambda}^3)$ is analogous, barring minor modifications. Thus, we report here only the maximally contracted diagrams stemming for the term $(V[h] \cdot_T V[h]) \cdot_T \Phi[f]$: 
  \begin{equation*}
        \wick{
        (\c1 \Phi [h] \c2 \Phi [h] \c3 {\tilde{\Phi}} [h] \, \;
        \c1 {\tilde{\Phi}} [h] \c2 {\tilde{\Phi}} [h]) \, \, \; \c3 \Phi [f]} \propto \lambda D \int_{\mathbb{R}^{3d}} dx \, dy \, dz \, H_F^2(x,y) H_F(x,z) h(x) h(y) f(z), 
    \end{equation*}
       \begin{equation*}
        \wick{
        (\c1 {\tilde{\Phi}} [h] \c2 {\tilde{\Phi}} [h] \, \;
        \c1 \Phi [h] \c2 \Phi [h] \c3 {\tilde{\Phi}} [h]) \, \, \; \c3 \Phi [f]} \propto \lambda D \int_{\mathbb{R}^{3d}} dx \, dy \, dz \, H_F^2(x,y) H_F(y,z) h(x) h(y) f(z).  
    \end{equation*}

 \end{itemize}
The fact that the interacting fields have non-zero expectations could have been deduced directly from Equation \eqref{Eq: MSR SPDE}. Indeed, adding a stochastic current to the free equation of motion of a scalar field modifies the one-point function of the underlying theory. As a result, even selecting a Gaussian, quasi-free state, $\omega_1$ no longer vanishes and it contributes to the two-point function. 
\end{example}

\paragraph{Wetterich equation in the Local Potential Approximation}
\label{Sec: Wetterich equation MSR} -- In this paragraph, we shall derive the Renormalization Group flow equations in the local potential approximation for a quantum field theory described by Equation \eqref{Eq: MSR action}. The strategy is similar to the one detailed for the two scalar fields, albeit with some key differences we succinctly highlight in the following. The form of the deformation map - see Equation \eqref{Eq: deformation map MSR} -- suggests a natural choice for the infrared regulator: 
\begin{equation}
\label{Eq: Qk MSR}
Q_k(\chi, \tilde{\chi}) = - \frac{1}{2} \int_{\mathbb{R}^d} dx \, q_k(x) \chi(x) \tilde{\chi}(x),
\end{equation}
where the factor $\frac{1}{2}$ is added to account for the commutativity of the fields. Without loss of generality, we can fix $q_k(x) := k^2 f(x)$, with $f \in C^{\infty}_0(\mathbb{R}^d)$ with support contained in large regions of spacetime. The regularised generating functional of the theory thus reads: 
\begin{equation}
    \label{Eq: Zk MSR}
    Z_k(j, \tilde{j}) := \omega(S(V)^{-1} \star S(V+J+Q_k)),
\end{equation}
where $$J(\chi, \tilde{\chi}) := \int_{\mathbb{R}^d} dx \, [\tilde{j}(x) \chi(x) + j(x) \tilde{\chi}(x)], \quad j, \tilde{j} \in C^{\infty}_0(\mathbb{R}^d)$$ is an external current, while $\omega$ is a state on the quantum algebra, not necessarily quasi-free. The connected counterpart of the generating functional $Z_k$ is given by
\begin{equation}
\label{Eq: Wk}
    W_k (j, \tilde{j}) := -i \log(Z_k(j, \tilde{j}),
\end{equation}
its first-order derivatives yielding the classical fields, \textit{i.e.}, 
\begin{equation}
    \label{Eq: classical fields}
    \frac{\delta W_k}{\delta \tilde{j}(x)} =: \varphi(x) \quad \text{and} \quad \frac{\delta W_k}{\delta j(x)} =: \tilde{\varphi}(x).
\end{equation}
The effective action and average effective action are
\begin{equation}
\label{Eq: Gammak}
    \tilde{\Gamma}_k (\varphi, \tilde{\varphi}) := W_k(j, \tilde{j}) - J(j_{\tilde{\varphi}}, \tilde{j}_{\varphi}), \quad \Gamma_k(\varphi, \tilde{\varphi})  = \tilde{\Gamma}_k(\varphi, \tilde{\varphi}) - Q_k(\varphi, \tilde{\varphi}), 
\end{equation}
with $j_{\tilde{\varphi}}$ and $\tilde{j}_{\varphi}$ the currents which solve Equation \eqref{Eq: classical fields} as a function of the classical fields. In the present setting, the \textbf{Polchinski equation} is of the form: 
\begin{equation}
    \label{Eq: Polchinski equation MSR}
     \partial_k W_k(j, \tilde{j}) = - \frac{1}{2} \int_{\mcM} d\mu_x \, \frac{1}{Z_k(j, \tilde{j})} \partial_k q_{k}(x) \omega\left(S(V)^{-1} \star (S(V + J + Q_k) \cdot_{T} T\eta \tilde{\eta})(x)\right),  
\end{equation}
while the \textbf{Wetterich equation} reads: 
\begin{equation}
    \label{Eq: Wetterich equation MSR}
     \partial_k W_k(j, \tilde{j}) = - \frac{1}{2} \int_{\mcM} d\mu_x \, \frac{1}{Z_k(j, \tilde{j})} \partial_k q_{k}(x) \omega\left(S(V)^{-1} \star (S(V + J + Q_k) \cdot_{T} T\eta \tilde{\eta})(x) - \varphi(x) \tilde{\varphi}(x)\right).  
\end{equation}
To find approximate local solutions to Equation \eqref{Eq: Wetterich equation MSR}, we work in the LPA scheme, detailed in Paragraph \ref{Sec: Wetterich for two scalar fields}. The key difference with respect to the analysis carried out for the two scalar fields is that the effective mass of the theory is now of the form
\begin{equation*}
    M := I_0^t -I_0 = -\int_{\mathbb{R}^d} dx \, \left( \frac{\delta^2 U_{k}}{\delta \varphi \delta \tilde{\varphi}} (\varphi_{cl}, \tilde{\varphi}_{cl}) - m^2\right)\varphi(x) \tilde{\varphi}(x)
\end{equation*}
in which only the mixed second-order derivatives of the effective potential appear. As before, in order to obtain explicit expressions, we specify the analysis on $4-$dimensional Minkowski spacetime $(\mathbb{R}^4, \eta)$ and we select as $\omega$ the vacuum state at zero temperature. Furthermore, we make the following ansatz on the structure of the effective potential $U_k$, namely, we assume that 
\begin{equation}
    \label{Eq: Uk MSR}
    U_k(\varphi, \tilde{\varphi}) = U_{k,0} + m^2_k \varphi \tilde{\varphi} + \frac{\lambda_k}{2} \varphi^2 \tilde{\varphi} - D_k \tilde{\varphi}^2. 
\end{equation}
Following slavishly the steps detailed in Paragraph \ref{Sec: Wetterich for two scalar fields}, the Wetterich equation can be recast in the form of an ordinary differential equation for $U_k$ with respect to the scale parameter $k$, \textit{i.e.}, 
\begin{equation}
    \label{Eq: partial Uk MSR}
    \partial_k U_k(\varphi, \tilde{\varphi}) = \frac{k}{8 \pi^2} (k^2 + m_k^2 + \frac{\lambda_k}{2} \varphi) \log \left(\frac{k^2 + m_k^2 + \frac{\lambda_k}{2} \varphi}{\mu^2} \right).
\end{equation}

\begin{remark}
   The regularised mass term in Equation \eqref{Eq: partial Uk MSR} is independent of both the auxiliary field $\tilde{\varphi}$ and of the diffusion coefficient $D_k$. This feature is not accidental, but rather reflects the structural form of the MSR action -- see Equation \eqref{Eq: MSR action}, where mass-like contributions arise solely in the mixed sector proportional to $\varphi \tilde{\varphi}$. Within the local potential approximation, this implies that the second functional derivatives of the effective action and, hence, the regulator-modified inverse propagator, depend only on the field $\varphi$. Although alternative choices of regularisation could in principle be considered, the present one preserves the causal and response-field structure of the theory, ensuring that the auxiliary field $\tilde{\varphi}$ retains its r\^ole as a response variable, while stochastic effects enter exclusively through correlation functions. 
\end{remark}

From Equation \eqref{Eq: partial Uk MSR}, one can directly deduce the beta functions for the relevant couplings: 
\begin{flalign}
\label{Eq: Uk0 MSR}
k \partial_k U_{k,0} &= \frac{k^2}{8 \pi^2} (k^2 +m_k^2) \log \left( \frac{k^2 + m_k^2}{\mu^2} \right), \\
\label{Eq: mk2 MSR}
k \partial_k m^2_k &= \frac{k^2}{8 \pi^2} \frac{\lambda_k}{2} \left[ \log \left( \frac{k^2 + m_k^2}{\mu^2} + 1\right)\right],\\
\label{Eq: Dk MSR}
k \partial_k D_k &= 0,\\
\label{Eq: lambdak MSR}
k \partial_k \lambda_k &= \frac{k^2}{4 \pi^2} \frac{\lambda_k^2}{4} \frac{1}{k^2 + m_k^2}.
\end{flalign}
The auxiliary field $\tilde\varphi$ does not play the r\^ole of an independent running coupling in the flow equations above. Rather, the beta 
functions in Equations \eqref{Eq: Uk0 MSR}, \eqref{Eq: mk2 MSR}, \eqref{Eq: Dk MSR} and \eqref{Eq: lambdak MSR} are obtained by projecting the Wetterich equation onto the corresponding monomials in the fields. In particular, the equations for $m_k^2$, $D_k$ and $\lambda_k$ are read off from the terms proportional to $\tilde\varphi$. Hence, evaluating the auxiliary field on a constant 
configuration, and in particular setting $\tilde\varphi=1$, only removes an overall $k$-independent factor and does not modify the asymptotic behaviour of 
the running couplings. The first equation, governing the field-independent part $U_{k,0}$ of the effective potential, is by construction independent of 
$\tilde\varphi$ and is therefore unaffected by this projection.

\begin{remark}
The absence of a non-trivial flow for the noise amplitude $D_k$ is a direct consequence of the Martin-Siggia-Rose structure of the theory. Indeed, the diffusion coefficient multiplies a purely quadratic term in the auxiliary field, whilst all interaction vertices remain linear in the physical field $\varphi$. As a result, within the present truncation, no loop diagrams can generate contributions proportional to $\tilde{\varphi}^2$, owing also to the causal structure of the propagators. As a consequence, the RG flow does not induce any renormalisation of the noise sector, and $D_k$ remains scale-independent. This property is stable under truncations preserving the MSR structure, but it may fail in the presence of multiplicative noise or higher-order couplings. This is consistent with the results obtained in \cite{Duch_2025}, albeit via a different approach.  
\end{remark}

In order to directly compare our results with those derived in \cite{Duch_2025}, we confine the analysis to the three-dimensional setting, namely, to $(\mathbb{R}^3, \eta)$, and we require the reference state $\omega$ to be the vacuum state at zero temperature. By \cite[Eq. (95)]{Brunetti_2009}, we have that the Hadamard function in the odd-dimensional setting can be written at the level of integral kernels as 
\begin{equation*}
    H_{M_k}(x) = \frac{1}{4 \sin (\frac{d}{2}-1) \pi} (2 \pi)^{\frac{2-d}{2}} {M_k}^{\frac{d}{2}- 1} |x^2|^{\frac{2-d}{4}} I_{\frac{d}{2}-1}(\sqrt{M_k^2 |x^2|}),
\end{equation*}
for $x^2 < 0$ and $d \ge 3$. Recalling that $I_\nu(y) \sim \frac{1}{\Gamma(\nu +1)} \left(\frac{y}{2}\right)^\nu$, with $F$ an entire analytic function, $H_M$ is smooth as a function of $M_k^2$. 
In particular, setting $d=3$, we can write $$\partial_k U_k(\varphi, \tilde{\varphi}) = \lim_{y \rightarrow x} H_{M_k}(x - y) =  \frac{k}{4\pi} M_k, \quad \text{with} \quad M_k = k^2 + m_k^2 + \frac{\lambda_k}{2} \varphi,$$ where $U_k$ is chosen as in Equation \eqref{Eq: Uk MSR}. Eliminating the auxiliary field by evaluating the final expressions at $\tilde{\varphi} = 1$, we thus obtain the beta functions for the relevant couplings: 
\begin{flalign}
    \label{Eq: 3D partial Uk0 MSR}
    k \partial_k U_{k,0} &=  \frac{k^2}{4\pi} (k^2 + m_k^2), \\
    \label{Eq: 3D mk}
    k \partial_k m_k^2 &= \frac{k^2}{8\pi} \lambda_k, \\
    \label{Eq: 3D Dk}
    k \partial_k D_k &= 0, \quad \text{no scaling of the noise amplitude} \\
    \label{Eq: 3D lambda k}
    k \partial_k \lambda_k &= 0,  \quad \text{no scaling of the coupling constant}
\end{flalign}
These results are in agreement with \cite{Duch_2025}. The perturbative approach adopted in the present work, however, has the additional advantage that its validity extends beyond the subcritical regime, \textit{i.e.}, for dimensions $d > 3$, thereby enabling the study of the scaling behaviour of the relevant couplings on higher-dimensional backgrounds.  

\subsection{(B) Self-interacting Dirac fields}
\label{Sec: Thirring model}
In this section, we discuss the algebraic quantum field theoretical framework applied to an interacting Fermionic model, \textit{i.e.}, the Thirring model in dimension $d \ge 2$. The present approach differs from earlier work, \textit{e.g.}, \cite{Rejzner_2011}, in that the anticommutativity of Fermionic fields is implemented at the quantum level via the deformation map, rather than by introducing Grassmann variables. The present approach differs from earlier treatments, such as \cite{Rejzner_2011}, in that no Grassmann-valued classical fields are introduced. Furthermore, we shall regard spinor and cospinor fields as independent field variables throughout the initial construction, imposing that they are related only at a later stage. This choice is particularly convenient for the perturbative algebraic formulation adopted below. The latter draws mainly on \cite{BonicelliCosteriDappiaggiRinaldi2024}, where an analogous framework is developed for the Thirring model, aside from the stochastic aspects considered therein, and on \cite{DappiaggiHackPinamonti2009}, where the algebraic quantization of the free Dirac field on arbitrary globally hyperbolic Lorentzian spacetimes is presented in detail.

Henceforth, we shall assume that the background $\mcM$ is spin, \textit{i.e.}, that its second Stiefel–Whitney class vanishes. This ensures the existence of at least one admissible spin structure on $\mcM$, see \textit{e.g.} \cite{LawsonMichelsohn1989}. To guarantee uniqueness, we further require that the spin bundle is trivial. Under these assumptions, the spinor and cospinor fields are smooth sections of the Dirac bundle and of its dual respectively, namely $\psi \in \Gamma^{\infty}(D\mcM)$ and $\bar{\psi} \in \Gamma^{\infty}(D^*\mcM)$. Under the triviality assumptions, we have that $\psi, \bar{\psi} \in C^{\infty}(\mcM; \Sigma_d) \simeq C^{\infty}(\mcM) \otimes \Sigma_d$, where $\Sigma_d = \mathbb{C}^{N_d}$ with $N_d := 2^{\lfloor \frac{d}{2} \rfloor}$ is the \emph{spinor space}. 

\begin{notation}
\label{Not: Dirac doublet}
    Also in this case we can introduce a convenient matrix notation -- see Remark \ref{Rmk: field doublet}. We define a field doublet $\Psi := \begin{pmatrix}
        \psi \\
        \bar{\psi}
    \end{pmatrix} \in \Gamma^{\infty}(D\mcM \oplus D^* \mcM) \simeq C^{\infty}(\mcM; \Sigma_d) \otimes \mathbb{C}^2 \simeq C^{\infty}(\mcM) \otimes \mathbb{C}^{N_d} \otimes \mathbb{C}^2 $, where we use the fact that $\Sigma_d \simeq \Sigma_d^*$. Henceforth, we shall carefully distinguish between the $\mathbb{C}^2-$indices, denoted by Latin letters, and the Dirac indices, denoted by Greek letters. For instance, $$\Psi_{i} = \begin{cases}
    \psi^{\rho} \quad \text{if} \, i = 1 \\ \bar{\psi}_{\rho} \quad \text{if} \, i = 2 \end{cases} \rho \in \{1, \ldots, N_d\}.$$ 
\end{notation}

A key ingredient in the dynamical descriptions of Fermions are the gamma fields $\gamma \in \Gamma(T^*\mcM \otimes \text{End}(S\mcM)) \simeq \Gamma^{\infty}(T^*\mcM \otimes \text{End}(\mathbb{C}^{N_d}))$. Locally, fixing a global frame $\{e_a\}_{a = 1, \ldots, N_d}$ of $S\mcM$, they can be written in coordinate expression as $\gamma^{\mu}(x) = e_a^{\mu}(x) \gamma(e_a)$, where $\gamma^{\mu}(x) \in \text{End}(\mathbb{C}^{N_d})$ are the Dirac gamma matrices. In this context, the operators ruling the dynamics are the Dirac operator and its formal adjoint, which are explicitly given by
\begin{flalign*}
    \slashed{D}^{\rho}_{\rho'} &:= (i (\gamma^{\mu})^{\rho}_{\rho'} \nabla_\mu - m \mathbb{I}^{\rho}_{\rho'}): \Gamma^{\infty}(D\mcM; \Sigma_d) \rightarrow \Gamma^{\infty}(D\mcM; \Sigma_d), \\
     (\slashed{D}^*)^{\rho}_{\rho'} &:= (i (\gamma^{\mu})^{\rho}_{\rho'} \nabla_\mu +  m \mathbb{I}^{\rho}_{\rho'}): \Gamma^{\infty}(D^*\mcM; \Sigma^*_d) \rightarrow \Gamma^{\infty}(D^*\mcM; \Sigma^*_d).
\end{flalign*}
These operators are Green hyperbolic and, hence, they admit unique advanced and retarded fundamental solutions $\Delta^{\psi}_{A/R}, \Delta^{\bar{\psi}}_{A/R} \in \mathcal{D}'(\mcM \times \mcM;\Sigma_d \otimes \Sigma_d^*)$, which can be built out of the ones for the Klein-Gordon operator via the Lichnerowicz–Weitzenb\"ock formula -- see \cite{LawsonMichelsohn1989}. Concretely, if $\Delta_{A/R} \in \mathcal{D}'(\mcM \times \mcM)$ denote the advanced and retarded propagators for a real, scalar field on $(\mcM, g_{\mcM})$, $\Delta^{\psi}_{A/R} = \slashed{D}^* \Delta_{A/R}$ and $\Delta^{\bar{\psi}}_{A/R} = \slashed{D} \Delta_{A/R}$ are those associated to the Dirac operator and to its formal adjoint, respectively. The Hadamard two-point functions for the Dirac field can thus be obtained from a scalar Hadamard two-point function $\Delta_+$ by applying the Dirac operator and inserting the appropriate spinorial parallel transport. We shall denote these by 
\begin{equation*}
    (\Delta^{\psi}_+)^{\rho}_{\rho'}, (\Delta^{\bar{\psi}}_+)^{\rho}_{\rho'} \in \mathcal{D}'(\mcM \times \mcM; \Sigma_d \otimes \Sigma_d^*) \simeq \mathcal{D}'(\mcM \times \mcM; \mathbb{C}^{2N_d}) ,
\end{equation*}
or we can equivalently write, in a more compact matrix notation 
\begin{equation*}
    \boldsymbol{\Delta}_+^{\Psi} := \begin{pmatrix}
        \Delta^{\psi}_+ & 0 \\
        0 & -\Delta^{\bar{\psi}}_+
    \end{pmatrix} \in \mathcal{D}'(\mcM \times \mcM, \text{End}(\mathbb{C}^{2N_d})),
\end{equation*}
where the minus sign is included for later convenience. 

The full action of the Thirring theory is thus specified by $I_0[\psi, \bar{\psi}] = I_0[\psi] + I_0[\bar{\psi}]$, with
\begin{flalign}
    \label{Eq: Thirring action}
    I_0[\psi] &= - \int_{\mcM} d\mu_x \, \left[ i \bar{\psi}_{\rho}(x) (\gamma^{\mu})^{\rho}_{\rho'} (x)  \nabla_{\mu} \psi^{\rho'}(x) - m \bar{\psi}_{\rho}(x) \psi^{\rho}(x) - \frac{\lambda}{2} (\bar{\psi} \gamma^{\mu} \psi)(x) (\bar{\psi} \gamma_{\mu} \psi)(x) \right] f(x), \\
     I_0[\bar{\psi}] &= - \int_{\mcM} d\mu_x \, \left[ i \psi^{\rho}(x) (\gamma^{\mu})^{\rho'}_{\rho} (x)  \nabla_{\mu} \bar{\psi}_{\rho'}(x) - m \bar{\psi}_{\rho}(x) \psi^{\rho}(x) - \frac{\lambda}{2} (\bar{\psi} \gamma^{\mu} \psi)(x) (\bar{\psi} \gamma_{\mu} \psi) (x) \right] f(x).
\end{flalign}
Here $\lambda \in \mathbb{R}$ is the coupling constant, while $f \in \mathcal{D}(\mcM)$ is a smooth, compactly supported cut off.

For the sake of notational brevity, let us denote by $V := \Sigma_d$ and let us introduce the vector space $W:= \bigoplus_{k, k' \ge 0} (V^{\otimes k} \otimes (V^*)^{\otimes k'})$ with $V^{\otimes 0} = (V^*)^{\otimes 0} = \mathbb{C}$ by convention. In the present setting, the configuration space is the Fréchet space $\mathcal{E}_V(\mcM) := \Gamma^{\infty}(D\mcM \oplus D^* \mcM) \simeq C^{\infty}(\mcM; V) \oplus C^{\infty}(\mcM; V),$ endowed with the product topology. Furthermore, we denote by $\mathcal{F}_W(\mcM)$ the space of functionals $u: \Gamma^{\infty}(D\mcM \oplus D^* \mcM) \rightarrow W$, which are smooth in the Bastiani sense -- see. 

\begin{remark}
\label{Rmk: functional valued distr}
To be precise, such $u$ should be understood as a functional-valued distribution, {\it i.e.}, $u \in \mathcal{D}'(\mathcal{M}; \mathcal{F}_W(\mathcal{M}))$. With a slight abuse of notation, we may still write $u \in \mathcal{F}_W(\mathcal{M})$, when referring to the evaluated functional $u(f; \cdot, \cdot)$ for a given but fixed test function $f \in \mathcal{D}(\mcM)$.
\end{remark}

\begin{remark}
    \label{Rmk: components of u}
    Bearing the above remark in mind, let  $\{e_{\rho}\}_{\rho = 1, \ldots, N_d}$ be an arbitrary basis of $V$ with dual counterpart given by $\{e^{\rho}\}_{\rho=1, \ldots, N_d}$. Fixed $f \in \mathcal{D}(\mcM)$, for any pair of field configurations $(\psi, \bar{\psi}) \in \Gamma^{\infty}(D\mcM) \times \Gamma^{\infty}(D^*\mcM)$, we can define the component of $u$ with values in $W^{(p,q)} = V^{\otimes p} \otimes (V^*)^{\otimes q}$ as 
\begin{equation}
    \label{Eq: component of Dirac functionals}
    u^{\lambda_1 \ldots \lambda_p}_{\rho_1 \ldots \rho_q} (f; \psi, \bar{\psi}) := \langle u(f; \psi, \bar{\psi}), e^{\lambda_1} \otimes \ldots \otimes e^{\lambda_p} \otimes e_{\rho_1} \otimes \ldots \otimes e_{\rho_q} \rangle_V, 
\end{equation}
where we indicate by $\langle \cdot, \cdot \rangle_V$ the duality pairing between $V$ and $V^*$. It descends that, exploiting the natural embedding $W^{(p,q}) \hookrightarrow W$ for $p, q \ge 0$, any $u \in \mathcal{D}'(\mcM; \mathcal{F}_W(\mcM))$ can be decomposed as 
\begin{equation}
    \label{Eq: decomposition of Dirac functionals}
    u = \sum_{p,q \ge 0} u^{\lambda_1 \ldots \lambda_p}_{\rho_1 \ldots \rho_q} e_{\lambda_1} \otimes \ldots \otimes e_{\lambda_p} \otimes e^{\rho_1} \otimes \ldots \otimes e^{\rho_q}, 
\end{equation}
where the right hand-side is as per Equation \eqref{Eq: component of Dirac functionals}. 
\end{remark}

\noindent The space $\mathcal{F}_W(\mcM)$ can be equipped with a notion of Fréchet derivative: for any $u \in \mathcal{F}_W(\mcM)$, we define its $(k,k')-$functional derivative as
\begin{flalign}
\notag
    u^{(k,k')}: \Gamma^{\infty}(D\mcM) \otimes \ldots \otimes \Gamma^{\infty}(D \mcM) \otimes \Gamma^{\infty} (D^* \mcM) \otimes \ldots \otimes \Gamma^{\infty}(D^* \mcM) \times \mathcal{E}_V(\mcM) &\longrightarrow W \\ \notag
    u^{(k,k')}(\eta_1 \otimes \ldots \otimes \eta_k \otimes \bar{\eta}_1 \otimes \ldots \otimes \bar{\eta}_{k'}; \psi, \bar{\psi}) &= \\ \label{Eq: Thirring fun derivatives} = \frac{\partial^{k+k'}}{\partial s_1 \ldots \partial s_k \partial \bar{s}_1 \ldots \partial \bar{s}_k} u(s_1 \eta_1 + \ldots + s_k \eta_k + \psi; \bar{s}_1 \bar{\eta}_1 + \ldots + \bar{s}_{k'} \bar{\eta}_{k'} + \bar{\psi}) &\big \vert_{\substack{s_1 = \ldots = s_k = 0 \\ \bar{s}_1 = \ldots = \bar{s}_{k'} = 0}}.
\end{flalign}
In particular, if $\exists (\tilde{k},\tilde{k}') \in \mathbb{N}_0 \times \mathbb{N}_0$ such that $u^{(k,k')} = 0$ for any $k \ge \tilde{k}$ or $k' \ge \tilde{k}'$, we say that $u$ is a polynomial functional of the fields, \textit{i.e.}, $u \in \text{Pol}_W(\mcM)$. 

\begin{remark}
\label{Rmk: wavefront set for vector-valued distributions}
Recall that for any $u \in \mathcal{D}'(\mcM, \text{Pol}_W(\mcM))$ -- \textit{cf.}, Remark \ref{Rmk: functional valued distr} --, $u(\psi, \bar{\psi}) \equiv u(\cdot, \psi, \bar{\psi})$ is a vector-valued distribution for any pair of field configurations $(\psi, \bar{\psi}) \in \mathcal{E}_V(\mcM)$. We thus consider the following non-polarised definition of the wavefront set of $u$: 
\begin{equation*}
    \text{WF}(u(\psi, \bar{\psi})) := \bigcup_{i=1}^{\text{dim} W} WF(u_i(\psi, \bar{\psi})), \quad u_i \in \mathcal{D}'(\mcM, \mathcal{F}(\mcM)), 
\end{equation*}
where $u_i$ denotes the $i-$th component of $u$ with respect to a chosen basis of $V$. 
\end{remark}

In light of the previous remark, we can introduce the following notable subspaces of $\mathcal{F}_W(\mcM)$:
\begin{itemize}
    \item[\ding{104}] We call \emph{microcausal functionals} the elements of 
    \begin{equation}
        \label{Eq: microcausal func Thirring}
        \mathcal{F}_{\mu c, W}(\mcM) := \{u \in \mathcal{F}_W(\mcM) \, | \, \text{WF}(u^{(k,k')}(\psi, \bar{\psi})) \subseteq G_{k+k'}(\mcM), \forall k,k' \in \mathbb{N}_0\}, 
    \end{equation}
    where 
    \begin{equation*}
        G_{k + k'} (\mcM) = T^*(\underbrace{\mcM \times \ldots \times \mcM}_{k+k'-\text{times}} \setminus \left( \bigcup_{x \in \mcM} (V^+_x)^{k+k'} \cup \bigcup_{x \in \mcM} (V^{-}_x)^{k+k'}\right)), 
    \end{equation*}
    $V^{\pm}_x$ denoting the future (\textit{resp}. past) pointing light cones at $x \in \mcM$. 
    \item[\ding{104}] We call \emph{regular functionals} the elements of 
     \begin{equation}
        \label{Eq: regular func Thirring}
        \mathcal{F}_{reg, W}(\mcM) := \{u \in \mathcal{F}_{\mu c, W}(\mcM) \, | \, u^{(k,k')}(\psi, \bar{\psi}) \in \mathcal{D}'(\mcM^{k+k'}), \forall k,k' \in \mathbb{N}_0\}. 
    \end{equation}
    \item[\ding{104}] Denoting by $\text{supp}(u(\psi, \bar{\psi})) := \bigcup_{i=1}^{\text{dim} W} \text{supp}(u_i(\psi, \bar{\psi}))$, we call \emph{local functionals} the elements of the space
      \begin{equation}
        \label{Eq: local func Thirring}
        \mathcal{F}_{loc, W}(\mcM) := \{u \in \mathcal{F}_{\mu c, W}(\mcM) \, | \, \text{supp}(u^{(k,k')}(\psi, \bar{\psi})) \subseteq \text{Diag}_{k+k'}(\mcM), \forall k,k' \in \mathbb{N}_0\}. 
    \end{equation}
\end{itemize}
The \emph{classical microcausal algebra} of functionals is the unital, commutative, $*-$algebra $\mathcal{A}^W_{cl}(\mcM) := (\mathcal{F}_{\mu c, W}(\mcM), \cdot, \ast)$ of microcausal functionals as per Equation \eqref{Eq: microcausal func Thirring}, endowed with the pointwise product as per Equation \eqref{Eq: pointwise prod} and with an involution given by the complex conjugation. 

\begin{example}
\label{Ex: spinor fields}
Central to the construction of the quantum theory are functionals that are linear in the fields, \textit{i.e.}, 
    \begin{flalign}
    \label{Eq: spinor functional} 
    \Psi^{\rho}[f] (\psi, \bar{\psi}) &:= \int_{\mcM} d\mu_x \, \psi^{\rho}(x) f(x), \quad \Psi := \Psi^{\rho} e_{\rho} \in \text{Pol}_V(\mcM),  \\
    \label{Eq: cospinor functional}
    \bar{\Psi}_{\rho}[f] (\psi, \bar{\psi}) &:= \int_{\mcM} d\mu_x \, \bar{\psi}_{\rho}(x) f(x), \quad \bar{\Psi} := \bar{\Psi}_{\rho} e^{\rho} \in \text{Pol}_{V^*}(\mcM),
    \end{flalign}
where $f \in \mathcal{D}(\mcM)$, $(\psi, \bar{\psi}) \in \Gamma^{\infty} (D \mcM) \times \Gamma^{\infty} (D^*\mcM)$, whilst $\{e_{\rho}\}$ and $\{e^{\rho}\}$ are as per Remark \ref{Rmk: components of u}. It follows by direct inspection that these functionals are both local and microcausal and, hence, they form the building blocks of the classical algebra $\mathcal{A}^W_{cl}$. Henceforth, we shall implicitly identify $\Psi$ (\textit{resp}. $\bar{\Psi}$) with its component $\Psi^{\rho}$ (\textit{resp}. $\bar{\Psi}_{\rho}$) and write $\Psi^{\rho} \in \text{Pol}_V(\mcM)$ (\textit{resp}. $\bar{\Psi}_{\rho} \in \text{Pol}_{V^*} (\mcM)$). 
\end{example}

Observe that, at the present stage, the classical algebra is commutative per construction. To incorporate the information on the CAR into the algebraic structure, a convenient strategy consists in selecting a deformation map which is sensitive to the ordering of fields within functionals. Consider
\begin{flalign*}
    D_{\hbar \Delta^{\psi}_+} &:= \langle \hbar \Delta^{\psi}_+, \frac{\delta}{\delta \psi} \otimes \frac{\delta}{\delta \bar{\psi}} \rangle = \hbar \int_{\mcM^2} d\mu_x d\mu_y \, (\Delta_+^{\psi})^{\rho}_{\rho'} \frac{\delta}{\delta \psi^{\rho}(x)} \otimes \frac{\delta}{\delta \bar{\psi}_{\rho'}(y)}, \\   D_{\hbar \Delta^{{\bar{\psi}}}_+} &:= \langle \hbar \Delta^{\bar{\psi}}_+, \frac{\delta}{\delta \bar{\psi}} \otimes \frac{\delta}{\delta \psi} \rangle = \hbar \int_{\mcM^2} d\mu_x d\mu_y \, (\Delta_+^{\bar{\psi}})^{\rho'}_{\rho} \frac{\delta}{\delta \bar{\psi}_{\rho}(x)} \otimes \frac{\delta}{\delta \psi^{\rho'}(y)},
\end{flalign*}
and define the deformation map as
\begin{equation}
\label{Eq: deformation map Thirring}
    D_{\hbar \text{tr}_{\mathbb{C}^2} \boldsymbol{\Delta}^{\Psi}_+} := D_{\hbar \Delta_+^{\psi}} - D_{\hbar \Delta_+^{\bar{\psi}}}.
\end{equation}
Notice that the minus sign in Equation \eqref{Eq: deformation map Thirring} naturally encodes the anticommutativity of the spinor fields, thus avoiding the need to introduce Grassmann variables. The deformed product is thus given by Equation \eqref{Eq: deformed product 2 scalar fields}, with $\boldsymbol{\Delta}_+$ replaced by $\boldsymbol{\Delta}^{\Psi}_+$ as in Equation \eqref{Eq: deformation map Thirring}. Formally, for any pair of microcausal functionals $F, G \in \mathcal{F}_{\mu c, W}(\mcM)$, we can write 
\begin{equation}
    \label{Eq: deformed product Thirring}
    F \star G = F \cdot G + \sum_{\substack{k_1, k_2 \ge 1\\ k = k_1 + k_2}} \frac{(-1)^{k_2} \hbar^{k} }{k_1! k_2!} \langle F^{(k_1, k_2)}, (\Delta_{+}^{\psi})^{\otimes k_1} \otimes (\Delta_+^{\bar{\psi}})^{\otimes k_2} G^{(k_2, k_1)} \rangle, 
\end{equation}
where the Fréchet derivatives are as per Equation \eqref{Eq: Thirring fun derivatives}. 

\begin{example}
As an example, we show that Equation \eqref{Eq: deformed product Thirring} yields the expected correlations between spinor fields. More precisely, a direct computation entails  
\begin{flalign*}
    \Psi^{\rho}[f] (\psi, \bar{\psi}) \star \bar{\Psi}_{\rho'}[f'] (\psi, \bar{\psi}) & = \Psi^{\rho}[f] (\psi, \bar{\psi}) \cdot \bar{\Psi}_{\rho'}[f'] (\psi, \bar{\psi}) + \hbar \int_{\mcM^2} d\mu_x d\mu_y \, (\Delta_+^{\psi})^{\rho}_{\rho'}(x,y) f(x) f'(y), \\
    \bar{\Psi}_{\rho'}[f'] (\psi, \bar{\psi}) \star \Psi^{\rho}[f] (\psi, \bar{\psi}) & =  \bar{\Psi}_{\rho'}[f'] (\psi, \bar{\psi}) \cdot \Psi^{\rho}[f] (\psi, \bar{\psi}) - \hbar \int_{\mcM^2} d\mu_x d\mu_y \, (\Delta_+^{\bar{\psi}})^{\rho}_{\rho'}(x,y) f(x) f'(y), 
\end{flalign*}
where $\Psi^{\rho}[f]$ and $\bar{\Psi}_{p'}[f']$ are the linear functionals in Example \ref{Ex: spinor fields}, with $f, f' \in \mathcal{D}(\mcM)$. 
Notice that, by construction, the ordering of the fields becomes relevant: whenever a cospinor field appears as the first argument of the deformation map, a minus sign arises in the correlations.
\end{example}

In this setting, consider globally-defined bi-distributions $H^{\psi}, H^{\bar{\psi}} \in \mathcal{D}'(\mcM \times \mcM, \text{End}(\mathbb{C}^{2N_d}))$ such that they locally coincide up to smoothing with the Hadamard two-point functions $\Delta_{+}^{\psi}$ and $-\Delta_+^{\bar{\psi}}$, respectively. Then, the local and covariant algebra of microcausal polynomial functionals is the anti-commutative, quantum algebra obtained by deformation with respect to $H^{\psi}$ and $H^{\bar{\psi}}$. The Wick-ordering of Fermionic observables is then implemented similarly to the scalar case -- see Section \ref{Sec: two scalar fields} and, in particular, Equation \eqref{Eq: alpha H}.

\textbf{Interacting theory}. Let us consider a self-interacting Dirac field theory on $(\mcM, g_{\mcM})$ with an interaction of the form: 
\begin{equation}
    \label{Eq: Thirring potential}
    V[h](\psi, \bar{\psi}) := \int_{\mcM} d \mu_x \, \frac{\lambda}{2} (\bar{\psi} \gamma^{\mu} \psi)(\bar{\psi} \gamma_{\mu} \psi) h(x), \quad h \in \mathcal{D}(\mcM), \lambda \in \mathbb{R}_+.  
\end{equation}
For any pair of microcausal functionals $F,G \in \mathcal{F}_{\mu c, W}(\mcM)$, we define their time-ordered product as 
\begin{equation}
    \label{Eq: Thirring time-order product}
    F \cdot_T G := F \cdot G + \sum_{\substack{k_1, k_2 \ge 1 \\k= k_1 + k_2}} \frac{(-1)^{k_2} \hbar^k}{k_1! k_2!} \langle F^{(k_1, k_2)}, \Delta_{F, \psi}^{\otimes k_1} \otimes \Delta_{F, \bar{\psi}}^{\otimes k_2} G^{(k_2, k_1)} \rangle,  
\end{equation}
where for notational ease we avoid to indicate the field configurations. Here, $\Delta_{F, \psi}, \Delta_{F, \bar{\psi}} \in \mathcal{D}'(\mcM \times \mcM; \text{End}(\mathbb{C}^{N_d}))$ denote the Feynman propagators for the Dirac operator and its adjoint, respectively. Note in addition the minus sign in Equation \eqref{Eq: Thirring time-order product}, which arises from the ordering of the spinorial fields within the functionals.    

\begin{example}
For the sake of clarity, we illustrate an explicit application of Equation \eqref{Eq: Thirring time-order product}. Consider the linear functionals $\Psi^{\rho} \in \text{Pol}_V(\mcM)$ and $\bar{\Psi}_{\rho'} \in \text{Pol}_{V^*}(\mcM)$ and let us compute 
\begin{flalign*}
    \Psi^{\rho} [f] (\psi, \bar{\psi}) \cdot_T \bar{\Psi}_{\rho'} [f'] (\psi, \bar{\psi}) &= \Psi^{\rho} [f] (\psi, \bar{\psi}) \cdot \bar{\Psi}_{\rho'} [f'] (\psi, \bar{\psi)} + \hbar \int_{\mcM^2} d\mu_x d\mu_y \, \Delta_{F, \psi}(x,y) f(x) f'(y), \\ 
    \bar{\Psi}_{\rho'} [f'] (\psi, \bar{\psi}) \cdot_T  \Psi^{\rho} [f] (\psi, \bar{\psi}) &=  \bar{\Psi}_{\rho'} [f'] (\psi, \bar{\psi})  \cdot \Psi^{\rho} [f] (\psi, \bar{\psi}) - \hbar \int_{\mcM^2} d\mu_x d\mu_y \, \Delta_{F, \bar{\psi}}(x,y) f'(x) f(y). 
\end{flalign*}
\end{example}
To define the interacting counterpart of a given vector-valued functional $F \in \text{Pol}_{loc, W}(\mcM)$, we can exploit the action of the Bogoliubov map in Equation \eqref{Eq: Bogoliubov map}, with the potential $V$ as per Equation \eqref{Eq: Thirring potential}. Interacting observables are expressed as formal power series in the coupling parameter $\lambda \in \mathbb{R}_+$, and truncating the series at a given order yields an approximation of the interacting observable to that perturbative order as illustrated by the ensuing example. 

\begin{example}
In analogy with the analysis of two mutually interacting scalar fields, \textit{cf}., Example \ref{Ex: interacting fields}, we now compute the expectation value of the interacting counterparts of the following functionals up to $\mathcal{O}(\lambda^3)$: 
\begin{flalign*}
    \Psi^{\rho}[f] (\psi, \bar{\psi}), \quad (\bar{\Psi}_{\rho'} \Psi^{\rho}) [f] (\psi, \bar{\psi}) &:= \int_{\mcM} d\mu_x \, \psi^{\rho} (x) \bar{\psi}_{\rho'}(x) f(x), \\ \Psi^2 [f] (\psi, \bar{\psi}) \equiv (\Psi^{\rho} \Psi^{\rho'})[f] (\psi, \bar{\psi}) &:= \int_{\mcM} d\mu_x \, \psi^{\rho}(x) \psi^{\rho'}(x) f(x),
\end{flalign*}
 for any $f \in \mathcal{D}(\mcM)$ and $(\psi, \bar{\psi}) \in \Gamma^{\infty}(D\mcM) \times \Gamma^{\infty}(D^*\mcM)$. By direct inspection, $R_V(\Psi^{\rho}[f]) \big \vert_{\psi, \bar{\psi} = 0} =_{\mathcal{O}(\lambda^3)} 0$. Due to the fact that the interacting potential contains an even number of fields, we deduce that $R_V(\Psi^{\rho}[f]) \big \vert_{\psi, \bar{\psi} = 0} =_{\mathcal{O}(\lambda^\infty)} 0$ at any perturbative order. In a similar fashion, one can show that $R_V(\bar{\Psi}_{\rho} [f]) \big \vert_{\psi, \bar{\psi} = 0} =_{\mathcal{O}(\lambda^{\infty})} 0$ and, more in general, $$R_V((\bar{\Psi}^n \Psi^m) [f]) \big \vert_{\psi, \bar{\psi} = 0} =_{\mathcal{O}(\lambda^{\infty})}  0$$ whenever $n \ne m$. Therefore, the only potentially non-vanishing expectation values should arise from functionals with symmetric spinorial content. Let $F[f ] := (\bar{\Psi}_{\rho'} \Psi^{\rho}) [f]$ and 
 \begin{flalign*}
     R_V(F) \big \vert_{\psi, \bar{\psi} = 0} =_{\mathcal{O}(\lambda^3)} \frac{\lambda^2}{2\hbar^2} \left[ \underbrace{2 V[h] \star_{\boldsymbol{H}} (V[h] \cdot_T F)}_{(\text{a})} - \underbrace{(V[h] \cdot_T V[h]) \cdot_T F[f]}_{(\text{b})} - \underbrace{(V[h] \cdot_{AT} V[h]) \star_{\boldsymbol{H}} F[f]}_{(\text{c})}  \right]. 
 \end{flalign*}
 Let us compute separately each term. 
 \begin{itemize}
     \item[(a)] Adopting the same graphical notation as in Example \ref{Ex: interacting fields}, the first contribution corresponds to 
\begin{equation*}
        \wick{
        \c1 \bPsi [h] \c2 {\textcolor{red}{\Psi [h]}} \c3 \bPsi [h] \c4 {\textcolor{red}{\Psi [h]}} \, \, \;
        (\c4 {\textcolor{red}{\bPsi [h]}} \c3 \Psi [h] \c5 \bPsi [h] \c1 \Psi [h] \, \;
        \c2 {\textcolor{red}{\bPsi [h]}} \c5 \Psi [f])}, 
    \end{equation*}
where denote in red those contractions performed via $H^{\psi}$ or $H_F^{\psi}$, whilst we indicate in black those contractions implemented by $-H^{\bar{\psi}}$ or $-H_{F}^{\bar{\psi}}$, depending on the contribution under scrutiny. In the above, the Dirac indices are omitted for the sake of notational clarity. At the level of integral kernels this yields, up to combinatorial factors,
\begin{equation}
    \propto \lambda^2 \int_{\mcM^3} d\mu_x d\mu_y d\mu_z \, (H_{F}^{\psi})^2(x,y) H^{\bar{\psi}}(x,y) H^{\bar{\psi}}(x,z) H_{F}^{\psi}(y,z) h(x) h(y) f(z), 
\end{equation}
where, with a slight abuse of notation, we have suppressed spinor indices and omitted the explicit appearance of Dirac gamma matrices. More precisely, $(H_{F}^{\psi})^2 := ( H_F^{\psi} \gamma^{\mu} H_F^{\psi} \gamma_{\mu})^{\rho}_{\rho'}$ and the corresponding integral kernel should be understood as involving the trace over Dirac indices of such terms. 
\item[(b)] Bearing in mind the above notational caveat, for the second contribution we have the following two non-vanishing diagrams:   
\begin{equation*}
        \wick{
        (\c1 \bPsi [h] \c2 {\textcolor{red}{\Psi [h]}} \c3 \bPsi [h] \c4 {\textcolor{red}{\Psi [h]}} \, \;
        \c2 {\textcolor{red}{\bPsi [h]}} \c3 \Psi [h] \c5 \bPsi [h] \c1 \Psi [h]) \; \, \;
        \c4 {\textcolor{red}{\bPsi [h]}} \c5 \Psi [f])}, 
    \end{equation*}
and 
\begin{equation*}
        \wick{
        (\c1 \bPsi [h] \c2 {\textcolor{red}{\Psi [h]}} \c3 \bPsi [h] \c4 {\textcolor{red}{\Psi [h]}} \, \;
        \c4 {\textcolor{red}{\bPsi [h]}} \c1 \Psi [h] \c2 \bPsi [h] \c5 \Psi [h]) \; \, \;
        \c5 {\textcolor{red}{\bPsi [h]}} \c3 \Psi [f])}, 
    \end{equation*}
which, at the level of integral kernels, read
\begin{equation*}
    \propto \lambda^2 \int_{\mcM^3} d\mu_x d\mu_y d\mu_z \, (H_F^{\bar{\psi}})^2 (x,y) H_F^{\psi} (x,y) H_F^{\psi}(x,z) H_F^{\bar{\psi}} (y,z) h(x) h(y) f(z), 
\end{equation*}
and 
\begin{equation*}
    \propto \lambda^2 \int_{\mcM^3} d\mu_x d\mu_y d\mu_z \, (H_F^{\psi})^2 (x,y) H_F^{\bar{\psi}} (x,y) H_F^{\bar{\psi}}(x,z) H_F^{\psi} (y,z) h(x) h(y) f(z), 
\end{equation*}
respectively. 
\item[(c)] The analysis of the last term is identical to $(\text{b})$ barring minor modifications, and, hence, we omit the discussion.  
 \end{itemize}
\end{example}

\paragraph{Wetterich equation in the Local Potential Approximation}
\label{Sec: Wetterich equation for the Thirring model}
In this section, we derive the Polchinski and Wetterich equations for self-interacting Dirac fields. The strategy closely parallels that used for two mutually interacting scalar fields; therefore, we limit ourselves to highlighting only the essential steps.  

We begin by introducing an infrared regulator of the form:
\begin{equation}
    \label{Eq: Thirring IR regulator}
    Q_k (\eta, \bar{\eta}) := -\int_{\mathcal{M}} d\mu_x \, \left[T(\eta^{\rho} [q_{k,\psi}]_{\rho}{}^{\rho'} \bar{\eta}_{\rho'})(x)- T(\bar{\eta}_{\rho'} [q_{k, \bar{\psi}}]^{\rho'}{}_{\rho} \eta^{\rho}) (x)\right], 
\end{equation}
for any $(\eta, \bar{\eta}) \in \Gamma^{\infty}(D\mcM) \times \Gamma^{\infty}(D^* \mcM)$, where $T(\cdot)$ denotes the time-ordering of the expression within brackets. Without loss of generality, we can choose $[q_{k,\psi}]_{\rho}^{\rho'} = [q_{k, \bar{\psi}}]^{\rho'}_{\rho} = k^2 \mathbb{I}_{\rho}^{\rho'}$. Furthermore, we consider the external current 
\begin{equation}
    \label{Eq: Thirring current}
    J(\eta, \bar{\eta}) := \int_{\mcM} d\mu_x \bar{j}_{\rho}(x) \eta^{\rho}(x) + \int_{\mcM} d\mu_x j^{\rho}(x) \bar{\eta}_{\rho}(x), \quad (j, \bar{j}) \in \Gamma^{\infty}_0(D\mcM) \times \Gamma^{\infty}_0(D^*\mcM). 
\end{equation}
In full analogy to the scalar theory, the regularise generating functional $Z_k$ and its connected counterpart $W_k$ are thus given by 
\begin{equation}
\label{Eq: Thirring generating functionals}
Z_k(j, \bar{j}) := \omega(S(V)^{-1} \star S(V + J + Q_k)), \quad W_k(j, \bar{j}) := -i \log Z_k(j, \bar{j}), 
\end{equation} 
with $\omega$ a state on the quantum algebra, not necessarily quasi-free. The first-order functional derivatives of $W_k$ yield the classical fields, \textit{i.e.}, 
\begin{equation*}
    \begin{cases}
        \frac{\delta W_k}{\delta \bar{j}_{\rho}(x)} =: \psi^{\rho}(x), \\
        \frac{\delta W_k}{\delta j^{\rho}(x)} =: \bar{\psi}_{\rho}(x), 
    \end{cases}
\end{equation*}
while $(j_{\psi})^{\rho},(j_{\bar{\psi}})_{\rho}$ are the currents which solve the above system as a function of $\bar{\psi}_{\rho}, \psi^{\rho}$, respectively. The connected, regularised generating functional is the key ingredient in the \textbf{Polchinski equation}, which reads
\begin{flalign}
    \label{Eq: Polchinski equation} 
    \notag
    \partial_k W_k(j, \bar{j}) &= - \int_{\mcM} d\mu_x \frac{1}{Z_k(j, \bar{j})} \left[ (\partial_k q_{k, \psi})^{\rho'}_{\rho}(x) \omega \left(S(V)^{-1} \star [S(V+J+Q_k) \cdot_T T(\eta^{\rho} \bar{\eta}_{\rho'})]\right)\right] \\ &+ \int_{\mcM} d\mu_x \, \frac{1}{Z_k(j, \bar{j})} \left[ (\partial_k q_{k, \bar{\psi}})^{\rho'}_{\rho}(x)  \omega \left(S(V)^{-1} \star [S(V+J+Q_k) \cdot_T T(\bar{\eta}_{\rho'} \eta^{\rho})]\right)\right]. 
\end{flalign}
Upon taking the Legendre transform of $W_k$ in both variables, one obtains the regularised effective action
\begin{equation*}
    \tilde{\Gamma}_k(\psi, \bar{\psi}) = W_k(j_{\psi}, j_{\bar{\psi}}) - J_{\Psi}(\psi, \bar{\psi}), 
\end{equation*}
where $J_{\Psi}$ is the functional in Equation \eqref{Eq: Thirring current} with the currents replaced by $j_{\psi}$ and $j_{\bar{\psi}}$. From $\tilde{\Gamma}_k$, one can derive the quantum equations of motion: 
\begin{equation*}
    \begin{cases}
        \frac{\delta \tilde{\Gamma}_k}{\delta \psi^{\rho}} = - (j_{\bar{\psi}})_{\rho}, \\
        \frac{\delta \tilde{\Gamma}_k}{\delta \bar{\psi}_{\rho}} = - (j_{\psi})^{\rho}. 
    \end{cases}
\end{equation*}
Crucial for the derivation of the Wetterich equation is the regularised average effective action specified by
\begin{equation*}
    \Gamma_k (\psi, \bar{\psi}) := \tilde{\Gamma}_k(\psi, \bar{\psi}) - Q_k(\psi, \bar{\psi}). 
\end{equation*}
The \textbf{Wetterich equation} thus reads
\begin{flalign}
    \label{Eq: Wetterich eq for Thirring fields}
    \notag  \partial_k \Gamma_k(\psi, \bar{\psi}) = - \int_{\mcM} d\mu_x \frac{1}{Z_k(j, \bar{j})} &\left[ (\partial_k q_{k, \psi})^{\rho'}_{\rho}(x) \omega \left(S(V)^{-1} \star [S(V+J+Q_k) \cdot_T T(\eta^{\rho} \bar{\eta}_{\rho'})] - T(\psi^{\rho} \bar{\psi}_{\rho'})\right)\right. \\ &+ \left. (\partial_k q_{k, \bar{\psi}})^{\rho'}_{\rho} (x)  \omega \left(S(V)^{-1} \star [S(V+J+Q_k) \cdot_T T(\bar{\eta}_{\rho'} \eta^{\rho})]  - T(\bar{\psi}_{\rho'}\psi^{\rho})\right)\right]. 
\end{flalign}
In the local potential approximation, we can express the average effective action as
\begin{equation}
    \label{Eq: LPA Thirring model}
    \Gamma_k (\psi, \bar{\psi}) = - \int_{\mcM} d\mu_x \, \left[ \bar{\psi}_{\rho} \slashed{D}^{\rho}_{\rho'} \psi^{\rho'} + \psi^{\rho} (\slashed{D}^*)^{\rho'}_{\rho} \bar{\psi}_{\rho'} + U_k(\psi, \bar{\psi})\right],
\end{equation}
where $U_k(\psi, \bar{\psi})$ is a suitable polynomial in the fields. Expanding it around classical field configurations, we obtain 
\begin{equation*}
    \Gamma_k(\psi, \bar{\psi}) =  \Gamma_k(\psi_{cl}, \bar{\psi}_{cl}) + \frac{1}{2} \Psi^T \Gamma_k^{(2)} (\psi_{cl}, \bar{\psi}_{cl}) \bar{\Psi} +\mathcal{O}(|\Psi - \Psi_{cl}|^3), 
\end{equation*}
where $\Psi \in \Gamma^{\infty}(D\mcM \oplus D^*\mcM)$ is the spinor doublet as per Notation \ref{Not: Dirac doublet}, while $\psi_{cl}, \bar{\psi}_{cl}$ are such that $\frac{\delta \Gamma_k}{\delta \psi^{\rho}} \big \vert_{\psi = \psi_{cl}} = 0$ and $\frac{\delta \Gamma_k}{\delta \bar{\psi}_{\rho}} \big \vert_{\psi = \psi_{cl}} = 0$. Here, we denote by $\Gamma_k^{(2)}(\psi, \bar{\psi})$ the Hessian matrix of $\Gamma_k$, taking values in the vector space $W$. As discussed at the beginning of Section \ref{Sec: (B) Self-interacting Dirac fields}, with a suitable choice of basis, $\Gamma_k^{(2)}$ can be written in such a way that the mixed derivatives lie on the diagonal. In this setting, the effective mass of the theory is given by 
\begin{flalign}
    \label{Eq: effective mass Thirring}
    \notag M &= I_0^t[\psi, \bar{\psi}] - I_0[\psi, \bar{\psi}] \\ &= - \int_\mcM d\mu_x \, \left[ \left( \underbrace{\frac{1}{2} [U_{k, \psi \bar{\psi}}^{(2)} (\psi, \bar{\psi}) - m_k \mathbb{I}]^{\rho'}_{\rho}}_{(M_\psi)^{\rho'}_{\rho}}  \psi^{\rho}(x) \bar{\psi}_{\rho'}(x) \right) + \left( \underbrace{\frac{1}{2} [U_{k, \bar{\psi} \psi}^{(2)} (\psi, \bar{\psi}) + m_k \mathbb{I}]^{\rho'}_{\rho}}_{(M_{\bar{\psi}})^{\rho'}_{\rho}} \bar{\psi}_{\rho'}(x)  \psi^{\rho}(x) \right) \right]. 
\end{flalign}
We can choose the effective potential to be of the form 
\begin{equation}
    \label{Eq: effective potential Thirring}
    U_k(\psi, \bar{\psi}) = U_{0,k} + m_k \bar{\psi}_{\rho}  \psi^{\rho'}  -m_k \psi^{\rho'} \bar{\psi}_{\rho} + \frac{\lambda_k}{2} (\bar{\psi} \gamma^{\mu} \psi) \bar{\psi}_{\rho} (\gamma_{\mu})^{\rho}_{\rho'} \psi^{\rho'}. 
\end{equation}
Notice that, being the relevant couplings the same for both the spinor and cospinor fields, we can discard the second term in the previous equation to derive the flow equations. 
Henceforth, we shall specialise the analysis to the case of $2-$dimensional Minkowski spacetime, \textit{i.e.}, $(\mathbb{R}^2, \eta)$. In this setting, 
\begin{equation*}
    \partial_k U_k (\psi, \bar{\psi}) = \lim_{y \rightarrow x} \frac{k^2}{2} \mathbb{I}^{\rho'}_{\rho} \left[ \left( \Delta_{F, M_{\psi}, k}^{\psi}(y,x) - H_{F, M_{\psi}, k}^{\psi} (y,x)\right)^{\rho}_{\rho'} - \left( \Delta_{F, M_{\bar{\psi}}, k}^{\bar{\psi}}(y,x) - H_{F, M_{\bar{\psi}}, k}^{\bar{\psi}} (y,x)\right)^{\rho}_{\rho'} \right],
\end{equation*}
and, focusing specifically on the first contribution, we have 
\begin{equation}
    \label{Eq: Wetterich for LPA Thirring}
    \partial_k U_k (\psi, \bar{\psi}) = \lim_{y \rightarrow x} \frac{k^2}{2} \mathbb{I}^{\rho'}_{\rho} \left( \Delta_{F, M_{\psi}, k}^{\psi}(y,x) - H_{F, M_{\psi}, k}^{\psi} (y,x)\right)^{\rho}_{\rho'}. 
\end{equation}
In the two-dimensional setting, for the vacuum state, this reduces to 
\begin{equation*}
     \partial_k U_{k, \psi \bar{\psi}} (\psi, \bar{\psi}) = \lim_{y \rightarrow x} \frac{k^2}{2} \mathbb{I}^{\rho'}_{\rho} (\slashed{D}^*)^{\rho}_{\rho'} \left( \Delta_{S, M, k}^{\infty}(y,x) - H_{M, k}^{\psi} (y,x)\right) = \frac{k^2}{2} \mathbb{I}^{\rho'}_{\rho} (M_{k})^{\rho}_{\rho'}, 
\end{equation*}
where $M_k$ is a suitable effective mass, here defined by $M_k := k^2 + m_k^2 + \frac{\lambda_{k}}{2} (\bar{\psi} \gamma^{\mu} \psi) \gamma_{\mu}$. We thus obtain the following set of beta functions for the relevant couplings: 
\begin{flalign}
    \label{Eq: Uk0 Thirring}
    k \partial_k U_{0,k} &= k M_k \bigg\vert_{\substack{\psi = 0\\ \bar{\psi} = 0}} = k (k^2 + m_k) \\
    \label{Eq: mk Thirring}
    k \partial_k m_k &= k^2 \frac{\delta^2 M_k}{\delta \psi^{\rho} \delta \bar{\psi}_{\rho}} \bigg \vert_{\substack{\psi = 0\\ \bar{\psi} = 0}} = k^3 \frac{\lambda_k}{2} \\ 
    \label{Eq: lambda k Thirring}
    k \partial_k \lambda_k &= 2 k^2 \frac{\delta^4 M_k}{\delta\psi^{\rho} \delta \bar{\psi}_{\rho} \delta \psi^{\rho'} \delta \bar{\psi}_{\rho'}} \bigg \vert_{\substack{\psi = 0\\ \bar{\psi} = 0}} = 0. 
\end{flalign}

\begin{remark}
The vanishing of the beta function for the self-interaction coupling in the three-dimensional Dirac model is compatible with simple dimensional considerations. In \(d=3\), the Dirac field has canonical dimension \([\psi]=\frac{(d-1)}{2}=1\), so that the quartic Thirring interaction, \textit{i.e.},
\[
(\bar\psi \gamma^\mu \psi)(\bar\psi \gamma_\mu \psi)
\]
has dimension \(4\). Its coupling therefore has negative mass dimension, \([\lambda]=-1\) and, therefore, it is not marginal. Within the local truncation used here, no logarithmic contribution is generated which could induce a non-trivial quantum running of this coupling. The corresponding beta function therefore vanishes. If one instead introduces a dimensionless coupling, its scale dependence will be purely canonical at this level of approximation.
\end{remark}

\section{Existence of local solutions to the Renormalization Group flow equations}
\label{Sec: Existence of local solutions to the Renormalization Group flow equations}

The issue we shall address in this section concerns whether it is possible to establish local existence of solutions to the Renormalization Group (RG) flow equations. To tackle this problem, it is convenient - albeit not mandatory - to work on an ultra-static spacetime $\mathcal{M}$, $\text{dim}(\mcM) = d$, so to have explicit and manageable expressions for the advanced and retarded propagators of the free theory. Furthermore, we shall assume that the background $\mcM$ is spin, namely, that its second Stiefel–Whitney class vanishes. This ensures the existence of at least one admissible spin structure on $\mcM$. To guarantee uniqueness, we further require that the spin bundle $S\mcM$ is trivial. 

As before, let $E \equiv E[\mcM, \pi, \mathbb{K}^n]$ be a trivial $\mathbb{K}$-vector bundle, where $\mathbb{K}$ denotes either $\mathbb{R}$ or $\mathbb{C}$, over $\mcM$, \textit{i.e.}, $E$ is globally diffeomorphic to $\mcM \times \mathbb{K}^n$. Smooth sections of this bundle $\Gamma^{\infty}(E) \simeq C^{\infty}(\mcM) \otimes \mathbb{K}^n$ correspond to field multiplets, where $\otimes$ is the algebraic tensor product and $n \in \mathbb{N}$ depends on the model under scrutiny. For the sake of readability, we separate the analysis of the two cases under consideration: \textbf{(A)} two mutually interacting scalar fields, see Section \ref{Sec: (A) Two mutually interacting scalar fields} and \textbf{(B)} self-interacting Dirac fields, see Section \ref{Sec: (B) Self-interacting Dirac fields}. Furthermore, we shall succinctly comment on how to extend this formalism to stochastic fields in the Martin-Siggia-Rose formalism, \textit{cf.}, Section \ref{Sec: MSR fields}. 

\subsection{(A) Two mutually interacting scalar fields}
\label{Sec: (A) Two mutually interacting scalar fields}

Consider two, mutually interacting scalar fields on $\mcM$ whose dynamics is governed by Equation \eqref{Eq: action 2 scalar fields}. In this scenario, the dynamical fields $\varphi_1, \varphi_2 \in \mathcal{E}(\mcM) := C^{\infty}(\mcM; \mathbb{R})$ can be equivalently regarded as a single smooth section of the trivial (real) vector bundle $E: = \mcM \times \mathbb{R}^2$. Thus, we can introduce the field doublet 
\begin{equation*}
    \Phi := \begin{pmatrix}
    \varphi_1 \\
    \varphi_2
\end{pmatrix} \in \Gamma^{\infty}(\mcM \times \mathbb{R}^2) \simeq \mathcal{E}_2(\mcM) \simeq \mathcal{E}(\mcM) \otimes \mathbb{R}^2,
\end{equation*}
where $\mathcal{E}_2(\mcM) := \mathcal{E}(\mcM) \oplus \mathcal{E}(\mcM)$. The Latin indices will be employed to indicate the components of such doublet.

\begin{notation}
\label{Not: trace}
In the following, we shall denote by $\mathcal{E} := L^2(\mcM) \otimes \mathbb{C}^2$ and let $A \in \text{Mat}(2; \mathcal{E})$ be such that each matrix entry $A_{ij}$ is a trace-class operator on $L^2(\mcM)$ for all $i, j = 1, 2$. Assume in addition that the action of such an operator on $\Psi \in \mathcal{E}$ is given by $$(A\Psi)_i(x) := \sum_{j=1, 2} (A_{ij} \Psi_j)(x) = \sum_{j=1,2} \int_{\mcM} d\mu_y \, A_{ij}(x,y) \Psi_j(y), \, \forall i =1, 2.$$ We define the \emph{Hilbert space trace of $A$ on $\mathcal{E}$} as
\begin{equation}
\label{Eq: Hilbert space trace}
\text{Tr}_{\mathcal{E}} (A) := \int_{\mcM} d\mu_x \, \text{tr}_{\mathbb{C}^2} A(x,x).
\end{equation}
We stress that if $A$ in not trace-class, Equation \eqref{Eq: Hilbert space trace} is valid only at a formal level. 
\end{notation}

\noindent In the same spirit of Section \ref{Sec: Wetterich for two scalar fields}, we shall consider an infrared mass regulator of the form 
\begin{equation*}
    \label{Eq: infrared regulator 1}
    Q_k (\chi_1, \chi_2) = - \frac{1}{2} \int_{\mathcal{M}} d\mu_x \, \left[q_{k,1}(x) T\chi_1^2(x) + q_{k,2}(x) T\chi_2^2 (x)\right], 
\end{equation*}
and an external current 
\begin{equation*}
    \label{Eq: current 1}
    J(\chi_1, \chi_2) := \int_{\mathcal{M}} d\mu_x \left[ j_1(x) \chi_1(x) + j_2(x) \chi_2(x)\right], \, \, j_1, j_2 \in C^{\infty}_0(\mathcal{M}).
\end{equation*}
Additionally, let us denote by 
\begin{equation}
    \mathrm{P}_0 := \begin{pmatrix}
        P_{0, 1} & 0 \\
        0 & P_{0, 2}
    \end{pmatrix} \, \, \text{and}  \, \, \mathsf{q}_k = \begin{pmatrix}
        q_{k, 1} & 0 \\
        0 & q_{k, 2}
    \end{pmatrix}, 
\end{equation}
where $P_{0,\ell} := -\Box_g + m_{\ell}^2 + \xi_{\ell} R$, $\ell = 1,2$. Furthermore, let 
\begin{equation}
   Z_k (j_1, j_2) = \omega(S(V)^{-1} \star S(V+ Q_k + J)), \quad  W_k(j_1, j_2) = - i \log Z_k(j_1, j_2),
\end{equation}
be the generating functional of the underlying regularized theory and its connected counterpart, respectively. For future convenience, we introduce the following short-hand notation: for any interacting operator $F$, we shall denote its weighted expectation value by
\begin{equation}
    \label{Eq: exp of F}
    \langle F \rangle := e^{-i W_k} \omega \left( S(V)^{-1} \star S(V+Q_k + J) \cdot_{T} F \right),
\end{equation}
where the deformed and time-ordered products are defined as per Equations \eqref{Eq: deformed product 2 scalar fields} and \eqref{Eq: time-ordered product}, whilst $V$ denotes the interacting potential.
A direct computation entails that the Hessian of $W_k(j_1, j_2)$ takes the form
\begin{equation}
\label{Eq: Wk 1}
[W_k^{(2)}(j_1, j_2)](x,y) = i
    \begin{pmatrix}
        \langle \chi_1(x) \cdot_T \chi_1(y) \rangle - \varphi_1(x) \varphi_1(y) & \langle \chi_1(x) \chi_2(y) \rangle - \varphi_1(x) \varphi_2(y) \\
         \langle \chi_2(x) \chi_1(y)\rangle - \varphi_2(x) \varphi_1(y) &   \langle \chi_2(x) \cdot_T \chi_2(y) \rangle - \varphi_2(x) \varphi_2(y)
    \end{pmatrix}, 
\end{equation}
where we recall that $\varphi_{\ell}$ are the classical fields, \textit{i.e.}, $\frac{\delta W_k}{\delta j_{\ell}(x)} =: \varphi_{\ell}(x)$, ${\ell}=1,2$. This implies that we can recast the \textbf{Polchinski equation} in the form: 
\begin{flalign*}
    \partial_k W_k (j_1, j_2) &= - \frac{1}{2} \text{Tr}_{\mathcal{E}} \left[ \partial_k \mathsf{q}_k : (W_k^{(2)} + \Phi \otimes \Phi):\right] \\ &= - \frac{1}{2} \int_{\mathcal{M}} d\mu_x \text{tr}_{\mathbb{C}^2} \left[\partial_k \mathsf{q}_k : (W_k^{(2)} + \Phi \otimes \Phi) : \right] (x,x) \\ &=- \frac{1}{2} \lim_{y \rightarrow x} \int_{\mathcal{M}} d\mu_x \, \sum_{\ell=1,2} \partial_k q_{k, \ell}(x) \left[ -i W_{k, \ell}^{(2)}(x,y) + \varphi_l(x) \varphi_l(y) - \tilde{H}_{F, \ell} (x,y) \right], 
\end{flalign*}
where $-i W_{k, \ell}^{(2)}(x,y) $ is the propagator of the interacting theory for the field $\varphi_{\ell}$, whilst $\tilde{H}_{F, \ell}(x,y)$, ${\ell} = 1, 2$, are the counter-terms stemming from the point-splitting regularization. In the above formula, $\text{Tr}_{\mathcal{E}}$ is as per Equation \eqref{Eq: Hilbert space trace} -- see also Notation \ref{Not: trace}. 

\begin{remark}
    Note that, by construction, the off-diagonal contributions in Equation \eqref{Eq: Wk 1} do not appear at the level of the Polchinski equation. Yet,  they still affect the diagonal contributions when taking the matrix inverse. 
\end{remark}

\noindent Thus, denoting by $G_{k, \ell} (x,y) := -i W_{k, \ell}^{(2)}(x,y)$, the \textbf{Renormalization Group (RG) flow equation} reads: 
\begin{flalign}
    \label{Eq: RG flow 2 scalars}
    \notag
    \partial_k \Gamma_k &= - \frac{1}{2} \int_{\mathcal{M}} d\mu_x \, \sum_{\ell=1,2} \partial_k q_{k, \ell}(x) : G_{k, \ell}(x,x) :_{\tilde{H}_{F,\ell}} \\ &= - \frac{1}{2} \lim_{y \rightarrow x} \int_{\mathcal{M}} d\mu_x \, \sum_{\ell=1,2} \partial_k q_{k, \ell}(x) \left( G_{k, \ell}(x,y) - \tilde{H}_{F,\ell} (x,y)\right),
\end{flalign}
where in the last line the normal-ordering prescription has been made explicit. We can now define the {\bf effective potential} $U_k$ by means of the relation:
\begin{equation}
\label{Eq: effective potential matrix form}
   \Gamma^{(2)}_k(\varphi_1, \varphi_2) - \mathsf{q}_k  := \mathrm{P}_0 + U^{(2)}_k(\varphi_1, \varphi_2), 
\end{equation} 
where $U_k^{(2)}(\varphi_1, \varphi_2)$ is the associated Hessian matrix. The operator appearing on the right hand-side of Equation \eqref{Eq: effective potential matrix form} is referred to in the literature as \emph{quantum wave operator}.  
Setting $[(\cdot)^{(2)}]_{\ell j} := \frac{\delta^2 (\cdot)}{\delta \varphi_{\ell} \varphi_j}$, Equation \eqref{Eq: effective potential matrix form} can be equivalently written componentwise as 
\begin{equation}
    \begin{cases}
        [\Gamma_{k}^{(2)}]_{\ell \ell} - q_{k, \ell} =: P_{0, \ell} + [U_{k}^{(2)}]_{\ell \ell}, \\
        [\Gamma_{k}^{(2)}]_{\ell j} =: [U_{k}^{(2)}]_{\ell j} \, \, \, \, \text{for} \, \, \ell \ne j.  
    \end{cases}
\end{equation}
\noindent Let
\begin{equation*}
   \Delta_{A/R} = \begin{pmatrix}
        \Delta_{A/R, 1} & 0 \\
        0 & \Delta_{A/R, 2}
    \end{pmatrix},
\end{equation*}
$\Delta_{A/R, \ell}$ being the advanced/retarded Green operators associated to $P_{0, \ell}$, $\ell = 1, 2$. In a similar fashion as in \cite{Dangelo_23},  we shall work in the \emph{local potential approximation}, namely, we assume that $U_k(\varphi_1, \varphi_2)$ is a multiplicative (scalar) operator, depending on polynomials of the fields but not on their derivatives, that is, 
\begin{flalign}
    \label{Eq: LPA assumptions}
    \notag
    U_k (\varphi_1, \varphi_2) &= \int_{\mathcal{M}} d\mu_x \, u (\varphi_1 (x), \varphi_2 (x), k) f(x), \, \, f \in C^{\infty}_{0}(\mathcal{O}) \\ [U_k^{(2)} (\varphi_1, \varphi_2)]_{\ell j}(x,y) & =\partial^2_{\ell j} u(\varphi_1(x), \varphi_2(x), k) f(x) \delta(x,y),   
\end{flalign}
where $\mathcal{O} \subset \mathcal{M}$ is a compact region of spacetime such that $\text{supp}(f) \subseteq \mathcal{O}$. Under this assumption, the operator $\mathrm{P}_k := \mathrm{P}_0 + U^{(2)}_k$ is still Green hyperbolic, as off-diagonal contributions -- being of order zero -- do not affect the principal symbol. Therefore, $\mathrm{P}_k$ admits unique advanced and retarded Green operators 
$$\Delta^U_{A/R}: \Gamma^{\infty}_0(\mcM; \mathbb{C}^2) \rightarrow \Gamma^{\infty} (\mcM; \mathbb{C}^2)$$ such that 
\begin{equation*}
    \mathrm{P}_k \circ \Delta^U_{A/R} = \text{id} \vert_{\Gamma^{\infty}_0(\mcM; \mathbb{C}^2)}, \, \, \, \, \Delta^U_{A/R} \circ \mathrm{P}_k \vert_{\Gamma^{\infty}_0(\mcM; \mathbb{C}^2)} = \text{id} \vert_{\Gamma^{\infty}_0(\mcM; \mathbb{C}^2)},
\end{equation*}
and, for every $\psi \in \Gamma^{\infty}_0(\mcM; \mathbb{C}^2)$,
\begin{equation*}
  \text{supp} \, \Delta^U_{A/R} (\psi)  \subset J^{\mp}(\text{supp} \psi).
\end{equation*}
Observe that, on account of the Schwartz kernel theorem, we can associate to the above Green operators matrix-valued bi-distributions $\Delta^U_{A/R} \in \mathcal{D}'(\mcM \times \mcM; \text{Mat}(2; \mathbb{C}))$, \textit{i.e.},
$$\Delta^U_{A/R} = \begin{pmatrix}
    [\Delta^U_{A/R}]_{1 1} &  [\Delta^U_{A/R}]_{1 2} \\
     [\Delta^U_{A/R}]_{2 1} &  [\Delta^U_{A/R}]_{2 2},
\end{pmatrix}$$ 
each entry lying in $\mathcal{D}'(\mcM \times \mcM)$. 

\begin{notation}
    With a slight abuse of notation, we denote by $\Delta^U_{A/R}$ both the Green operators and the corresponding distributional kernel, since we feel that no confusion can arise. 
\end{notation}

\begin{remark}
    Although $[U^{(2)}_k]_{12} = [U^{(2)}_k]_{21}$ due to the commutativity of the functional derivatives, it is {\it not} true in general that $[\Delta^U_{A/R}]_{12} = [\Delta^{U}_{A/R}]_{21}$. The reason is that the advanced and retarded propagators are not symmetric in their entries. Hence, we can only conclude that, at level of integral kernels, $[\Delta^U_{A/R}]_{12} (x,x)= [\Delta^{U}_{A/R}]_{21} (x,x)$.
\end{remark}

\noindent We can now introduce the advanced and retarded M\o ller operators associated to $\mathrm{P}_k := \mathrm{P}_0 + U_k^{(2)}$. These intertwine the free and quantum wave operators and they are specified by
\begin{equation}
    \label{Eq: adv/ret Moller}
    (\mathbb{I}- U_{k}^{(2)} \Delta^U_{A}) \, \, \text{and} \, \, (\mathbb{I}-\Delta^U_{R} U_{k}^{(2)}).
\end{equation}
Using the M\o ller operators, we can recast Equation \eqref{Eq: RG flow 2 scalars} in the form of a differential equation for the effective potential: 
\begin{equation}
\label{Eq: RG flow for Uk}
    \partial_k U_k = - \frac{1}{2} \int_{\mathcal{M}} d\mu_x \,\sum_{\ell =1,2}  \partial_k q_{k, \ell}(x)  [(\mathbb{I}-\Delta^U_{R} U_{k}^{(2)}) w (\mathbb{I}- U_{k}^{(2)} \Delta^U_{A})]_{\ell \ell} (x,x), 
\end{equation}
where,
$$w := \begin{pmatrix}
    w_{11} & 0 \\
    0    & w_{22}
\end{pmatrix} \in \mathcal{D}'(\mcM \times \mcM; \text{Mat}(2; \mathbb{C})),$$
with $w_{\ell \ell}(x,y) := \Delta_{F, \ell} (x,y) - H_{F, \ell} (x,y),$ for each $\ell =1,2$.  


Henceforth, we shall assume that the classical fields $\varphi_1$ and $\varphi_2$ are constant. This implies that the integral kernel of the effective potential $u(\varphi_1, \varphi_2, k)$ together with its second derivatives in Equation \eqref{Eq: LPA assumptions} are constant functions with respect to the variable $x$. In contrast with the strategy detailed in the previous section, we make the following two choices: 
\begin{itemize}
  \item[\ding{104}] In order to further simplify the analysis, we shall assume that the infrared regulator $Q_{k, \ell}$ has an integral kernel with a linear dependence on $k$, that is:  
    \begin{equation*}
        q_{k, \ell}(x) := (k_{0, \ell} + \epsilon k) f(x), \, \, \ell = 1,2
    \end{equation*}
    where $f$ is the adiabatic cutoff function in Equation \eqref{Eq: LPA assumptions} and $\epsilon > 0$ is chosen to be the same for both fields. Without loss of generality, we shall set $\epsilon \equiv 1$ and assume that $k \in \mathbb{R}_+$. 
    Observe, in addition, that
    the constant contribution $\propto k_{0, \ell}$ can be conveniently reabsorbed in the free mass of the theory and, hence, it can be discarded. 
    \item[\ding{104}] We shall refrain ourselves from expanding the effective average action around a fixed background solution. Specifically, we do not need to truncate the resulting expression up to quadratic contributions in the fields, but we can keep the full non-linear structure in the effective potential.   
\end{itemize}
 
\noindent Employing the short-hand notation $u^{(2)} := (\partial^2_{\ell j} u)_{\ell j}$, Equation \eqref{Eq: RG flow for Uk} can thus be recast as an initial, boundary value problem for $u(\varphi_1, \varphi_2, k)$: 
\begin{equation}
\label{Eq: BV problem}
    \begin{cases}
        \partial_k u = \sum_{\ell=1,2} G_{k, \ell}(u^{(2)}) =: \text{Tr}_{\mathcal{E}} (\mathrm{G}_{k}(u^{(2)})) \\
        u(\varphi_1, \varphi_2, a) = \psi \\
        u \vert_{\partial X \times [a,b]} = \beta
    \end{cases}
\end{equation}
where $X \subseteq \mathbb{R}^2$ is a compact set containing all possible values of $(\varphi_1, \varphi_2)$, $k \in [a, b] \subseteq \mathbb{R}_+$, whilst $\psi \in C^{\infty}(X)$ and $\beta \in C^{\infty}(\partial X \times [a,b])$ are assigned functions. In the above equation, the functional $G_{k, \ell}$ is defined for each $\ell = 1, 2$ as 
\begin{equation}
\label{Eq: Gkl}
    G_{k, \ell} (u^{(2)}) := - \frac{1}{2 ||f||_{L^1}} \int_{\mathcal{M}} d\mu_x \, \partial_k q_{k,\ell}(x) \left[ (\mathbb{I}- u^{(2)} \Delta^u_{R} f) \otimes (\mathbb{I}- u^{(2)} \Delta^u_{R} f) (w) \right]_{\ell \ell} (x,x), 
\end{equation}
where $\Delta^u_{R} \in \mathcal{D}'(\mcM \times \mcM; \text{Mat}(2; \mathbb{C}))$ is the retarded fundamental solution of the Green hyperbolic operator $\mathrm{P}_{0} + f u^{(2)}$.

\begin{notation}
\label{Not: matrix}
Consider matrix-valued bi-distributions $A_1, A_2, W \in \mathcal{D}'(\mcM \times \mcM; \text{Mat}(2; \mathbb{C})) \simeq \mathcal{D}'(\mcM \times \mcM) \times \text{Mat}(2; \mathbb{C})$. Since in Equation \eqref{Eq: Gkl} we are interested in the coinciding point limit, we can decompose $$W(x,y) := \sum_{i,j =1, 2} W_{ij} (x,y) E_{ij},$$ where $i,j$ are matrix indices, $W_{ij} \in \mathcal{D}'(\mcM \times \mcM)$, while $(E_{ij})_{ab} := \delta_{ia} \delta_{jb}$ is the Weyl basis. Thus, we can write at the level of integral kernels 
\begin{equation*}
    A_1 (x, x') \otimes A_2 (y,y') = \sum_{i_1, j_1 = 1,2} \sum_{i_2, j_2 = 1,2} (A_1)_{i_1 j_1}(x,x') (A_2)_{i_2 j_2}(y,y') E_{i_2 j_2} E_{i_1 j_1}  
\end{equation*}
and 
\begin{equation*}
    \left[A_1 (x,x') \otimes A_2 (y, y')\right] (W) (x',y') =  \sum_{i, j = 1,2} \sum_{i_1, i_2 = 1,2} (A_1)_{i_1 i} (x',x) (A_2)_{i_2 j} (y',y)  W_{ij} (x',y') E_{i_1 i_2}, 
\end{equation*}
where we used the property $(E_{i_2 j_2})^T = E_{j_2 i_2}$ and $E_{i_1 j_1} E_{ij}  = \delta_{j_1 i} E_{i_1 j}$. Henceforth, we shall compactly write 
\begin{flalign*}
   [(A_1 \otimes A_2) (W)]_{\ell \ell} (x,y) := \sum_{i, j = 1,2} \int_{\mcM^2} d\mu_x' d\mu_y' \, (A_1)_{\ell i} (x,x') (A_2)_{\ell j} (y,y')  W_{ij} (x',y').
\end{flalign*}
\end{notation}

\begin{remark}
Observe that Equation \eqref{Eq: Gkl} can be equivalently written as follows: 
    \begin{equation}
        \label{Eq: Gkl 1}
         G_{k, \ell} (u^{(2)}) := - \frac{1}{2 ||f||_{L^1}} \int_{\mathcal{M}} d\mu_x \, \partial_k q_{k,\ell}(x) : G_{k, \ell} : (x,x),  
    \end{equation}
    where 
    \begin{equation*}
        : G_{k, \ell} : (x,y) := [(\mathbb{I} - \Delta^u_{R} U_{k}^{(2)}) w (\mathbb{I} - U_{k}^{(2)} \Delta^u_{A})]_{\ell \ell} (x,y), \, \, \ell = 1,2. 
    \end{equation*}
\end{remark}

To establish local existence of solutions to Equation \eqref{Eq: BV problem}, the strategy consists in applying the \emph{Nash-Moser Theorem} in the Hamiltonian formulation, \textit{cf.} \cite{Moser1966a, Moser1966b, Nash1956} \cite[Part III]{Hamilton_1982}. We recall it here for the reader's convenience -- see also Appendix \ref{Sec: Appendix A} for further details.

\begin{theorem}[Nash-Moser Theorem in Hamilton formulation]
\label{Thm: Nash-Moser}
    Consider two tame Fréchet spaces $F$ and $G$ and let $R: \mathcal{U} \subset F \rightarrow G$ be a tame map. Under the additional requirements: 
    \begin{itemize}
        \item[1.] The map $D R: \mathcal{U} \times F \rightarrow G$, \textit{i.e.}, the first functional derivative of $R$, admits a unique inverse $E: \mathcal{U} \times G \rightarrow F$ such that if $DR(u) v = f$, then $E(u) f= v$ for all $u \in \mathcal{U}$ and $f \in G$,
        \item[2.] the inverse $E$ is a smooth tame map,
    \end{itemize}
$R$ is locally invertible and $R^{-1}$ is tame smooth. 
\end{theorem}

\noindent In order to exploit Theorem \ref{Thm: Nash-Moser}, we shall identify in our setting the tame map playing the role of $R$ and prove it abides by the conditions $1.$ and $2.$ in the above statement. For the sake of clarity, we shall divide the following analysis in separate steps. 

\textbf{Step 0: Preliminary definitions}. Let $\Phi := (\varphi_1, \varphi_2) \in X$ and define $F := C^{\infty}(X \times [a,b])$ which is a Fréchet space with respect to the family of seminorms: 
\begin{equation}
    ||u||_{n} = \sum_{|\alpha| \le n} \sup_{\Phi, k} |D^{\alpha} u(\Phi, k)|,
\end{equation}
where the derivatives $D^{\alpha}$ are taken both with respect to the fields and to the parameter $k$. The space $F$ is also tame, being the space of smooth functions over a compact set. Consider the tame Fréchet subspace
\begin{equation*}
    F_0 := \{u \in F \, | \, u(\Phi, a) = 0, u \vert_{\partial X \times [a,b]} = 0\} \subset F
\end{equation*}
and decompose the solution $\tilde{u}$ of Equation \eqref{Eq: BV problem} as
\begin{equation*}
    \tilde{u} = u_b + u, \, \, u \in F_0,
\end{equation*}
where $u_b \in F$ is such that it abides by the initial and boundary conditions, \textit{i.e.}, $u_b(\Phi, a) = \psi$ and $u_b \vert_{\partial X \times [a,b]} = \beta$. Assume in addition that there exists a positive constant $A > 0$ such that $||u||_4 \le A$. This requirement will prove crucial in deriving the sought norm estimates.

\textbf{Step 1: The Renormalization Group operator is tame smooth}. 
In this context, the map in Theorem \ref{Thm: Nash-Moser} is given by the Renormalization Group (RG) operator, \textit{i.e.},  
\begin{equation}
\label{Eq: RG operator}
\mathcal{RG}: F_0 \subset F \rightarrow F, \, \, \, \mathcal{RG}(u) := \partial_k (u+u_b) - \sum_{\ell = 1, 2} G_{k, \ell} (u^{(2)}+ u_b^{(2)}),
\end{equation} 
where $G_{k, \ell}$ is as per Equation \eqref{Eq: Gkl}. In order to show that $\mathcal{RG}$ is tame smooth in a suitably small neighborhood $\mathcal{U} \subset F_0$ of $0$, it suffices to prove that $G_{k, \ell}(\tilde{u}^{(2)})$, $\ell = 1,2$, is a tame smooth map for $\tilde{u} \in u_b + \mathcal{U}$. To this avail, we need to devise the counterparts of \cite[Lem. 4.1]{Dangelo_23} to this setting. 

\begin{lemma}
\label{Lem: 4.1}
Let $\mcM$ be a ultra-static spacetime. Given $g \in \Gamma^{\infty}_0(\mcM; \mathbb{C}^2)$, consider 
$$ h = \Delta^U_R g.$$ Then, $h$ is a past compact smooth section with compact support on any Cauchy hypersurface $\Sigma$. In addition, since $$h = (\mathbb{I} - \Delta^U_R U_k^{(2)}) \varphi,$$ where $\varphi = \Delta_R g$, the following estimates hold true: 
\begin{itemize}
    \item[i.] there exist positive constants $c, C' > 0$ such that $C'$ depends on the support of $f$ in $U_k$ but not on $\tilde{u}$, while $c$ does not depend on $U_k$ and
\begin{equation}
\label{Eq: Lem. 4.1 1}
    ||h||^t_{L^{\infty}} \le c ||h||^t_{H^2} \lesssim c \, ||\varphi||^t_{H^2} \exp({C' ||\tilde{u}^{(2)}||_{\infty}}) \lesssim c ||\varphi||^t_{H^2} \exp({C' ||\tilde{u}||_2}),
\end{equation}
where $||\tilde{u}^{(2)}||_{\infty} : =\max_{\ell = 1, 2} \sum_{j = 1, 2} |\tilde{u}^{(2)}|_{\ell j}$ is the induced $\ell^{\infty}$ operator norm on $\text{Mat}(2; \mathbb{C})$ and the symbol $\lesssim$ corresponds to $\le$ modulus irrelevant numerical factors, which can be reabsorbed in a redefinition of the constants $c, C'$.
\item[ii.] there exist a positive constant $\tilde{C} > 0$ depending only on the support of $g$ such that
\begin{equation}
\label{Eq: Lem. 4.1 2}
    ||h||^t_{L^{\infty}} \lesssim c \exp({C'||\tilde{u}||_2}) \int_{-\infty}^t d \tau \, (t-\tau) ||g||^{\tau}_{H^2} \lesssim \tilde{C} \, \exp({C'||\tilde{u}||_2}) \sup_{\tau \le t} ||g||^{\tau}_{H^2},
\end{equation}
where $c, C' > 0$ are as per item $i.$
\end{itemize}
In the above estimates, we denote by $||\cdot||^t_{L^\alpha}$ the norm on $L^{\alpha}(\Sigma_t; \mathbb{C}^2)$, given  for $\alpha \in [1, \infty)$ by 
\begin{equation*}
    ||h||^t_{L^{\alpha}} := \left(\int_{\Sigma_t} d\mu_{\textnormal{\textbf{x}}} \, |h (t, \textnormal{\textbf{x}})|^{\alpha}\right)^{\frac{1}{\alpha}},
\end{equation*}
where $d\mu_{\textnormal{\textbf{x}}}$ is the measure induced by the Riemannian metric on the Cauchy hypersurface $\Sigma_t$ at a fixed time $t$, whilst for $\alpha = \infty$, we set 
\begin{equation*}
    ||h||^t_{L^{\infty}} := \text{ess} \sup_{\textnormal{\textbf{x}} \in \Sigma_t} |h(t, \textnormal{\textbf{x}})|,
\end{equation*}
$|h(t, \textbf{x})|$ being the $\mathbb{C}^2-$norm of $h$. Similarly, $||\cdot||^t_{H^2}$ is the norm on $H^2(\Sigma_t; \mathbb{C}^2)$.
\end{lemma}

\begin{proof}
($i.$) A direct consequence of $\mathrm{P}_0$ and $\mathrm{P}_k$ being Green hyperbolic is that $\Delta_R$ and $\Delta^U_R$ both map past compact smooth sections to past compact smooth sections. This entails that $h, \varphi \in \Gamma^{\infty}_{pc}(\mcM; \mathbb{C}^2)$. Recalling also that 
    $$\Delta^U_R = \Delta_R(\mathbb{I}-U_k^{(2)}\Delta^U_R) = (\mathbb{I}- \Delta^U_R U_k^{(2)}) \Delta_R.$$ Note that, since $U_k^{(2)} = f \tilde{u}^{(2)}$, where $f$ is a smooth, compactly supported function and $\tilde{u}^{(2)}$ is constant on $\mcM$, we have the following: 
    \begin{equation*}
        h = \Delta_R g - \Delta_R U_k^{(2)} \Delta^U_R g =  \Delta_R g - \Delta_R U_k^{(2)} h = \varphi - \Delta_R U_k^{(2)} h,  
    \end{equation*}
    or, equivalently, componentwise we can write $$h_{\ell} = \varphi_{\ell} - \sum_{j=1, 2} [\Delta_{R} U_k^{(2)}]_{\ell j} h_j = \varphi_{\ell} - \sum_{j=1, 2} \Delta_{R, \ell} [U_k^{(2)}]_{\ell j} h_j,\, \, \ell = 1, 2.$$
    Denote by $D_{\Sigma_t}$ the Laplace operator on $\Sigma_t$ and define $\omega_{\ell} := \sqrt{D_{\Sigma_t} +m^2_{\ell}}$ for $\ell = 1, 2$. Thus, 
    \begin{equation*}
        h_{\ell}(t, \textbf{x}) = \varphi_{\ell}(t, \textbf{x}) - \sum_{j=1,2} [\tilde{u}^{(2)}]_{\ell j} \int_{-\infty}^t d\tau \, \frac{\sin(\omega_{\ell}(t-\tau))}{\omega_{\ell}} (fh_j) (\tau, \textbf{x}). 
    \end{equation*}
As in \cite[Lem. 4.1]{Dangelo_23}, we can conclude that, for each $\ell = 1,2$, 
\begin{flalign*}
    ||h_{\ell}||^t_{L^2} & \le ||\varphi_{\ell}||^2_{L^2} + \sum_{j=1,2} |[\tilde{u}^{(2)}]_{\ell j} | \int_{a}^t d\tau \, (t-\tau) ||f||^{\tau}_{L^{\infty}} ||h_j||^{\tau}_{L^2}, \\ ||h_{\ell}||^t_{H^2} & \le ||\varphi_{\ell}||^t_{H^2} + C \sum_{j=1,2} |[\tilde{u}^{(2)}]_{\ell j} |  \int_{a}^t d\tau ||h_j||^{\tau}_{H^2}
\end{flalign*}
where the norms here are on $L^{\alpha}(\Sigma_t)$ and $H^2(\Sigma_t)$, respectively, $C > 0$ is a positive constant independent on $\tilde{u}^{(2)}$, while $a \in \mathbb{R}$ depends on $f$. Applying the Gr\"onwall lemma in integrated form, \textit{cf.}, \cite{Gronwall1919}, we obtain 
\begin{equation*}
    ||h_{\ell}||^t_{H^2} \le ||\varphi_{\ell}||^t_{H^2} \, \exp({C \sum_{j=1,2} |[\tilde{u}^{(2)}]_{\ell j}|}), \, \, \, \ell =1,2.
\end{equation*}
Now, recalling that, 
\begin{equation*}
    ||h||^t_{L^2} = \left(\sum_{\ell = 1, 2} (||h_{\ell}||^t_{L^2})^2 \right)^{\frac{1}{2}} \le \sum_{\ell = 1, 2} ||h_{\ell}||^t_{L^2} \le \sum_{\ell = 1, 2} \left[||\varphi_{\ell}||^t_{H^2} + C \sum_{j=1,2} |[\tilde{u}^{(2)}]_{\ell j} |  \int_{a}^t d\tau ||h_j||^{\tau}_{H^2} \right],
\end{equation*}
an, passing to the Sobolev space norm, we get 
\begin{flalign*}
    ||h||^t_{H^2} &\le \sum_{\ell = 1, 2} ||\varphi_{\ell}||^t_{H^2} \, \exp(C \sum_{j = 1, 2} |[\tilde{u}^{(2)}]_{\ell j} |) \le \sum_{\ell = 1, 2} ||\varphi_{\ell}||^t_{H^2} \, \exp(C' ||\tilde{u}^{(2)}||_{\infty}) \\ &\le \sqrt{2} ||\varphi||^t_{H^2} \exp(C' ||\tilde{u}^{(2)}||_{\infty}).
\end{flalign*}
In a similar fashion as in \cite[Lem. 4.1]{Dangelo_23}, one can show that there exists $c > 0$ independent of $U_k$ such that
\begin{equation*}
    ||h||^t_{L^{\infty}} \le c ||h||^t_{H^2}.  
\end{equation*}
The estimate in Equation \eqref{Eq: Lem. 4.1 1} thus follows. \\
($ii.$) To prove the estimate in Equation \eqref{Eq: Lem. 4.1 2}, it suffices to note that 
\begin{equation*}
    \varphi_{\ell}(t, \textbf{x}) = [\Delta_R g]_{\ell} (t, \textbf{x}) = \int_{-\infty}^t d\tau \, \frac{\sin(\omega_{\ell}(t-\tau))}{\omega_{\ell}} g_\ell (\tau, \textbf{x}),
\end{equation*}
where we exploit the fact that the matrix-valued bi-distribution $\Delta_R$ is diagonal. Hence, for each $\ell = 1,2$, it holds that 
\begin{equation*}
    ||\varphi_{\ell}||^t_{H^2}  \le  \int_{-\infty}^{t} d\tau \, (t-\tau) ||g_\ell||^{\tau}_{H^2}, 
\end{equation*}
which entails 
\begin{equation*}
    ||\varphi||^t_{H^2}  \le  \int_{-\infty}^{t} d\tau \, (t-\tau) \sum_{\ell=1,2} ||g_\ell||^{\tau}_{H^2} \le \sqrt{2} \int_{-\infty}^{t} d\tau \, (t-\tau) ||g||^{\tau}_{H^2},
\end{equation*}
where in the last line we apply the Cauchy-Schwarz inequality. 
\end{proof}

Observe that, in order to prove that the RG operator as per Equation \eqref{Eq: RG operator} is tame smooth, it is sufficient to show that the maps $G_{k, \ell}$, $\ell = 1, 2$ in Equation \eqref{Eq: Gkl} are tame smooth. This is indeed the case, as granted by the following lemma. 

\begin{lemma}
\label{Lem: 4.4}
The functionals $G_{k, \ell}$, $\ell = 1, 2$ are smooth functions of $\tilde{u}^{(2)}$ and, in particular, they are tame smooth for $\tilde{u} \in u_b + \mathcal{U}$. 
\end{lemma}

\begin{proof}
For the sake of clarity, we divide the proof in two parts. \\
\textbf{Part 1: $G_{k, \ell}(\tilde{u}^{(2)})$ is smooth}. A direct consequence of the explicit choice of $q_{k, \ell}$, is that $G_{k, \ell}$ depends on $(\Phi, k)$ only through $\tilde{u}^{(2)}$, \textit{i.e.}, $G_{k, \ell}(\tilde{u}^{(2)}) (\Phi, k) = G_{k, \ell}(\tilde{u}^{(2)} (\Phi, k))$. 
The $n-$th order functional derivative of $\tilde{G}_{k, \ell} (\tilde{u}) := G_{k, \ell}(\tilde{u}^{(2)})$ with respect to $\tilde{u}$ thus reads
    \begin{flalign*}
        \tilde{G}^{(n)}_{k, \ell} (\tilde{u}; v_1, ..., v_n) &= \frac{(-1)^{n+1}}{2 ||f||_{L^1}} \int_{\mathcal{M}}  d\mu_x \, \partial_k q_{k,\ell}(x) \\ &\left[ \sum_{\boldsymbol{m} \in \{0, 1\}^n} \prod_{i=1}^n \left((\mathbb{I}- \tilde{u}^{(2)} \Delta^{\tilde{u}}_{R}) (v_i^{(2)} \Delta^{\tilde{u}}_R)^{m_i} f \otimes (\mathbb{I}- \tilde{u}^{(2)} \Delta^{\tilde{u}}_{R}) (v_i^{(2)} \Delta^{\tilde{u}}_R)^{1-m_i} f \right) (w) \right]_{\ell \ell} (x,x) \\ &= \frac{(-1)^{n+1}}{2 ||f||_{L^1}} \int_{\mathcal{M}}  d\mu_x \, \partial_k q_{k,\ell}(x) \sum_{\boldsymbol{m} \in \{0, 1\}^n} \left[(\mathbb{I}- \tilde{u}^{(2)} \Delta^{\tilde{u}}_{R}) \otimes (\mathbb{I}- \tilde{u}^{(2)} \Delta^{\tilde{u}}_{R}) \circ  \right.\\& \left. \prod_{i=1}^n  \left( (v_i^{(2)} \Delta^{\tilde{u}}_R)^{m_i} f \otimes (v_i^{(2)} \Delta^{\tilde{u}}_R)^{1-m_i} f \right) (w) \right]_{\ell \ell} (x,x),
    \end{flalign*}
where the sum is over the $2^n$ choices of $\boldsymbol{m} = (m_1, \ldots, m_n)$ with $m_i \in \{0,1\}$, while we define $A^0 = id$ for any linear operator $A$ such that $A(x) \in \text{Mat}(2; \mathbb{C})$, $x \in \mcM$. Due to the compactness of the supports of $f$ and $q_{k, \ell}$ and thanks to the smoothness of $w$, one has that the integral in the above expression yields a finite result. Thus, we can conclude that $G_{k, \ell}$ is smooth as a function of $\tilde{u} \in F$. \\
\textbf{Part 2. $G_{k, \ell}(\tilde{u}^{(2)})$ is tame smooth}.
To show that $G_{k, \ell}$ is tame, we need to estimate \begin{flalign*}
    ||\tilde{G}_{k, \ell} (\tilde{u})||_n := \sum_{|\alpha| \le n} \sup_{\Phi, k} |D^{\alpha} G_{k, \ell} (\tilde{u}^{(2)}(\Phi, k))| = ||\tilde{G}_{k, \ell}(\tilde{u})||_0 + \underbrace{\sum_{1 \le |\alpha| \le n} \sup_{\Phi, k} |D^{\alpha} G_{k, \ell} (\tilde{u}^{(2)}(\Phi, k))|}_{(1)}, 
\end{flalign*}
in terms of the seminorm $||\tilde{u}||_{n+r}$ for a certain $r \in \mathbb{N}$. Term $(1)$ can be more conveniently rewritten using the functional version of the {\it Faà di Bruno formula}, that is
\begin{flalign*}
   |D^{\alpha} \tilde{G}_{k, \ell}(\tilde{u}(\Phi, k))| &= |\sum_{1 \le p \le |\alpha|} \sum_{\substack{\alpha_1 + \ldots + \alpha_p = \alpha \\ |\alpha_i| \ge 1}} \tilde{G}^{(p)}_{k, \ell} (\tilde{u} (\Phi, k)) (D^{\alpha_i} \tilde{u}^{(2)}(\Phi, k), \ldots, D^{\alpha_p} \tilde{u}^{(2)}(\Phi, k))| \\ & \le C_{\alpha} \sum_{p \le |\alpha|} ||\tilde{G}^{(p)}_{k, \ell} (\tilde{u})||_0  \sum_{\substack{\alpha_1 + \ldots + \alpha_p = \alpha \\ |\alpha_i| \ge 1}} \prod_{i=1}^{p} \sup_{\Phi, k} |D^{\alpha_i} \tilde{u}^{(2)}(\Phi, k)| \\ & \le C_{\alpha} \sum_{p \le |\alpha|} ||\tilde{G}^{(p)}_{k, \ell} (\tilde{u})||_0 \sum_{\substack{\alpha_1 + \ldots + \alpha_p = \alpha \\ |\alpha_i| \ge 1}} \prod_{i=1}^p ||\tilde{u}^{(2)}||_{|\alpha_i|} \\ & \le C_{\alpha} \sum_{p \le |\alpha|} ||\tilde{G}^{(p)}_{k, \ell} (\tilde{u})||_0 ||\tilde{u}^{(2)}||^p_{|\alpha| - p +1},
\end{flalign*}
where we used the estimate $1 \le |\alpha_i| = |\alpha| - \sum_{j \ne i} |\alpha_j| \le |\alpha| -p + 1$. Observe that in the above equation we have {\it formally} defined $$||\tilde{G}^{(p)}_{k, \ell}(\tilde{u})||_0 := \sup_{\substack{v_1, \ldots, v_p \in F \\ ||v_i||_{2} = 1}} ||\tilde{G}^{(p)}_{k, \ell}(\tilde{u}; v_1, \ldots, v_p)||_0, \, \, \, \text{for any $p \in \mathbb{N}$},$$ for $\tilde{u} \in u_b + \mathcal{U}$ with $\mathcal{U}$ a sufficiently small neighbourhood of $0$. 
Hence, calling $m := |\alpha|$, we have
\begin{equation}
    \label{Eq: estimate on Gkl seminorm n}
    ||\tilde{G}_{k, \ell}(\tilde{u})||_n < ||\tilde{G}_{k, \ell}(\tilde{u})||_0 + \sum_{m = 1}^n \sum_{p=1}^m ||\tilde{G}^{(p)}_{k, \ell} (\tilde{u})||_0 ||\tilde{u}^{(2)}||^p_{m-p+1}, 
\end{equation}
where we discarded the irrelevant constant factor $C_{\alpha}$. Let us carefully examine the second contribution, splitting the analysis into two steps. 

\textit{Step 2a}. First and foremost, we need to devise a suitable estimate for the zeroth seminorm of the functional derivatives of $\tilde{G}_{k, \ell}$. This will entail that the multi-linear operators $\tilde{G}^{(p)}_{k, \ell} (\tilde{u}): F_0 \times \ldots \times F_0 \rightarrow F$ are also bounded for all $p \in \mathbb{N}$ with respect to the seminorm $||\cdot||_0$. Lemma \ref{Lem: 4.1} entails that, for any space compact smooth section $g \in \Gamma^{\infty}_{sc} (\mcM; \mathbb{C}^2)$, for each $i \in \{1, \ldots, p\}$, it holds true that
\begin{flalign*}
    ||(v_i^{(2)} \Delta^{\tilde{u}}_R)^{m_i} f g ||^t_{H^2} \le \begin{cases}
        \tilde{C} \sup_{\tau \le t} ||g||^{\tau}_{H^2}, \, \, \text{if} \, \, m_i = 0\\
        \tilde{C} \exp{(C' ||\tilde{u}||_2)} \sup_{\tau \le t} ||g||^{\tau}_{H^2} ||v_i||_4, \, \, \text{if} \, \, m_i = 1
    \end{cases},
\end{flalign*}
where $\tilde{C} > 0$ is a suitable positive constant depending on $f$ which incorporates all irrelevant numerical factors in Equations \eqref{Eq: Lem. 4.1 1} and \eqref{Eq: Lem. 4.1 2}. It descends that, the worst possible scenario corresponds to $m_i = 1$, that is 
\begin{flalign*}
    ||\prod_{i=1}^p (v_i^{(2)} \Delta^{\tilde{u}}_R)^{m_i} f g ||^t_{H^2} \le \tilde{C}^{p} \exp{(p C' ||\tilde{u}||_2)} \sup_{\tau \le t} ||g||^{\tau}_{H^2} \prod_{i=1}^{p} ||v_i||_4 \le c \, \tilde{C}^{p} \exp{(p C' ||\tilde{u}||_2)}, 
\end{flalign*}
where $c > 0$ is a positive constant depending on $p$ and on $\max_{i} ||v_i||_4$. Similarly, we have 
\begin{equation*}
    ||(\mathbb{I} - \tilde{u}^{(2)} \Delta^{\tilde{u}}_R) g ||^t_{H^2} \le (1+ ||\tilde{u}||_2) \tilde{C} \exp{(C' ||\tilde{u}||_2)} \sup_{\tau \le t} ||g||^{\tau}_{H^2} \le c' \exp{(C' ||\tilde{u}||_2)},
\end{equation*}
where we use that, per hypothesis, $||\tilde{u}||_2 \le ||u||_2 + ||u_b||_2 \le A'$ for some $A' > 0$. Now we introduce in Equation \eqref{Eq: Gkl} a cutoff function $\chi \in C^{\infty}_0(\mcM)$ such that $\chi \equiv 1$ in a region containing the support of $f$. We need to use the above estimates on the integral kernel of $\tilde{G}^(p)_{k, \ell}$: 
\begin{equation*}
    a_p (x,y) := \prod_{i=1}^p \left((\mathbb{I}- \tilde{u}^{(2)} \Delta^{\tilde{u}}_{R}) (v_i^{(2)} \Delta^{\tilde{u}}_R)^{m_i} f \otimes (\mathbb{I}- \tilde{u}^{(2)} \Delta^{\tilde{u}}_{R}) (v_i^{(2)} \Delta^{\tilde{u}}_R)^{1-m_i} f \right) (\chi w \chi) (x,y) \in \text{Mat}(2; \mathbb{C}).
\end{equation*}
It descends that
\begin{flalign*}
    \sup_{x \in \text{supp}(f)} |a_p (x,x)| &\le \sup_{x, y \in \text{supp}(f)} |a_p (x,y)| \\ &\le    \sup_{t_x, t_y \in \text{supp}(f)} (\tilde{C}')^p \exp{ ((p+2) ||\tilde{u}||_2)} ||\chi w \chi||^{(t_x, t_y)}_{H^2 \times H^2}, 
\end{flalign*}
where $||\cdot||^{(t_x, t_y)}_{H^2 \times H^2}$ denotes the Sobolev norm on $\Sigma_{t_x} \times \Sigma_{t_y}$ and we exploit the fact that, for linear operators $A_1: H^2 (\Sigma_{t_x}) \rightarrow H^2(\Sigma_{t_x}), A_2:  H^2(\Sigma_{t_y}) \rightarrow  H^2(\Sigma_{t_y})$ and for any smooth function $W \in C^{\infty} (\Sigma_{t_x} \times \Sigma_{t_y})$, it holds that $$||(A_1 \otimes A_2) W||^{(t_x, t_y)}_{H^2 \times H^2} \le ||A_1||^{t_x}_{H^2_x}  ||A_2||^{t_y}_{H^2_y}  ||W||^{(t_x, t_y)}_{H^2 \times H^2},$$ with $||\cdot||^{t_i}_{H^2_i}$ denoting the corresponding operator norms. Since $e^{C' ||\tilde{u}^{(2)}||_2} \le C_1 (1+ ||\tilde{u}^{(2)}||_2)$ for a sufficiently large positive constant $C_1 > 0$, we can conclude that, for fixed $p \in \{1, ..., m\}$ 
\begin{equation}
\label{Eq: estimate on Gpkl}
    ||\tilde{G}^{(p)}_{k, \ell}(\tilde{u})||_0 \le \text{const}. (1+ ||\tilde{u}||_2),  
\end{equation}
where $\text{const.}$ is independent on $\tilde{u}$. 

{\it Step 2b}. Going back to Equation \eqref{Eq: estimate on Gkl seminorm n}, an argument analogous to the one discussed in \cite[Lem. 4.1]{Dangelo_23} implies that we can estimate 
$$ ||\tilde{u}^{(2)}||^p_{m-p+1} < C ||\tilde{u}^{(2)}||_0^{p-1} ||\tilde{u}||_{m+3},$$ hence, 
\begin{flalign*}
    ||\tilde{G}_{k, \ell}(\tilde{u})||_n &< C \left( ||\tilde{G}_{k, \ell}(\tilde{u})||_0 + \sum_{m=1}^n \sum_{p=1}^{m} ||\tilde{G}^{(p)}_{k, \ell} (\tilde{u})||_0 ||\tilde{u}^{(2)}||^{p-1}_0 ||\tilde{u}||_{m+3} \right) \\ & < C (1 + ||\tilde{u}||_2) \left(\sum_{m=1}^n \sum_{p=1}^{m} ||\tilde{u}||^{m-1}_2 ||\tilde{u}||_{m+3} \right) \\ & <  C(1+ ||\tilde{u}||_{n+3}),
\end{flalign*}
where in the second line we use that $||\tilde{u}||_{m+3} \le ||\tilde{u}||_{n+3}$ and with a slight abuse of notation we still denote by $C$ the overall constant which depends on $n$. This concludes the proof. 
\end{proof}

\noindent Hence, since the RG operator is built as a linear combination of tame smooth functionals, we obtain the sought result. 
\begin{proposition}
    Denote by $\mathcal{U} \subset F_0$ a small neighborhood of $0$ and assume that, for any $u \in \mathcal{U}$, there exists some positive constant $A$ such that $||u||_2 \le A$. Then the RG operator $\mathcal{RG}: \mathcal{U} \subset F_0 \rightarrow F$ as per Equation \eqref{Eq: RG operator} is tame smooth -- see Appendix \ref{Sec: Appendix A}.
\end{proposition}

\textbf{Step 2: The linearized RG operator is tame smooth}. The goal is to prove the validity of item $1.$ in Theorem \ref{Thm: Nash-Moser}. Adopting the same notations as in Step 1, we define the linearized RG operator as the linear map 
\begin{equation}
    \label{Eq: linearized RG operator}
    L: (u_b + \mathcal{U}) \times F_0 \rightarrow F, \, \, L(u) v:= \partial_k v - \text{Tr}_{\mathcal{E}}[[\sigma(u)] (v)],
\end{equation}
where
\begin{flalign*}
  \text{Tr}_{\mathcal{E}}[[\sigma(u)] (v)] &:= \sum_{\ell =1, 2} [\sigma_{\ell}(u)] (v) = \frac{1}{2||f||_{L^1}} \int_{\mathcal{M}} d\mu_x \, \sum_{\ell = 1, 2} \partial_k q_{k, \ell}(x) \\ & \sum_{m_1 \in \{0, 1\}} \left[ (\mathbb{I}- u^{(2)} \Delta^u_{R}) (v^{(2)} \Delta^u_R)^{m_1} f \otimes (\mathbb{I}- u^{(2)} \Delta^u_{R}) (v^{(2)} \Delta^u_R)^{1- m_1} f  (w) \right]_{\ell \ell} (x,x), 
\end{flalign*}
where $[\cdot]_{\ell \ell} (y,x)$ is as per Notation \ref{Not: matrix}. Observe that, since $\partial_k q_{k, \ell}$ is constant in $k$, the dependence of $\sigma_{\ell}(u)(\cdot)$ on $(\Phi, k)$ is only through $u^{(2)}$.

\begin{remark}
    In contrast to the case of a single scalar field, here $\sigma_{\ell}(u)$ is \emph{not} a multiplicative operator acting on $v^{(2)}$. A direct inspection entails that it is not possible to factor out the dependence on the second field derivatives of $v$ and, hence, $\sigma_{\ell}(u)(\cdot)$ should be thought of as a linear operator acting on $v^{(2)}$. 
\end{remark}

\begin{remark}
   In our setting, the diffusion operator appearing in the linearized RG equation is no longer a scalar but it becomes a $2$ by $2$ matrix acting on the second partial derivatives of $v \in F_0$ with respect to the fields. Thus, we can rewrite the action of $L(u)$ making the matrix components explicit as follows. Define
   \begin{equation*}
       \sum_{\ell = 1,2} [\sigma_{\ell}(u)] (v) := \text{Tr}_{\mathcal{E}} [[\sigma(u)](v)] = \text{tr}_{\mathbb{C}^2} \int_{\mcM} d\mu_x [\sigma(u)] (v) (x, x). 
    \end{equation*}
Hence, we can define an integral operator $\Sigma (u) (\cdot):= \int_{\mcM} d\mu_x [\sigma(u)] (\cdot) (x, x)$ acting on $v^{(2)}$ such that $$\text{Tr}_{\mathcal{E}} [[\sigma(u)](v)] := \sum_{i,j=1,2} \Sigma_{i j}(u)(\partial_i \partial_j v),$$ with components  
 \begin{flalign}
  \label{Eq: B(u)}
   \notag \Sigma_{ij} (u) &:= \int_{\mcM} d\mu_x \, \sum_{\ell = 1,2} \sum_{n, m = 1, 2} \left( (\mathbb{I} - u^{(2)} \Delta^u_R)_{\ell i} (\Delta^u_R)_{j n} (\mathbb{I} - u^{(2)} \Delta^u_R)_{\ell m} \right. \\ & \left. + (\mathbb{I} - u^{(2)} \Delta^u_R)_{\ell n} (\mathbb{I} - u^{(2)} \Delta^u_R)_{\ell i} (\Delta^u_R)_{j m} \right) (w)_{n m} (x,x).
  \end{flalign}
\end{remark}

In order to prove that the linearized RG operator $L$ is a tame smooth map, we shall need the following property of the linear operators $\sigma_{\ell}$. 

\begin{proposition}
\label{Prop: sigma is tame smooth}
    The linear maps $\sigma_{\ell}(u)$ are tame smooth for each $\ell = 1, 2$.
\end{proposition}

\begin{proof}
    Note that, by direct inspection, $[\sigma^{(n)}_{\ell} (u)] (v_1, ..., v_n) = \tilde{G}^{(n+1)}_{k, \ell} (u; v_1, ..., v_n)$ and, thus, $[\sigma_{\ell}(u)] (v)$ is smooth as a function of $u^{(2)}$. In addition, being related to the functional derivative of $\tilde{G}_{k, \ell}$, on account of Lemma \ref{Lem: 4.4} it is straightforward to conclude that $\sigma_{\ell}(u): F_0 \rightarrow F$ is also a tame map. We conclude by observing that the linear combination of tame smooth maps $\sum_{\ell = 1, 2} [\sigma_{\ell} (u)] (v)$ is also tame smooth. 
\end{proof}

\noindent Especially in view of the physical applications, it seems natural to wonder whether, with a suitable choice of states, one can guarantee that for each $\ell = 1, 2$, $[\sigma_{\ell}(u)](v)$ is strictly positive as well as $||u||_2 \le A$ in a sufficiently small neighborhood of $0$. The answer is affirmative as illustrated by the following proposition.

\begin{proposition}
\label{Prop: sigma is positive-definite}
    Consider the initial, boundary-value problem in Equation \eqref{Eq: BV problem}. Assume that the functions $\beta, \psi$ and $u_b \in F$ are chosen in such a way that there exists $\epsilon > 0$ for which $||\beta||_2 + ||\psi||_2 < \epsilon$ and $||u_b||_2 \le \epsilon$. Then we can always choose $w \in C^{\infty}(\mcM \times \mcM; \text{Mat}(2; \mathbb{C}))$ for which there exists $\mathcal{U} \subset F_0$, such that, for all $u \in \mathcal{U}$, $v \in F_0$ $[\sigma_{\ell} (\tilde{u})](v) := [\sigma_{\ell}(u_b + u)](v) \ge c > 0$ and $||u||_2 < A = \epsilon$. 
\end{proposition}

\begin{proof}
We fix $\ell = 1, 2$ and observe that, on account of Proposition \ref{Prop: sigma is tame smooth}, we can write 
\begin{equation}
    [\sigma_{\ell} (u)](v) = [\sigma_{\ell} (0)](v) + \int_{0}^1 d \lambda \, \frac{d}{d\lambda} [\sigma_{\ell}(\lambda(u + u_b))](v) \ge [\sigma_{\ell} (0)](v) - \sup_{\lambda \in [0,1]} ||[\sigma^{(1)}_{\ell} (\lambda \tilde{u})](v) (u+u_b)||_0,  
\end{equation}
where we denote by $\sigma_{\ell}^{(1)}$ the functional derivative of $[\sigma_{\ell}(\tilde{u})] (v)$ with respect to $\tilde{u}$. 
Since $$ [\sigma_{\ell} (0)] (v):= \frac{1}{2||f||_{L^1}} \int_{\mathcal{M}^2} d\mu_x d\mu_y \, \partial_k q_{k, \ell}(x) \sum_{m \in \{0, 1\}} \left[ ((v^{(2)} \Delta_R)^{m} f \otimes (v^{(2)} \Delta_R)^{1- m} f)  (w) \right]_{\ell \ell} (y,x),$$ is linear in $w$ and it cannot be identically zero for any $w$, we can choose a state such that $[\sigma_{\ell} (0)](v) \ge (2 \epsilon C + c) > 0$, where $C > \sup_{\lambda} ||\sigma^{(1)}_{\ell} (\lambda \tilde{u})||_0$. Thanks to the smoothness of $\sigma_{\ell}(\tilde{u})$, we can choose $u_b$ such that $||u_b||_2 \le (||\beta||_2 + ||\psi||_2) < \epsilon$ and consider a sufficiently small neighborhood of $0$, $\mathcal{U} \subset F_0$ such that $||u||_2 < \epsilon, \forall u \in \mathcal{U}$. Furthermore, since $\sigma_{\ell}^{(1)}(\tilde{u}) (v, v_1) = \tilde{G}_{k, \ell}^{(2)} (v, v_1)$, upon taking the supremum norm over $v, v_1$,  we have that 
\begin{flalign*}
    ||[\sigma^{(1)} (\lambda \tilde{u})] (u + u_b)||_0 &\le ||\sigma^{(1)}(\lambda \tilde{u})||_0 ||u_b + u||_0  \\ &\le ||\tilde{G}^{(2)}_{k, \ell}(\lambda \tilde{u})||_0 ||u_b + u||_0 \\ &\le C'(1+ ||u + u_b||_2) ||u_b + u||_0 \\ &\le C'' ||u_b + u||_2 \\ &\le 2 \epsilon C'', 
\end{flalign*}
for suitable constants $C', C''$ which depends on $\epsilon$. 
\end{proof}

In \textbf{Step 4}, in order to prove that the inverse of the linearized RG operator is a tame smooth map, we shall make an additional requirement on the diffusion coefficient, which is justified in view of the following proposition. 

\begin{proposition}
\label{Prop: requirement 2}
    Let $\epsilon' > 0$ and choose $u_b \in F$ compatible with the boundary conditions such that $||u_b||_3 \le A$, with $A > 0$ sufficiently small. If $[a,b] \subset \mathbb{R}_+$ is such that $(b-a)$ is sufficiently small, it holds that, for all $u \in \mathcal{U}$, $||u||_4 < A$, 
    \begin{equation}
        \label{Eq: estimate on K(Du)}
        ||K(Du)||_0 := \sup_{||v||_2 = 1} ||[K(Du)](v)||_0 < \epsilon', 
    \end{equation}
    where $D \in \{\partial_{\varphi_1}, \partial_{\varphi_2}, \partial_k\}$ and $$[K(Du)](v) := \frac{\sum_{\ell = 1, 2} [D(\sigma_{\ell}(u))] (v)}{\sum_{\ell = 1,2} [\sigma_{\ell}(u)](v)},$$ for any $v \in F_0$. 
\end{proposition}

\begin{proof}
    On account of Proposition \ref{Prop: sigma is positive-definite}, we have that $\sum_{\ell = 1,2} [\sigma_{\ell}(u)](v) \ge c > 0$ for some constant $c$. This implies that $$\frac{1}{\sum_{\ell = 1,2} [\sigma_{\ell}(u)](v)} < \frac{1}{c}, \, \, \, v \in F_0.$$ Note that for each $\ell = 1, 2$, we have that  
    $$[D(\sigma_{\ell}(u(\Phi, k))](v) = [D(\sigma_{\ell}(u(\Phi, a))](v) + \int_{a}^k \partial_{\chi} [D(\sigma_{\ell}(u(\Phi, \chi))](v) \, d\chi.$$
    The smoothness of both $\sigma_{\ell}(u)$ and its first functional derivative entail that $$||D(\sigma_{\ell}(u(\Phi, k))||_0 \le ||D(\sigma_{\ell}(u(\Phi, a))||_0 + (b-a) ||\sigma_{\ell}(u)||_2.$$
    Exploiting the fact that $\sigma_{\ell}(u)$ are also tame maps, we have $||\sigma_{\ell}||_2 \le (1 + ||u||_4) \le C(1+A)$. Additionally, the continuity of $\sigma_{\ell}(u)$ with respect to $u$ entails that $$|D(\sigma_{\ell}(u(\Phi, a))(v)| \le C ||D [u_b^{(2)}(\Phi, a)||_{\infty} < ||u_b||_3 \le A, $$ and, thus, 
    $$||K(Du)||_0 < \frac{C}{c} (A + (b-a) (1+A) \le \epsilon',$$ for sufficiently small $A$ and $b-a$. 
\end{proof}

We can now prove the following proposition. 

\begin{proposition}
    The linearization of the RG operator $L(u) v = \partial_k v - \sum_{\ell=1,2} [\sigma_{\ell}(u)](v)$ is tame smooth. 
\end{proposition}

\begin{proof}
Via a combination of Leibnitz rule and an interpolating argument one obtains that 
\begin{flalign*}
    ||L(u) v||_n &\le ||v||_{n+1} + \sum_{\ell = 1,2} ||[\sigma_{\ell}(u)] (v)||_{n} \\&\le ||v||_{n+1} + \sum_{\ell = 1,2} ||\sigma_{\ell}(u)||_n ||v||_{n+2} 
\end{flalign*}
for some positive constant $C$. Since $\sigma_{\ell}(u)$ is tame smooth for each $\ell =1, 2$, see Proposition \ref{Prop: sigma is tame smooth}, and $L(u)$ is a linear combination of tame maps, $L: \mathcal{U} \times F_0 \rightarrow F$ is tame smooth. 
\end{proof}

\textbf{Step 3: The linearization of the RG operator is invertible}. Under the hypothesis that $\sigma_{\ell} > 0$ on $X \times [a, b]$, the linearized RG operator $L(u): F_0 \rightarrow F$ is a parabolic operator.

We have the following result -- see \cite{Dangelo_23}. 

\begin{proposition}[Existence and uniqueness of an inverse]
Let $u \in \mathcal{U}$ and denote by $L := L(u)$ the linearized RG operator as per Equation \eqref{Eq: linearized RG operator}. Under the assumption that $[\sigma_{\ell}(u)](v)$, $\ell = 1,2$, is strictly positive for every $u \in \mathcal{U}$, $v \in F_0$, there exists a unique inverse $E: F \rightarrow F_0$ compatible with the initial and boundary conditions $E(g)(\Phi, a) = 0$, $E(g) \vert_{\partial X \times [a,b]} = 0$, for all $g \in C^{\infty}_0(X \times [a, b])$, such that $$E(Lg) = L(E(g)) = g, \, \, \,  g \in F_0.$$ Furthermore, the inverse is continuous in the following sense: 
\begin{equation*}
    ||E(g)||_0 < C ||g||_0, \, \, \text{for some positive constant} \, \, C > 0.
\end{equation*}   
\end{proposition}

\textbf{Step 4: The inverse is tame smooth}. Having assessed existence and uniqueness of the inverse, we need to show that the map $E$ is tame smooth, so to apply the Nash-Moser theorem in Hamilton formulation. To this end, we shall need two auxiliary results.

\begin{lemma}
    \label{Lem: 4.10}
    Assume that for each $\ell = 1,2$ $[\sigma_{\ell}(u)](v)$ is strictly positive for $v \in F_0$ and $u \in \mathcal{U} \subset F_0$ such that $||u||_4 \le A$. In addition, suppose that there exists a sufficiently small $\epsilon > 0$ for which $\sup_{D} \sup_{||v||_2 = 1} ||[K(Du)](v)||_0 < \epsilon$, where $[K(Du)](v) := \frac{\sum_{\ell = 1, 2} [D(\sigma_{\ell}(u))] (v)}{\sum_{\ell = 1,2} [\sigma_{\ell}(u)](v)}$ for $v \in F_0$, while $D \in \{\partial_{\varphi_1}, \partial_{\varphi_2}, \partial_k\}$, see Proposition \ref{Prop: requirement 2}. Then, the following estimate holds true: 
    \begin{equation}
        \label{Eq: estimate 4.10 1}
        ||v||_1 < C ( ||g||_1 + \sum_{\ell = 1, 2} ||\sigma_{\ell}(u)||_1 ||g||_0),
    \end{equation}
    for some positive constant $C > 0$.
\end{lemma}

\begin{proof}
    Recall that the linearized RG operator is given by $L(u) v = \partial_k v - \sum_{\ell = 1, 2} [\sigma_{\ell} (u)] (v)$. The uniform continuity of the inverse $E$ entails that if $L(u) v = g$, then $$ ||E(g)||_0 = ||v||_0 < C ||g||_0.$$ Consider now $D \in \{\partial_{\varphi_1}, \partial_{\varphi_2}, \partial_k\}$ and apply the above estimate to $Dv$: $$||Dv||_0 < C ||L(u) Dv||_0. $$ Note that $$L(u) Dv = \partial_k D v - \sum_{\ell = 1,2} [\sigma_{\ell}(u)] (Dv) = D (L(u) v) - \sum_{\ell = 1, 2} [\sigma_{\ell}^{(1)}(u) (Du)] (v),$$ where we exploit the commutativity of $D$ with the partial derivatives, \textit{i.e.}, $(Dv)^{(2)} = D v^{(2)}$, and we denote by $[\sigma_{\ell}^{(1)}(u)]$ the functional derivative of the linear operator $\sigma_{\ell}(u)$ as a function of $u$. In particular, being $[\sigma_{\ell}(u)](v)$ strictly positive by hypothesis, one can multiply and divide by $\sum_{\ell = 1,2} [\sigma_{\ell}(u)] (v) := - L(u) v + \partial_k v$, thus obtaining 
    \begin{flalign*}
         L(u) Dv & = D(L(u) v) + \frac{\sum_{\ell = 1, 2} [\sigma_{\ell}^{(1)}(u) (Du)] (v)}{\sum_{\ell = 1,2} [\sigma_{\ell}(u)](v)} (L(u) v - \partial_k v). \\
        & = D(L(u) v) + [K(Du)](v) (L(u) v - \partial_k v).
    \end{flalign*}
    Therefore, the uniform continuity of $E$ and the fact that $\sum_{\ell = 1,2}[\sigma_{\ell}(u)](v) \ge c' > 0$ imply that 
    \begin{flalign*}
        ||Dv||_0 &< C ( ||D(L(u) v)||_0 +  ||K(Du)||_0 ( ||L(u) v||_0 + ||v||_1)) \\
        & < C (||Dg||_0 + ||K(Du)||_0 (||g||_0 + ||v||_1)).
    \end{flalign*}
    Thus, using the fact that $||u||_1 \le ||u||_4 \le A$, we have that 
    \begin{flalign*}
        ||v||_1 &\le \left( ||v||_0 + \sum_{D \in \{\partial_{\varphi_1}, \partial_{\varphi_2}, \partial_k\} } ||Dv||_0 \right) \\ & < C (||g||_1 + \sum_{\ell = 1,2} ||\sigma_{\ell}(u)||_1 ||g||_0 + \sup_{D} ||K(Du)||_0 ||v||_1), 
    \end{flalign*}
which entails that
    \begin{flalign*}
        (1 - C \sup_{D} ||K(Du)||_0) ||v||_1 < C (||g||_1 + \sum_{\ell = 1,2} ||\sigma_{\ell}(u)||_1 |g||_0).
    \end{flalign*}
    Since, by hypothesis, $\sup_{D} ||K(Du) ||_0 < \epsilon$, for $\epsilon > 0$ sufficiently small, we can guarantee that $ (1 - C \sup_{D} ||K(Du) ||_0)  > 0$ and, hence, the sought result descends: 
    \begin{equation*}
        ||v||_1 < C' (||g||_1 + \sum_{\ell = 1,2} ||\sigma_{\ell}(u)||_1 ||g||_0),
    \end{equation*}
    for a suitable positive constant $C' > 0$.    
\end{proof}


\begin{lemma}
    \label{Lem: 4.11}
    In the same setting of Lemma \ref{Lem: 4.10}, for every $n \in \mathbb{N}$, it holds that 
    \begin{equation}
        \label{Eq: lemma 4.11 1}
        ||v||_n < C (||g||_{n + 1} + \sum_{\ell = 1,2} ||\sigma_{\ell}(u)||_{n+1} ||g||_0),
    \end{equation}
    for some positive constant $C > 0$. 
\end{lemma}

\begin{proof}
    The proof goes by induction with respect to $n \in \mathbb{N}$. The case $n = 1$ is a direct consequence of Lemma \ref{Lem: 4.10} and the standard property that $||\sigma_{\ell}(u)||_{1} \le ||\sigma_{\ell}(u)||_{2}$ -- which still holds true at level of operator seminorms. Assume now that Equation \eqref{Eq: lemma 4.11 1} is valid up to $n-1$. We shall prove that it holds true also for $n$. To this avail, apply Equation \eqref{Eq: lemma 4.11 1} to $Dv$ with $D \in \{\partial_{\varphi_1}, \partial_{\varphi_2}, \partial_k\}$: 
    $$ ||Dv||_{n-1} < C (||L(u) Dv||_n + \sum_{\ell = 1,2} ||\sigma_{\ell}(u)||_{n} ||L(u) Dv||_0), $$ where we used that $L(u) v = g$. We need to estimate separately the two contributions. 
    \begin{flalign*}
        ||L(u) Dv||_n &< ||g||_{n+1} + C(||[K(Du)](v) g||_n + ||[K(Du)](v) \partial_k v||_n) \\ &< ||g||_{n+1} + C \left( ||K(Du)||_0 ||g||_n + ||K(Du)||_0 ||\partial_k v||_n \right) \\ & < ||g||_{n+1} + C \left( \sum_{\ell=1,2} ||\sigma_{\ell}(u)||_1 ||g||_n + ||K(Du)||_0 ||\partial_k v||_n \right),
    \end{flalign*}
    where $[K(Du)](v)$ is as above. By an interpolation argument, we can write 
    \begin{flalign*}
        ||L(u) Dv||_n &< ||g||_{n+1} + C \sum_{\ell=1,2} \left( ||\sigma_{\ell}(u)||_0 ||g||_{n+1} + ||\sigma_{\ell}(u)||_{n+1} ||g||_{0} \right. \\ & \left. + ||K(Du)||_0 ||\partial_k v||_{n} + ||K(Du)||_n ||\partial_k v||_{0} \right) \\ & < C \left( (1 + \sum_{\ell = 1,2} ||\sigma_{\ell}(u)||_0) ||g||_{n+1} + \sum_{\ell = 1,2} ||\sigma_{\ell}(u)||_{n+1} ||g||_{0} \right. \\ & \left. + ||K(Du)||_0 ||\partial_k v||_{n} + ||K(Du)||_n ||v||_{1} \right).
    \end{flalign*}
    Using now Lemma \ref{Lem: 4.10}, we can conclude that, modulus constant factors,  
    \begin{flalign*}
        ||L(u) Dv||_n &< (1 + \sum_{\ell = 1,2} ||\sigma_{\ell}(u)||_0) ||g||_{n+1} + \sum_{\ell = 1,2} ||\sigma_{\ell}(u)||_{n+1} ||g||_{0}  \\ &  + ||K(Du)||_0 ||v||_{n+1} + ||K(Du)||_n (||g||_1 + \sum_{\ell = 1,2} ||\sigma_{\ell}(u)||_1 ||g||_0). 
    \end{flalign*}
In addition, Lemma \ref{Lem: 4.10} together with the Leibnitz rule entail that 
\begin{flalign*}
    ||L(u)Dv||_0 &< ||g||_1 + ||K(Du)||_0 (||g||_0 + ||v||_1) \\& < ||g||_1 + \sum_{\ell = 1,2} ||\sigma_{\ell}(u)||_1 ((1 + \sum_{\ell = 1,2} ||\sigma_{\ell}(u)||_1) ||g||_0 + ||g||_1). 
\end{flalign*}
Hence, combining the two inequalities we get
\begin{flalign*}
  (1 - C ||K(Du)||_0) ||v||_{n} &< C \left[(1 + \sum_{\ell = 1,2} ||\sigma_{\ell}(u)||_0) ||g||_{n+1} + \sum_{\ell = 1,2} ||\sigma_{\ell}(u)||_{n+1} ||g||_{0} \right. \\ & \left. + ||K(Du)||_n (||g||_1 + \sum_{\ell = 1,2} ||\sigma_{\ell}(u)||_1 ||g||_0) \right. + \sum_{\ell=1,2} ||\sigma_{\ell}(u)||_{n} \\& \left. \left(||g||_1 + \sum_{\ell = 1,2} ||\sigma_{\ell}(u)||_1 ( ||g||_0  + \sum_{\ell = 1,2} ||\sigma_{\ell}(u)||_1 ||g||_0 + ||g||_1) \right) \right].
\end{flalign*}
Since we have that $(1 - C ||K(Du)||_0) >0$, we can divide both sides of the inequality by this quantity, thus obtaining
\begin{flalign*}
    ||v||_{n} &< C \left[(1 + \sum_{\ell = 1,2} ||\sigma_{\ell}(u)||_0) ||g||_{n+1} + \sum_{\ell = 1,2} ||\sigma_{\ell}(u)||_{n+1} ||g||_{0} \right. \\ & \left. + ||K(Du)||_n (||g||_1 + \sum_{\ell = 1,2} ||\sigma_{\ell}(u)||_1 ||g||_0) \right. + \sum_{\ell=1,2} ||\sigma_{\ell}(u)||_{n} \\& \left. \left(||g||_1 + \sum_{\ell = 1,2} ||\sigma_{\ell}(u)||_1 ( ||g||_0  + \sum_{\ell = 1,2} ||\sigma_{\ell}(u)||_1 ||g||_0 + ||g||_1) \right) \right].
\end{flalign*}
Invoking an interpolating argument, we can write $$||g||_1 ||h||_n \le C (||g||_0 ||h||_{n+1} + ||g||_{n+1} ||h||_0).$$ Furthermore, we have that $\sum_{\ell =1,2} ||\sigma_{\ell}(u)||_1 < A$ and, hence, $||K(Du)||_n \le C \sum_{\ell = 1,2}||\sigma_{\ell}(u)||_{n+1}$. The sought result thus descends. 
\end{proof}

\begin{proposition}
Assume that $\sigma_{\ell}(u)$ is strictly positive for each $\ell =1, 2$ and $u \in \mathcal{U}$ such that $||u||_4 \le A$ for some positive constant $A$. In addition, suppose that $\sup_{D} ||K(Du)||_0 < \epsilon$, where $K(Du)(v)$ is as per Lemma \ref{Lem: 4.10}. Then, the inverse $E$ is tame smooth. 
\end{proposition}

\begin{proof}
    The inverse $E$ being tame smooth means that, there exists a positive constant $C > 0$, such that $$||E(g)||_n \le C (1 + ||g||_{n+r})$$ for suitable $r, n \in \mathbb{N}$. We begin by noticing that the dependence of $L(u)$ on $u$ is only through the linear operators $\sigma_{\ell}(u)$, $\ell = 1,2$, which are tame smooth maps of $u$. Since the composition of tame smooth maps is tame, it suffices to study how $L(u)$ depends on $\sigma_{\ell}(u)$, $\ell = 1,2$. More precisely, Lemmata \ref{Lem: 4.10} and \ref{Lem: 4.11} provide estimates which allows us to control the seminorms of the derivative of $v$ with those of $g$. Hence, the desired result is a direct consequence of these two. 
\end{proof}

Thanks to the estimates in the previous proposition, a direct application of the Nash-Moser theorem \ref{Thm: Nash-Moser} entails the following. 

\begin{theorem}
    The RG operator admits a unique family of tame smooth local inverses and unique local solutions exist. 
\end{theorem}

\subsubsection{Stochastic fields in the Martin-Siggia-Rose formalism}
\label{Sec: MSR fields}
Via the Martin-Siggia-Rose formalism, this case can be regarded as a specific instance of the two, mutually interacting scalar fields with the caveats discussed in Section \ref{Sec: pAQFT for Martin-Siggia-Rose fields}. Indeed, working in the Lorentzian setting, the MSR fields $\varphi, \tilde{\varphi} \in \mathcal{E}(\mathbb{R}^d)$ and, thus, we can introduce the field doublet 
$$\boldsymbol{\Phi} := \begin{pmatrix}
    \varphi \\
    \tilde{\varphi}
\end{pmatrix} \in \mathcal{E}(\mathbb{R}^d) \otimes \mathbb{R}^2.$$ The analysis of this scenario proceeds along the same lines as in the two-scalar-field case. Yet, a crucial distinction arises in the interpretation of the results, since  $\tilde{\varphi}$ plays the r\^ole of an \emph{auxiliary} field associated with the stochastic contribution to the equation of motion.

\noindent In this case, the fields are smooth sections of the trivial vector bundle $E := \mathbb{R}^d \times \mathbb{R}^2$ and we can set $\mathcal{E}:= L^2(\mathbb{R}^d) \otimes \mathbb{C}^2$. As the mass term of the theory is proportional to the product $\varphi \tilde{\varphi}$, it is convenient to choose an infrared cut off of the form: 
\begin{equation}
\label{Eq: Qk MSR 1}
Q_k(\chi, \tilde{\chi}) = - \int_{\mcM} d\mu_x \,  q_k(x) T(\chi(x) \tilde{\chi}(x)) = - \frac{1}{2} \text{Tr}_{\mathcal{E}} [\mathfrak{q}_k T(\boldsymbol{\chi} \otimes \tilde{\boldsymbol{\chi}})],
\end{equation}
where $\mathfrak{q}_k := q_k \mathbb{I}_{\mathbb{C}^2}$ and the factor $\frac{1}{2}$ takes into account the commutativity of the fields -- see Remark \ref{Rmk: commutativity}. The analysis then closely mimics that of Section \ref{Sec: (A) Two mutually interacting scalar fields}, hence we refrain from repeating the discussion.  

\subsection{(B) Self-interacting Dirac fields}
\label{Sec: (B) Self-interacting Dirac fields}
The last scenario we examine is that of a self-interacting Dirac field on a $d-$dimensional globally hyperbolic spacetime $\mcM$, see Section \ref{Sec: Thirring model}. In this context, the r\^ole of $E$ is played by the vector bundle $\mathcal{X} := D\mcM \oplus D^*\mcM$, where $D\mcM$ and $D^*\mcM$ are respectively the Dirac bundle and its dual. Assuming triviality of the underlying spin bundle, $D\mcM \simeq \mcM \times \Sigma_d$ and $D^*\mcM \simeq \mcM \times \Sigma_d^*$, where $\Sigma_d := \mathbb{C}^{N_d}$ with $N_d := 2^{\lfloor \frac{d}{2} \rfloor}$ is the \emph{spinor space}. Also in this case we can introduce a field doublet 
\begin{equation}
    \label{Eq: Psi}
    \Psi := \begin{pmatrix}
        \psi^{\rho} \\
        \bar{\psi}_{\rho}
    \end{pmatrix} \in \Gamma^{\infty}(\mathcal{X}) \simeq C^{\infty}(\mcM; \Sigma_d) \otimes \mathbb{C}^2 \simeq C^{\infty}(\mcM) \otimes \mathbb{C}^{2N_d}, 
\end{equation}
where $\psi^{\rho} \in C^{\infty}(\mcM; \Sigma_d)$, $\bar{\psi}_{\rho} \in C^{\infty}(\mcM; \Sigma^*_d)$ and we exploited the fact that $\Sigma_d$ and $\Sigma_d^*$ are canonically isomorphic as complex vector spaces. A few remarks are in due order: 
\begin{itemize}
    \item[\ding{104}] Observe that Equation \eqref{Eq: Psi} involves objects with a different Dirac index structure. To avoid this potential hurdle, one may alternatively choose the field doublet in the \emph{Nambu–Gorkov} basis, \textit{i.e.},
\begin{equation*}
    \Psi_{NG} := \begin{pmatrix}
        \psi^{\rho} \\
        (\bar{\psi}^T)^{\rho}
    \end{pmatrix} \in \Gamma^{\infty}(D\mcM \oplus D \mcM) \simeq C^{\infty}(\mcM) \otimes \mathbb{C}^{2N_d}.
\end{equation*} 
The advantage of such a choice of basis lies in the fact that it allows us to write the relevant matrix-valued operators as
\begin{equation*}
    A = \begin{pmatrix}
        A_{11} &   A_{12} \\
        A_{21} &  A_{22}
   \end{pmatrix} \in \text{End}(\mathbb{C}^{2N_d}),
\end{equation*}
where each entry is an endomorphism of $\mathbb{C}^{N_d}$. However, as the Wetterich equation entails a trace over the Dirac indices, {\it cf.}, Equation \eqref{Eq: RG flow Thirring}, this does not present an obstruction, and we therefore retain our original choice.
\item[\ding{104}] Notice, in addition, that the field doublet mixes two types of indices: the $\mathbb{C}^2$-vector indices of the field doublet and the Dirac indices of its spinorial components, here denoted by Greek letters. The distinct nature of the indices does not constitute a problem for the ensuing discussion, yet it is convenient to keep in mind that they carry an intrinsically different meaning. 
\end{itemize}

\begin{notation}
\label{Not: Dirac matrices}
Let us denote by $\mathcal{E} := \Gamma^{\infty}(\mathcal{X})$ and let $A: \mathcal{E} \rightarrow \mathcal{E}$ be a matrix-valued operator which in block-form reads: 
\begin{equation*}
    A := \begin{pmatrix}
        A_{11} & A_{12} \\
        A_{21} & A_{22} 
    \end{pmatrix},
\end{equation*}
where $A_{11}:= (A_{11})^{\rho'}{}_{\rho}: \Gamma(D\mcM) \rightarrow \Gamma(D\mcM)$, $A_{12}:=(A_{12})^{\rho \rho'}: \Gamma(D^*\mcM) \rightarrow \Gamma(D\mcM)$, $A_{21}:= (A_{21})_{\rho \rho'}: \Gamma(D\mcM) \rightarrow \Gamma(D^*\mcM)$ and $A_{22}:= (A_{22})_{\rho'}{}^{\rho}: \Gamma(D^*\mcM) \rightarrow \Gamma(D^*\mcM)$. Hence, we can formally define its Hilbert space trace on $\mathcal{E}$ as:
\begin{equation}
    \label{Eq: trace Dirac}
    \text{Tr}_{\mathcal{E}} (A) := \int_{\mcM} d\mu_x \, \text{tr}_{\text{Dirac}}[\text{tr}_{\mathbb{C}^2} A](x,x) = \int_{\mcM} d\mu_x \, \sum_{\ell = 1,2} \sum_{\rho=1}^{N_d}  (A_{\ell \ell})^{\rho}{}_{\rho} (x,x).
\end{equation}
Notice that the Dirac trace $\text{tr}_{\text{Dirac}}$ can only be defined for endomorphism of $D\mcM$ or $D^*\mcM$. 
This suggests that, since the final result should be a scalar also in the Dirac indices, the off-diagonal elements should not play any r\^ole in the derivation of the Renormalization Group flow equations. 
\end{notation}

Henceforth, let us consider an infrared mass regulator of the form 
\begin{flalign}
    \label{Eq: infrared regulator Dirac}
    \notag
    Q_k (\eta, \bar{\eta}) &:= - \text{Tr}_{\mathcal{E}} \left[ \mathfrak{q}_k T(\boldsymbol{\eta} \otimes_{\mathbb{C}^2} \overline{\boldsymbol{\eta}}) \right]
    \\& = -\int_{\mathcal{M}} d\mu_x \, \left[T(\eta^{\rho} [q_{k,\psi}]_{\rho}{}^{\rho'} \bar{\eta}_{\rho'})(x)- T(\bar{\eta}_{\rho'} [q_{k, \bar{\psi}}]^{\rho'}{}_{\rho} \eta^{\rho}) (x)\right], 
\end{flalign}
where we denote by 
\begin{equation*}
    \mathsf{q}_k = \begin{pmatrix}
     q_{k, \psi}  & 0 \\ 0 & -q_{k, \bar{\psi}}
    \end{pmatrix}, \quad (\mathfrak{q}_k)_{11} := q_{k, \psi} \, \text{and} \, (\mathfrak{q}_k)_{22} := -q_{k, \bar{\psi}}, 
\end{equation*} 
while $\boldsymbol{\eta} \in \Gamma^{\infty}(\mathcal{X})$ is the field doublet as per Equation \eqref{Eq: Psi}. In order to further simplify the ensuing analysis, we shall assume that each non-vanishing entry of $\mathfrak{q}_k$ is proportional to the identity in the Dirac indices and it depends linearly on $k$, that is:  
    \begin{equation}
    \label{Eq: ass on qk}
        (q_{k, \psi})^{\rho}{}_{\rho'}(x) := \epsilon k f(x) \mathbb{I}^{\rho}{}_{\rho'}, \, \, \epsilon > 0
    \end{equation}
 and, analogously for $q_{k, \bar{\psi}}$, where $f$ is the same adiabatic cut off that will be chosen in the LPA. Without loss of generality, we shall set $\epsilon \equiv 1$ and assume that $k \in \mathbb{R}_+$. We highlight that any contribution to Equation \eqref{Eq: ass on qk} constant in $k$ can be incorporated into the free mass of the theory. 
Let 
\begin{flalign}
    \label{Eq: current Dirac}
    \notag
    J(\eta, \bar{\eta}) &:= \int_{\mcM} d\mu_x \, \overline{\boldsymbol{j}}^T (x) \boldsymbol{\eta}(x) \\ &= \int_{\mathcal{M}} d\mu_x \left[ \bar{j}_{\rho}(x) \eta^{\rho}(x) + j^{\rho}(x) \bar{\eta}_{\rho}(x)\right], \, \, \boldsymbol{j} = \begin{pmatrix}
        j \\
        \bar{j}
    \end{pmatrix} \in \Gamma^{\infty}_0(\mathcal{X})
\end{flalign}
be an external current and
\begin{equation}
\label{Eq: P0 Thirring}
    \mathrm{P}_0 := \begin{pmatrix}
        P_{0, \psi} & 0 \\
        0 & P_{0, \bar{\psi}}
    \end{pmatrix},
\end{equation}
be the operator implementing the free equations of motion, where $P_{0,\psi} \equiv (P_{0,\psi})^{\rho}{}_{\rho'} := \slashed{D}^{\rho}{}_{\rho'}$, $P_{0,\bar{\psi}} \equiv (P_{0,\bar{\psi}})_{\rho}{}^{\rho'} := (\slashed{D}^*)_{\rho}{}^{\rho'}$. Furthermore, we define by
\begin{equation}
   Z_k (j, \bar{j}) = \omega(S(V)^{-1} \star S(V+ Q_k + J)), \quad  W_k(j, \bar{j}) = - i \log Z_k(j, \bar{j}),
\end{equation}
the generating functional of the regularised theory and its connected counterpart, respectively.   
By direct computation the Hessian of $W_k(j, \bar{j})$ reads in a convenient choice of basis:
\begin{equation}
\label{Eq: Wk Thirring}
[W_k^{(2)}(j, \bar{j})](x,y) = i
    \begin{pmatrix}
       \langle \eta(x) \cdot_T \bar{\eta}(y) \rangle - \psi(x)  \bar{\psi}(y) &  \underbrace{\langle \eta(x) \eta(y) \rangle}_0 - \psi(x) \psi(y) \\  \underbrace{\langle \bar{\eta}(x) \bar{\eta}(y)}_0 \rangle - \bar{\psi}(x) \bar{\psi}(y) & \langle\bar{\eta}(x) \cdot_T \eta(y)\rangle - \bar{\psi}(x)  \psi(y)
    \end{pmatrix}, 
\end{equation}
where $\psi, \bar{\psi}$ are the classical fields, \textit{i.e.}, $\frac{\delta W_k}{\delta j(x)} = \langle \bar{\eta} (x) \rangle =: \bar{\psi}(x)$ and $\frac{\delta W_k}{\delta \bar{j}(x)} = \langle \eta (x)\rangle =: \psi(x)$, and, for the sake of notational ease, we have left the Dirac indices understood. Note that the off-diagonal correlations $\langle \eta (x) \eta(y) \rangle$ and $\langle \bar{\eta}(x) \bar{\eta}(y) \rangle$ must be identically zero since such terms break the Fermionic number symmetry. As a consequence, the \textbf{Polchinski equation} involves only the diagonal contributions in Equation \eqref{Eq: Wk Thirring} and it reads:
\begin{flalign*}
    \partial_k W_k(j, \bar{j}) &= \langle \partial_k Q_k (\eta, \bar{\eta}) \rangle \\ & = -\int_{\mathcal{M}} d\mu_x \, \left[T(\eta^{\rho} [\partial_k q_{k,\psi}]_{\rho}{}^{\rho'} \bar{\eta}_{\rho'})(x)- T(\bar{\eta}_{\rho'} [\partial_k q_{k, \bar{\psi}}]^{\rho'}{}_{\rho} \eta^{\rho}) (x)\right] 
    \\ &= - \text{Tr}_{\mathcal{E}} \left( (\partial_k \mathfrak{q}_{k})^{\rho'}{}_{\rho}  \left[-i W_k^{(2)} + \Psi \otimes_{\mathbb{C}^2} \bar{\Psi} \right]^{\rho}{}_{\rho'}\right),  
\end{flalign*}
where $\langle \cdot \rangle$ denotes the weighted expectation value. Observe that the above equation is only formal, as a regularization procedure needs to be implemented for the time ordered products to yield a finite result. Thereby, denoting by $G_{k, \psi} := \langle \eta \cdot_T \bar{\eta} \rangle$ and $G_{k, \bar{\psi}} := \langle \bar{\eta} \cdot_T \eta \rangle$, the \textbf{Renormalization Group (RG) flow equation} reads: 
\begin{flalign}
    \label{Eq: RG flow Thirring}
    \notag
    \partial_k \Gamma_k &= \langle \partial_k Q_k \rangle - \partial_k Q_k \\ \notag & = - \text{Tr}_{\mathcal{E}} \left(\partial_k \mathfrak{q}_k : [\Gamma_k^{(2)} - \mathfrak{q}_k]^{-1} :_{\tilde{\mathrm{H}}_F} \right) \\&= - \int_{\mathcal{M}} d\mu_x \,  \left[ \partial_k (q_{k, \psi})^{\rho'}{}_{\rho}(x) : (G_{k, \psi})^{\rho}{}_{\rho'}(x,x) :_{\tilde{H}_{F,\psi}} - \partial_k (q_{k, \bar{\psi}})^{\rho'}{}_{\rho}(x) : (G_{k, \bar{\psi}})^{\rho}{}_{\rho'} (x,x) :_{\tilde{H}_{F, \bar{\psi}}} \right], 
\end{flalign}
where $\Gamma_k (\psi, \bar{\psi}) := W_k(\bar{j}_{\psi}, j_{\bar{\psi}}) - J(\psi, \bar{\psi}) - Q_k(\psi, \bar{\psi})$ is the average effective action, whilst $\tilde{H}_{F, \psi / \bar{\psi}}$ are the counterterms arising from the point splitting procedure. The {\bf effective potential} $U_k$ is then implicitly defined via the following relation: 
\begin{equation*}
    \Gamma_k^{(2)} (\psi, \bar{\psi}) - \mathfrak{q}_k := \mathrm{P}_0 + U_k^{(2)}(\psi, \bar{\psi}) =: \mathrm{P}_k, 
\end{equation*}
where the Dirac indices are left understood for notational brevity. In terms of the $\mathbb{C}^2-$matrix components, we can write: 
\begin{equation*}
\begin{cases}
       [\Gamma_k^{(2)}]_{1 1} - q_{k, \psi}=: P_{0, \psi} + [U_k^{(2)}]_{1 1}, \quad
       [\Gamma_k^{(2)}]_{2 2} + q_{k, \bar{\psi}} =: P_{0, \bar{\psi}} + [U_k^{(2)}]_{2 2} \\
       [\Gamma_k^{(2)}]_{1 2} =: [U_k^{(2)}]_{1 2}  \quad   [\Gamma_k^{(2)}]_{2 1} =: [U_k^{(2)}]_{2 1}. 
\end{cases}
\end{equation*}

\begin{notation}
Not to excessively burden the notation, we will refrain from indicating explicitly the Dirac indices, whenever we feel that confusion may not arise.  
\end{notation} 

In the local potential approximation, see Section \ref{Sec: (A) Two mutually interacting scalar fields}, the quantum wave operator $\mathrm{P}_k$ is still Green hyperbolic and, thus, it admits unique advanced and retarded Green operators: 
\begin{flalign*}
  \Delta^U_{A/R}: & \Gamma^{\infty}_0 (\mathcal{X}) \longrightarrow \Gamma^{\infty} (\mathcal{X}) \\
    \mathrm{P}_k \circ \Delta^U_{A/R} = \Delta^U_{A/R} \circ \mathrm{P}_k &= \text{id} \vert_{\Gamma^{\infty}_0(\mathcal{X})}, \quad \text{supp} \Delta^U_{A/R}(\Psi) \subseteq J^{\mp} (\text{supp}(\Psi)),
\end{flalign*}
for all $\Psi \in \Gamma^{\infty}_0(\mathcal{X})$. With a slight abuse of notation, we will still indicate by $\Delta^U_{A/R}$ the associated matrix-valued bi-distributions in $\mathcal{D}'(\mcM \times \mcM; \mathcal{X} \boxtimes \mathcal{X})$, \textit{i.e.}, 
\begin{equation}
\label{Eq: DeltaU Thirring}
   \Delta^U_{A/R} = \begin{pmatrix}
        [\Delta^U_{A/R}]_{1 1} &  [\Delta^U_{A/R}]_{1 2} \\
          [\Delta^U_{A/R}]_{2 1}  &  [\Delta^U_{A/R}]_{2 2} 
    \end{pmatrix}, 
\end{equation}
where each matrix entry is a bi-distribution lying in $\mathcal{D}'(\mcM \times \mcM; D\mcM \oplus D^* \mcM)$. By means of the advanced and retarded quantum M\o ller operators, we can rewrite Equation \eqref{Eq: RG flow Thirring} in terms of $U_k$ as follows: 
\begin{flalign}
    \label{Eq: Wetterich for Uk Thirring}
    \notag
    \partial_k U_k &= - \text{Tr}_{\mathcal{E}} \left( \partial_k \mathfrak{q}_k [(\mathbb{I} - \Delta^U_R U_k^{(2)}) w (\mathbb{I} - U_k^{(2)}\Delta^U_A) ] \right) \\ & -  \int_{\mcM} d\mu_x \, \partial_k q_{k, \psi}  [(\mathbb{I} - \Delta^U_R U_k^{(2)}) w (\mathbb{I} - U_k^{(2)}\Delta^U_A)]_{1 1} (x,x) - \partial_k q_{k, \bar{\psi}} [(\mathbb{I} - \Delta^U_R U_k^{(2)}) w (\mathbb{I} - U_k^{(2)}\Delta^U_A)]_{2 2}, 
\end{flalign}
where $w \in C^{\infty}(\mcM \times \mcM; \mathcal{X} \boxtimes \mathcal{X})$ reads:
\begin{equation*}
    w = \begin{pmatrix}
        w_{\psi} & 0 \\
        0 & w_{\bar{\psi}},
    \end{pmatrix}
\end{equation*}
where $w_{\psi}(x,y) := \Delta_{F, \psi}(x,y) - H_{F, \psi}(x,y)$ and $w_{\bar{\psi}}(x,y) := \Delta_{F, \bar{\psi}}(x,y) - H_{F, \bar{\psi}}(x,y)$. 

Under the same assumptions on the effective potential detailed in Section \ref{Sec: (A) Two mutually interacting scalar fields}, Equation \eqref{Eq: Wetterich for Uk Thirring} can be equivalently formulated as a mixed Cauchy problem for its integral kernel $u(\Psi, k)$: 
\begin{equation}
\label{Eq: BV problem Dirac}
    \begin{cases}
        \partial_k u = \text{Tr}_{\mathcal{E}} (\mathrm{G}_{k}(u^{(2)})) \\
        u(\Psi, a) = \alpha \\
        u \vert_{\partial X \times [a,b]} = \beta
    \end{cases}
\end{equation}
where $X \subseteq \mathbb{R}^2$ is a compact set encompassing all possible values of the field doublet $\Psi$, $k \in [a, b] \subseteq \mathbb{R}_+$, whilst $\alpha \in C^{\infty}(X)$ and $\beta \in C^{\infty}(\partial X \times [a,b])$ are assigned functions. Here $u^{(2)}$ is the Hessian matrix of $u$, which in a convenient choice of basis reads: 
\begin{equation*}
    u^{(2)} (\Psi, k) := \left(\frac{\partial}{\partial \Psi} \otimes_{\mathbb{C}^2} \frac{\partial}{\partial \bar{\Psi}}\right) u (\Psi, k) = \begin{pmatrix}
        \partial_{\psi} \partial_{\bar{\psi}} u & \partial_{\psi} \partial_{\psi} u \\   \partial_{\bar{\psi}} \partial_{\bar{\psi}} u &  \partial_{\bar{\psi}} \partial_{{\psi}} u
    \end{pmatrix}.
\end{equation*}
In the above equation, the functional $G_{k}(u^{(2)}) := \sum_{\ell = 1,2} G_{k, \ell}(u^{(2)})$, is defined as follows: 
\begin{equation}
\label{Eq: Gkl Dirac}
    G_{k, \ell} (u^{(2)}) := - \frac{1}{ ||f||_{L^1}} \int_{\mathcal{M}} d\mu_x \, \left((\partial_k \mathfrak{q}_{k})_{\ell \ell}(x) \left[ (\mathbb{I}- u^{(2)} \Delta^u_{R} f) \otimes (\mathbb{I}- u^{(2)} \Delta^u_{R} f) (w) \right]_{\ell \ell}\right)^{\rho}{}_{\rho} (x,x), 
\end{equation}
where $\Delta^u_{R} \in \mathcal{D}'(\mcM \times \mcM; \mathcal{X} \boxtimes \mathcal{X})$ is the retarded fundamental solution of the Green hyperbolic operator $\mathrm{P}_{0} + f u^{(2)}$. Barring the occurrence of Dirac indices, the analysis is then identical to the one in Section \ref{Sec: (A) Two mutually interacting scalar fields}, hence we omit it. 

\begin{remark}
All the operational seminorms should be taken both with respect to the Dirac indices and to the fictitious $\mathbb{C}^2$ indices corresponding to the components of the field doublet. If we assume the underlying spin bundle to be trivial, this poses no further complications in the analysis, as both the Dirac indices in $\mathbb{C}^{N_d} \otimes \mathbb{C}^2$ and the $\mathbb{C}^2-$matrix ones can be treated on equal footing. As a consequence, all the estimates derived in Section \ref{Sec: (A) Two mutually interacting scalar fields} for the two scalar fields are also true in the present setting. For instance, taking $u \in C^{\infty}(X \times [a,b])$, the associated Hessian matrix $u^{(2)} \in \text{End}(\mathbb{C}^{N_d} \otimes \mathbb{C}^2)$. Therefore, we can set 
\begin{equation*}
    ||u^{(2)}||_0 := \sup_{\Psi, k} |u^{(2)}(\Psi, k)|,
\end{equation*}
where $| \cdot |$ denotes the matrix norm on $\text{End}(\mathbb{C}^{N_d} \otimes \mathbb{C}^2)$. 
\end{remark}

As a consequence, also in this case, we can establish the following result, whose proof is omitted, being identical to the one discussed in the scalar scenario. 

\begin{theorem}
There exists a unique family of tame smooth local inverses to the RG operator for the Thirring model and, hence, unique local solutions exist. 
\end{theorem}

\section*{Acknowledgements}
\addcontentsline{toc}{section}{Acknowledgements}

BC is grateful to Claudio Dappiaggi, Nicolò Drago and Nicola Pinamonti for enlightening discussions on an early version of the project. The work of BC has been supported by a fellowship of the University of Pavia and BC also acknowledges the support of the INFN --Sezione di Pavia and of Gruppo Nazionale di Fisica Matematica, part of INdAM. The work of BC is in part supported by "Progetto Giovani GNFM 2025" under the project "Hadamard states for linearized Yang-Mills theories" fostered by Gruppo Nazionale di Fisica Matematica -- INdAM. 





\appendix

\section{Tame smooth maps}\label{Sec: Appendix A}
In this Appendix, we recall some relevant notions, useful to establish local existence of solutions to the RG flow equations via the Nash-Moser theorem -- see Sections \ref{Sec: (A) Two mutually interacting scalar fields}, \ref{Sec: MSR fields} and \ref{Sec: (B) Self-interacting Dirac fields}. We refer the reader to \cite{Moser1966a, Moser1966b, Nash1956} \cite[Part II]{Hamilton_1982} for a more comprehensive account of these topics.

Let $F$ be a vector space. We call \emph{seminorm} on $F$ a function $||\cdot||: F \to \mathbb{R}$ which satisfies the following requirements: for all $f,g \in F$ and for all $c \in \mathbb{R}$,  
\begin{itemize}
    \item[1.] $||f|| \ge 0$,
    \item[2.] $||f+g|| \le ||f|| + ||g||$,
    \item[3.] $||cf|| = |c| ||f||$.
\end{itemize}
Let $I$ be an index set, possibly countable. The family of seminorms $\{||\cdot ||_i\}_{i \in I}$ uniquely a defines a topology: this is coarsest topology on $F$ such that all seminorms are continuous, with a base of neighbourhoods of $0$ given by finite intersections of sets of the form $\{ f \in F \, | \, ||f||_i < \epsilon\}$. More precisely, given a sequence of elements in $F$, $\{f_j\}_{j \in \mathbb{N}} \subset F$, and $f \in F$ we say that $f_j$ converges to $f$ if $||f_j - f||_i \rightarrow 0$, for all $i \in I$. If $F$ is endowed with such a collection of seminorm, we call $F$ a \emph{locally convex topological vector space}. The topology on $F$ is said to be: 
\begin{itemize}
    \item[\ding{104}] \emph{Hausdorff}, if $||f||_n = 0$ for all $n \in \mathbb{N}$ implies that $f=0$, 
    \item[\ding{104}] \emph{metrizable}, if the family of seminorms $\{||\cdot||_n\}_{n \in \mathbb{N}}$ is countable.  
\end{itemize}
In addition, $F$ is \emph{complete} is every Cauchy sequence in $F$ is convergent, in the sense clarified above. 

\begin{definition}
\label{Def: Frechet space}
    A topological vector space $F$ that is locally convex, Hausdorff, metrizable and complete is called \emph{Fréchet space}. If the family of seminorms $\{||\cdot||_n\}_{n \in \mathbb{N}}$ is increasing in strength, \textit{i.e.}, $||f||_n \le ||f||_{n+1}$, for any $n \in \mathbb{N}$, $F$ is said to be a \emph{graded Fréchet space}. 
\end{definition}

\noindent Let $F$ be a graded Fréchet space as per Definition \ref{Def: Frechet space} and let $B$ be some Banach space. Denote by $$\mathcal{S}_{\exp}(\mathbb{N},B) := \{ (x_j)_{j \in \mathbb{N}} \subset B \, | \, \exists \alpha > 0 \, \text{such that} \, \sup_{j} e^{\alpha j} ||x_j||_B < \infty \},$$ the space of exponentially decreasing sequences on $B$. $F$ is said to be \emph{tame} if there exist two linear maps $L: F \to \mathcal{S}_{\exp}(\mathbb{N}, B)$ and $M: \mathcal{S}_{\exp}(\mathbb{N}, B) \to F$ such that $ML: F \to F$ is the identity, namely, 
\begin{equation}
\label{Eq: Tame space}
    F \mathrel{\mathop{\longrightarrow}^{L}} \mathcal{S}_{\exp}(\mathbb{N}, B) \mathrel{\mathop{\longrightarrow}^{M}} F.
\end{equation} 
Therefore, we can formulate the following definition. 

\begin{definition}
\label{Def: tame map}
    Consider two graded Fréchet spaces $F$ and $G$ and let $\mathcal{U} \subset F$ be an open subset of $F$. Suppose, additionally, that we have a map $P: \mathcal{U} \rightarrow G$. Then, $P$ is \emph{tame} of degree $r$ and base $b$ if it is continuous and if, for all $f$ in a neighbourhood of $f_0 \in \mathcal{U}$ and for all $n \ge b$, there exists a constant $C \equiv C(n) > 0$ such that
\begin{equation}
    ||P(f)||_n \le C( 1 + ||f||_{n+r}).
\end{equation}
\end{definition}

\end{document}